\RequirePackage{fix-cm}
\documentclass[twocolumn]{svjour3}          
\smartqed  
\usepackage{graphicx}

\usepackage{amsmath,amssymb,amsfonts}
\usepackage{textcomp}
\usepackage[linesnumbered,ruled,vlined]{algorithm2e}
\usepackage{url} 
\usepackage{todonotes}
\usepackage{cellspace}
\usepackage{float}
\usepackage{subfigure}
\usepackage{multirow}
\usepackage{caption}
\usepackage{tikz}
\usepackage{enumitem}
\usepackage{xcolor}
\usepackage{hyperref}
\usepackage{overpic}
\usepackage{tikz}
\usepackage{balance}

\usepackage{siunitx}

\usepackage{microtype}
\usepackage{booktabs}
\usepackage{makecell}

\usepackage{todonotes}

\begin{document}

\title{VTRQ: Enabling Verifiable Trajectory Range Queries in Hybrid-Storage Blockchains
}


\author{Zhongming Yao         \and
       Junchang Xin \and
       Yumeng Song \and
       Yusen Mao \and
       Kristian Torp \and
       Yuemin Ding \and
       Divesh Srivastava \and
       Yushuai Li \and
       Christian S. Jensen \and
       Tianyi Li
}


\institute{Zhongming Yao \at
              Aalborg University, Denmark \\
               Northeastern University, China\\
              \email{zyao@cs.aau.dk}           
           \and
           Junchang Xin \at
              Northeastern University, China\\
               \email{xinjunchang@mail.neu.edu.cn}
            \and
           Yumeng Song \at
              Aalborg University, Denmark \\
               \email{yumengs@cs.aau.dk}
            \and
           Yusen Mao \at
              Northeastern University, China\\   
              \email{maoyusen@stumail.neu.edu.cn}
            \and
           Kristian Torp \at
              Aalborg University, Denmark \\
               \email{torp@cs.aau.dk}
            \and
            Yuemin Ding \at
              University of Navarra, Spain\\
               \email{yueminding@tecnun.es}
            \and           
           Divesh Srivastava \at
              AT\&T Chief Data Office, USA\\
               \email{divesh@research.att.com}
            \and
           Yushuai Li \at
              Aalborg University, Denmark \\
               \email{yusli@cs.aau.dk}
            \and
           Christian S. Jensen \at
              Aalborg University, Denmark \\
               \email{csj@cs.aau.dk}
            \and
           Tianyi Li \at
              Aalborg University, Denmark \\
               \email{tianyi@cs.aau.dk}
}

\date{Received: date / Accepted: date}

\maketitle

\begin{abstract}
Due to their increasingly large volumes, outsourcing of trajectory storage and querying to third-party service providers has become attractive. However, in such outsourced environments, service providers may return incorrect, e.g., incomplete, tampered, or invalid query results, making verifiability of query results an important consideration. Existing hybrid-storage blockchains offer limited support for trajectory data, lacking authenticated data structures (ADS) that enable efficient verification.
For example,  ADSs designed for queries on one-dimensional data are unsuitable for queries on multidimensional trajectory data, while ADSs tailored for discrete data may yield incomplete results when applied to continuous trajectory data.
We propose the first framework for \textbf{v}erifiable \textbf{t}rajectory \textbf{r}ange \textbf{q}ueries in hybrid-storage blockchains, called \texttt{VTRQ}. It features two efficient ADSs: (i) a spatial ADS for road networks that leverages hierarchical organization to aggregate trajectory, edge, and node hashes, thus reducing redundant computations and improving spatial verification efficiency; and (ii) a temporal ADS based on interval trees, which indexes only the start and end times of trajectories, thereby enabling pruning and efficient temporal verification. By separating spatial and temporal indexing, the method reduces the need for data comparison, enhancing both query and verification efficiency.
{To aggregate spatial and temporal query results, \texttt{VTRQ} provides a spatio-temporal edge aggregation mechanism that combines temporal verification of spatial nodes, spatial intersection computation, and temporal intersection analysis to achieve spatio-temporal filtering.
}
Experiments on two real-world datasets show that \texttt{VTRQ} is capable of up to 6× faster query performance and up to an order-of-magnitude improvement in verification efficiency,  compared to the state-of-the-art method.

\keywords{Trajectory \and Verifiable Queries \and Range Queries \and Blockchain}
\end{abstract}

\vspace{4mm}

\section{Introduction}\label{sec:introduction}

\sloppy 

\setlength{\marginparwidth}{1cm}



With the growing use of trajectory data in areas such as traffic management~\cite{wu2024fedapt,li2022evolutionary}, logistics optimization~\cite{kong2022trajectory}, and personalized services~\cite{shang2014personalized}, efficient management and analysis of such data have become increasingly important. Trajectory datasets can be very large. For example, Didi Chuxing’s ride-hailing platform~\cite{didiGlobalWebsite} generates some 15 billion trajectory points per day across all locations, and in Chengdu alone, it generates some 3.6 million trajectory points per day. We use a Chengdu dataset to evaluate the proposed framework empirically. Due to the very large data volume, storing and querying the data requires substantial resources, and it is generally attractive for data owners to be able to outsource these tasks to cloud service providers. In practice, such outsourcing is already widespread. For example, Didi Chuxing outsources the handling of its trajectory data to Alibaba Cloud~\cite{aliyunStartup1077761}, Geotab stores daily fleet trajectories in Google Cloud’s BigQuery~\cite{googleCloudGeotabCaseStudy}, and the BMW Group manages its vehicle telemetry and trajectory data on Amazon~Web Services (AWS)~\cite{awsBmwGroupCaseStudy}. However, service providers may not return correct query results for reasons such as misconfiguration, cyber attacks, program vulnerabilities, and resource constraints~\cite{xia2022litmus,zhou2021veridb,yu2021secure,wu2023enabling,zhang2019gem,zhang2021authenticated,li2024authenticated1,li2024authenticated2}.
For example, 19+ million California voter registration records stored in a misconfigured AWS database were deleted by cybercriminals~\cite{msspalertAwsVoterLeak}.
Therefore, the ability to verify efficiently the completeness and soundness of query results obtained from external service providers is important.

In an outsourcing environment, a blockchain can serve as a decentralized trust anchor for verifying query results. Specifically, efficient verification of a query result requires both verification information corresponding to the query results and a reference fingerprint that the service provider cannot forge. By computing a digest from the query result and the accompanying verification information and then comparing it with the reference fingerprint, the query result can be validated~\cite{mykletun2006authentication,pang2005verifying,li2010authenticated}. However, if the reference fingerprint is stored and maintained by the service provider, the verification loses its independence: the service provider may fabricate the query results, verification information, and matching fingerprint in a mutually consistent manner, reducing verification to self-assertion. Requiring the data owner to  remain online to distribute or sign fingerprints for every query is also impractical, as it incurs excessive operational and computational overheads in large-scale deployments~\cite{pang2005verifying,mykletun2006authentication}. Being an append-only and tamper-evident distributed ledger, a blockchain provides a publicly verifiable and immutable platform for storing reference fingerprints, thereby enabling trustworthy validation of query results.

Proposals~\cite{wu2023enabling,zhang2019gem,zhang2021authenticated,li2024authenticated1,li2024authenticated2} exist that leverage blockchain technology to support third-party service provisioning.
They use the security, transparency, and immutability of blockchains to provide verification mechanisms for query results.
The core idea is that a data owner sends raw data and an authenticated data structure (ADS) to a service provider, while uploading a digest of the ADS to a blockchain.
When computing a query, the service provider must return both the query result and a verification object (VO) using the ADS. The client compares the hash value of the ADS root, computed from the query result and the VO, and then verifies it with the digest stored on the blockchain.



A trajectory captures continuous movement as a sequence of timestamped GPS points with longitude and latitude~\cite{li2020compression,li2021trace,hamdi2022spatiotemporal,yu2023continuous,hu2023spatio}. Therefore, verifying spatio-temporal query results—such as those of range queries, which are common and practical—must consider not only spatial and temporal constraints but also the continuity of the underlying movement. However, existing blockchain-based data outsourcing proposals~\cite{wu2023enabling,zhang2019gem,zhang2021authenticated,li2024authenticated1,li2024authenticated2}, employ ADSs designed for point-based queries (e.g., supporting point or range conditions), and query results are returned as discrete points. 
They fail to consider the continuity of the underlying movement.
For example, consider the spatio-temporal range query $ q = (([4, 7], $ $[3, 6]), [4, 5]) $ in Fig.~\ref{Fig:introduction}. Here, the spatial component $([4, 7], [3, 6])$ defines a rectangle, with $[4, 7]$ and $[3, 6]$ specifying longitude and latitude ranges. The temporal component, $[4, 5]$, is a time interval.
Point-based processing—when trajectory continuity is not taken into account—can fail to identify valid results. 
For example, assuming linear movement between consecutive points in a trajectory,  then trajectory $\textit{Tr}_4$ traverses the query region, but none of its sample points intersect the spatial query region. Similarly, trajectory $\textit{Tr}_1$ satisfies the temporal constraint, although none of its points do. This highlights a key limitation of traditional point-based methods: due to the lack of continuity-aware reasoning, they can miss trajectories that actually satisfy the spatio-temporal constraints.
We focus on designing a framework tailored for trajectory range queries that incorporates ADSs along with verification mechanisms to enable verifiable trajectory range queries in blockchain-based data outsourcing architectures.
This involves two key challenges.

\begin{figure*}[]
\centerline{\includegraphics[width=0.75\textwidth]{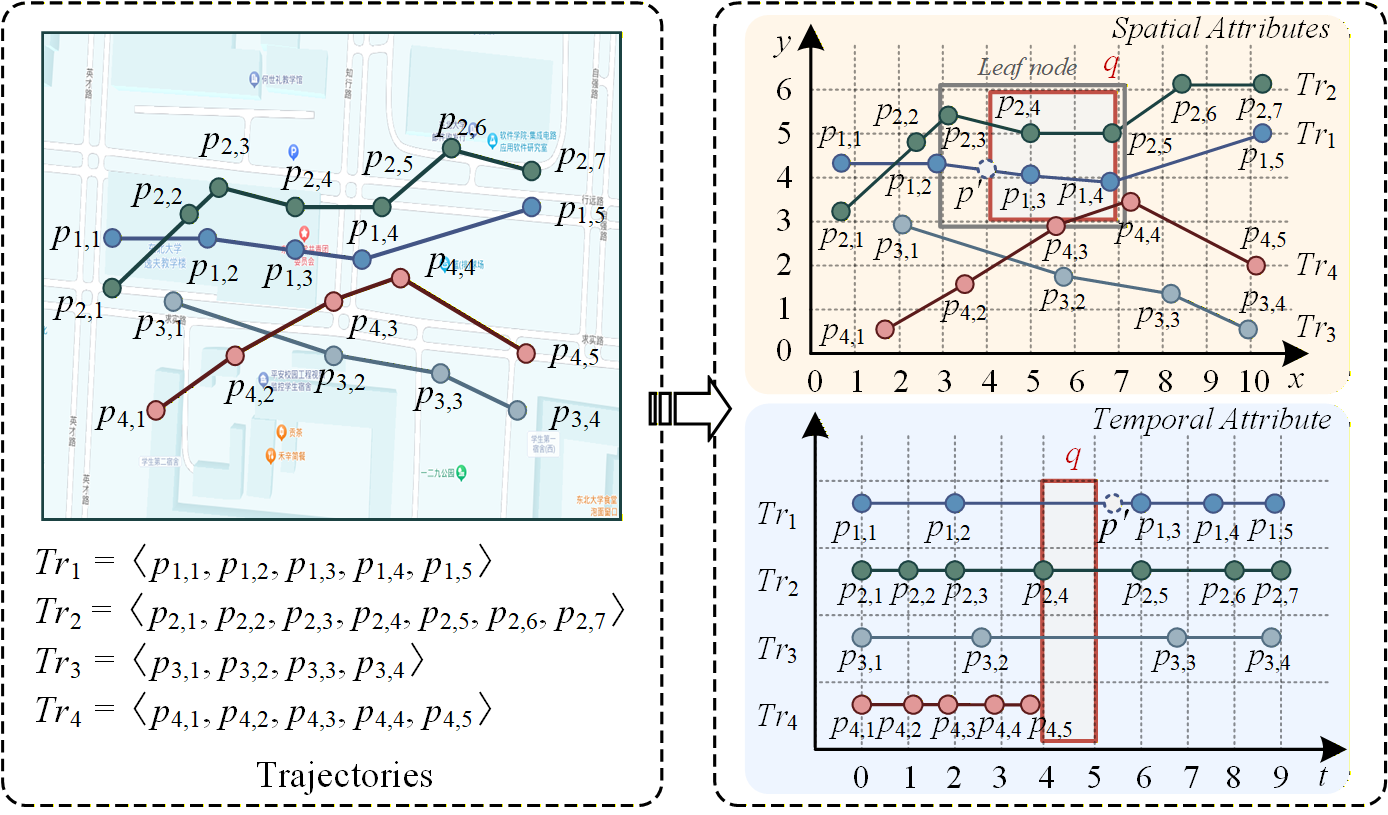}}
\caption{Example of a trajectory range query.}
\label{Fig:introduction}
\end{figure*}

\textit{Challenge I: How to design an ADS to support verifiable trajectory range queries?}
Unlike traditional~indexes, ADSs require extra information for result verification. It involves embedding cryptographic hashes into the index. For trajectory data with massive numbers of~points,~this causes high storage and computation overheads. 
In existing designs~\cite{zhang2019gem,zhang2021authenticated,xu2019vchain,ruan2021lineagechain,zhu2019sebdb,pei2020efficient,jia2025mest,cui2025towards,sun20253fs}, these problems are amplified on trajectory data, since many trajectory points generate excessive verification data.
This reduces construction and verification efficiency, and under memory constraints, it may even render ADSs impractical to construct or use.
Furthermore, mainstream spatio-temporal indexes, such as those based on minimum bounding rectangles (MBRs)~\cite{pfoser2000novel,tao2001mv3r} and grid partitioning~\cite{chakka2003indexing,yang2017novel}, rely on hierarchical partitioning, with each leaf node being responsible for a region and storing all trajectories intersecting this region. 
When hashes are embedded into such indexes to construct ADSs, the hash digest of a leaf node typically commits to the entire set of trajectories stored in that node. 
As a result, once a query or update accesses a leaf node, all trajectories contained in that node must participate in hash recomputation, even if only a small number of trajectories actually satisfy the query, leading to reduced~efficiency.


\textit{Challenge II: How to support efficient verifiable trajectory range queries based on a given ADS?}
Beyond the structural issues in ADS design, verifiable query efficiency remains a key challenge. Query processing and verification must guarantee result completeness without relying simply on point-by-point checks, which are prohibitive at scale.
This inefficiency stems from the fact that, in order to handle spatio-temporal information, spatio-temporal indexes~\cite{pfoser2000novel,tao2001mv3r,chakka2003indexing,yang2017novel} store the complete spatial and temporal details of trajectories.
Thus, after spatial filtering, every point and edge of a candidate trajectory must be checked for temporal validity. For example, given the query $ q = (([4, 7], [3, 6]), [4, 5]) $ in Fig.~\ref{Fig:introduction}, trajectories $\textit{Tr}_1$, $\textit{Tr}_2$, and $\textit{Tr}_4$ satisfy the spatial constraint. Existing methods must check all edges of these trajectories—including those of $\textit{Tr}_4$, which does not overlap with the query's temporal interval. Such repeated checks lead to redundant computations, severely reducing query and verification efficiency in large datasets.
Moreover, embedding hashes into the indexes exacerbates the problem, since query processing and verification involve processing large authentication data, increasing computational overheads.

We introduce a novel framework for \underline{\textbf{v}}erifiable \underline{\textbf{t}}rajectory \underline{\textbf{r}}ange \underline{\textbf{q}}ueries, called \texttt{VTRQ}.
We design two ADSs: a road-oriented spatial ADS (S-ADS) and an interval-tree-based temporal ADS (T-ADS).
To address Challenge I, we first compute a hash value for each trajectory and use it as a unique identifier for the trajectory. In S-ADS, these hash values are further organized into three hierarchical levels: trajectory level, edge level, and tree node level. Trajectory hashes are aggregated at the edge level, and edge hashes are further reused and aggregated at the node level. This design avoids redundant hash computations and repeated concatenations involving non-result trajectories, thereby improving both the construction of S-ADS and verification efficiency. In T-ADS, node hashes are computed by directly aggregating trajectory hashes and metadata used for verification and query acceleration, enabling optimized verification performance.

To address Challenge II, we decouple the spatial and temporal attributes of trajectories and use S-ADS and T-ADS independently.
Further decoupling spatial attributes (e.g., longitude and latitude) is not beneficial, as they form inherently coupled two-dimensional spatial constraints in trajectory queries; separating them would not improve pruning power but would introduce additional overheads for result merging and verification.
S-ADS stores only spatial information used for querying trajectories, while T-ADS stores only the start and end times of trajectories. This separation reduces the amount of data processed, thereby accelerating both query processing and result verification. S-ADS constructs MBRs on boundary edges to enable spatial pre-filtering, while T-ADS stores the maximum timestamp in each subtree to support temporal pruning. These optimizations further accelerate query processing and result verification.
To ensure accurate result aggregation after the decoupling, we propose a Spatio-Temporal Edge Aggregation (STEA) mechanism to achieve fine-grained alignment between spatial and temporal attributes in verifiable trajectory range queries. 
This mechanism first intersects results obtained from S-ADS and T-ADS to form candidate results. Then, fine-grained filtering is applied, combining temporal verification of spatial nodes, spatial intersection computation, and temporal intersection analysis to accurately identify matching trajectory edges.
The STEA mechanism can handle cases such as the one in Fig.~\ref{Fig:introduction}, where the sub-trajectory $\overline{p_{1,2}p_{1,3}}$ of trajectory $\textit{Tr}_1$ satisfies the spatial and temporal ranges of query $q$. While the portion $\overline{p' p_{1,3}}$ is in the spatial range, its temporal interval begins after $5$ and does not overlap with the query interval $[4,5]$, so $\textit{Tr}_1$ is not in the result of the query $q$. 
The main contributions of the paper are summarized as follows:

\begin{itemize}[leftmargin=*]
\item We present \texttt{VTRQ}, a framework for verifiable trajectory range queries in hybrid-storage blockchains. To the best of our knowledge, this is the first work that uses ADSs to enable verifiable trajectory range queries.

\item \texttt{VTRQ} implements two ADSs—a road-aware spatial ADS and an interval tree-based temporal ADS—to capture essential semantics and support efficient verifiable trajectory range queries.

\item  \texttt{VTRQ} introduces a spatio-temporal edge aggregation method to combine intermediate results from the ADSs and extract correct trajectory edges.


\item The experimental study on three real-world datasets demonstrates that \texttt{VTRQ} is capable of up to 6× higher query
efficiency and up to an order of magnitude higher
verification efficiency compared to the state-of-the-art method.

\end{itemize}

We review related work in Section~\ref{sec:related_work} and cover preliminaries in Section~\ref{sec:preliminaries}. 
Section~\ref{sec:ADS} introduces the proposed ADSs, while Section~\ref{sec:query_verification} presents the accompanying query processing and verification methods. 
Section~\ref{sec:verifiability} presents the security analysis. 
Section~\ref{sec:experiments} covers the experimental study, and Section~\ref{sec:conclusion} offers conclusions.

\section{Related work}\label{sec:related_work}

\subsection{Embedded ADS-based methods}

Some studies~\cite{xu2019vchain,ruan2021lineagechain,zhu2019sebdb,pei2020efficient,jia2025mest,bitcoin,wood2014ethereum,wang2022vchain+,singh2023efficient,hao2023efficient,yao2023efficient} build and maintain an ADS in the blockchain’s block. In these methods, an ADS is embedded directly into the blockchain. Users can verify the completeness and soundness of the data by using ADS digests embedded in block headers.
Initially, Bitcoin~\cite{bitcoin} used a Merkle tree as its ADS, primarily to organize and verify transactions within each block, ensuring data integrity and security. Building on this concept, Ethereum~\cite{wood2014ethereum} employs a more sophisticated structure known as the Merkle Patricia Trie, which not only verifies transactions but also stores the entire account state, enabling more efficient state management and faster data access.
Xu et al.~\cite{xu2019vchain} propose an accumulator-based ADS to support verifiable Boolean range queries.
Similarly, Wang et al.~\cite{wang2022vchain+} introduce a sliding window accumulator ADS to optimize verifiable range queries.
Ruan et al.~\cite{ruan2021lineagechain} introduce LineageChain, which employs a Merkle graph as an ADS to capture provenance and features a skip list index optimized for verifiable provenance queries.
Singh et al.~\cite{singh2023efficient} extend this line of work by proposing two ADSs: a two-tier deterministic append-only skip list and a partitioned \textit{B+} tree.
Zhu et al.~\cite{zhu2019sebdb} utilize the Merkle B-tree as an ADS to ensure the soundness and completeness of query results accessed through the index.
Pei et al.~\cite{pei2020efficient} propose a Merkle semantic trie-based ADS to verify queries.
Jia et al.~\cite{jia2025mest} propose the Merkle Extendible Hash Table, an authenticated secondary index for verifying queries on non-primary key attributes.

Despite their security guarantees, the embedded ADS-based methods have several notable drawbacks. First, storing raw data directly on-chain increases the storage burden of a blockchain markedly, resulting in poor scalability and high maintenance costs. Second, an ADS must be constructed for every block, which introduces additional computational overhead during data ingestion. Most importantly, these methods are unsuitable for outsourced storage scenarios—which is our focus—where the data is already stored on external storage nodes.

\subsection{External ADS-based methods}

\begin{figure*}[]
\centerline{\includegraphics[width=0.61\textwidth]{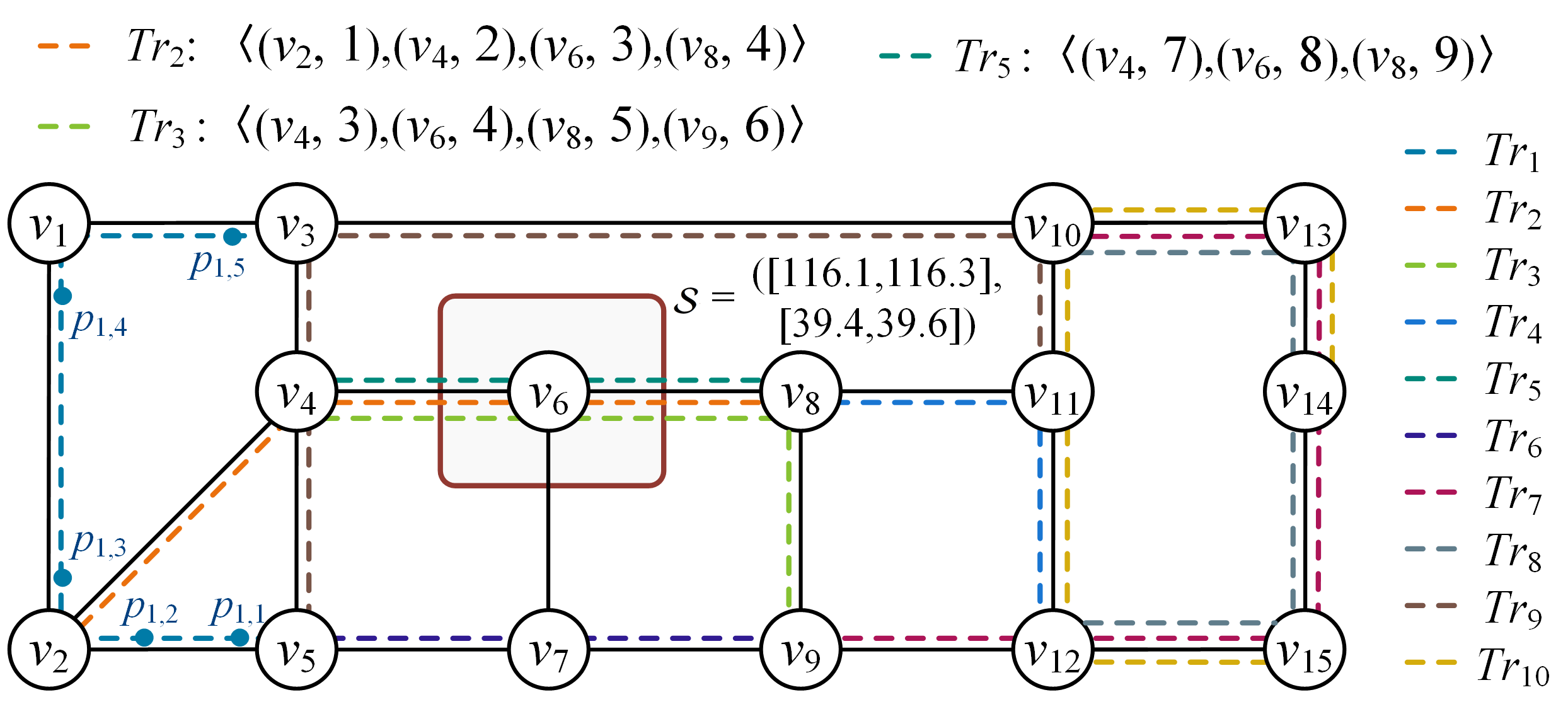}}
\caption{Example trajectories in a road network.}
\label{Fig:definition}
\end{figure*}
To enable verifiable queries when data is stored on external high-performance servers, some studies~\cite{wu2023enabling,zhang2019gem,zhang2021authenticated,li2024authenticated1,li2024authenticated2,cui2025towards,sun20253fs,zhou2023veridkg,hao2022efficient,yao2023learned} also store the ADS at the external service provider instead of on the blockchain.
{By recording the cryptographic digest of the ADS on the blockchain, users can verify the completeness and soundness of query results provided by an external service provider.}
Wu et al.~\cite{wu2023enabling} investigate authenticated queries over graph data in a blockchain-assisted cloud setting, proposing the PAGB authenticated data structure to support privacy-preserving and verifiable graph queries.
Zhang et al.~\cite{zhang2019gem} propose the $\textit{GEM}^2$-tree to minimize gas consumption, where gas refers to the unit of computational cost in blockchain systems, without impacting query performance. Further, Zhang et al.~\cite{zhang2021authenticated} introduce the Chameleon inverted index to achieve constant maintenance costs for keyword queries. 
Li et al.~\cite{li2024authenticated1} propose the Merkle Path DAG for graph data to efficiently handle keyword queries.
Li et al.~\cite{li2024authenticated2} design the MELTree and store the root digests on-chain, along with a verification object construction algorithm to verify the completeness and soundness of query results.
Cui et al.~\cite{cui2025towards} propose a virtual keyword forest that enables secure and reliably verifiable query processing in hybrid-storage blockchains.
Sun et al.~\cite{sun20253fs} propose $E^3FS$ for dynamically updatable datasets, which supports multi-keyword fuzzy search, an important search functionality.
Zhou et al.~\cite{zhou2023veridkg} propose VeriDKG, a verifiable SPARQL query engine for decentralized knowledge graphs that employs the RGB-Trie, a blockchain-maintained authenticated data structure, to provide correctness proofs for query results.

The external ADS-based methods offload raw data to off-chain storage service providers, while storing only the digest of the corresponding ADS on-chain to ensure data completeness and soundness. This design reduces the on-chain storage overhead and improves system scalability, making it particularly attractive for outsourcing environments with large-scale data.
While prior studies have explored verifiable queries in outsourced settings, they fall short of providing efficient support for correct trajectory range querying and query verification.
To the best of our knowledge, we are the first to propose efficient and verifiable trajectory range querying.

\section{Preliminaries}\label{sec:preliminaries}

\subsection{Trajectory range queries}

A raw trajectory is a series of raw GPS points $p=((x,y),t)$, where $x$ is a longitude, $y$ is a latitude, and $t$ is a timestamp~\cite{li2022evolutionary,yao2024camel,hu2024estimator}.
Fig.~\ref{Fig:definition}  shows example raw trajectory $\langle p_{1,1}, \ldots, p_{1,5} \rangle$.

\begin{definition}
A \textbf{road network} is modeled as an undirected graph $G = (V, E)$. Here,
$V = \{v_1, v_2, \ldots, v_{p}\}$ is a set of vertices representing road intersections. Each $v_i=(x,y)$ has a longitude $x$ and a latitude $y$. 
Next, $E$ is a set of edges, each of which is represented by two vertices from $V$. Edges model road segments in-between two intersections.
\end{definition}

Fig.~\ref{Fig:definition} shows a road network with 15 intersections and 21 segments. For example, vertex $v_1$ has two edges connecting it to vertices $v_2$ and $v_3$.

Map-matched trajectories can be obtained from raw trajectories through map-matching~\cite{lou2009map,ma2020forecasting}. In practice, a raw trajectory often does not start or end at a vertex but rather at some point on an edge. 
When this occurs, the start and end points of a trajectory are moved to the nearest vertices, and their relative positions (i.e., distances from the original points to the snapped vertices along the edge) are recorded as offsets~\cite{li2020compression,li2021trace,crosby2019embedding}.
Each point in the map-matched trajectory has a timestamp. When the matched vertex directly corresponds to a raw GPS sample, the original timestamp is preserved. Otherwise, the timestamp is estimated by linear interpolation based on the point’s relative position along the edge.
For simplicity, we assume that all trajectories are aligned in the following context. Accordingly, we index map-matched trajectories and refer to them simply as \textit{trajectories} from now on.
Following the literature~\cite{wang2022road}, we define a trajectory as follows.

\begin{definition}\label{def:trajectory}
    
A  \textbf{trajectory}  $\textit{Tr} = \langle (v_0, t_0),$  $(v_1, t_1), \ldots,$  $(v_m, t_m) \rangle $ is obtained by map-matching a raw trajectory onto a road network $G = (V, E)$, such that $v_k\in V$ and $(v_{k-1}, v_k) \in E $, and $0<k\leq m$. 
Specifically, each pair $(v_i, t_i)$ indicates that the moving object is located at vertex $v_i \in V$ at time $t_i$.
\end{definition}

\begin{example}\label{ex:tra}

In Fig.~\ref{Fig:definition}, 10 dashed lines with different colors represent trajectories $\{\textit{Tr}_1, \ldots, \textit{Tr}_{10}\}$. Specifically, $\textit{Tr}_2$, $\textit{Tr}_3$, and $\textit{Tr}_5$ consist of four, four, and three vertices, respectively.

\end{example}

Trajectory range queries are fundamental operations in trajectory data management, as they underpin more advanced spatio-temporal analytical tasks. Such queries retrieve all trajectories—or some segments of trajectories—that lie within a specified spatio-temporal range~\cite{xu2017range,yin2022efficient,xu2023query}.

\begin{definition}
Given a spatial range $\mathcal{S}$ specified by a longitude interval $[x_{\textit{min}}, x_{\textit{max}}]$ and a latitude interval $[y_{\textit{min}}, y_{\textit{max}}]$, and a temporal range $\mathcal{T}$ specified by a time interval $[t_{\textit{start}}, t_{\textit{end}}]$,
a \textbf{trajectory range query} $q = (\mathcal{S}, \mathcal{T})$ returns all trajectories $\textit{Tr}$ that contain at least one point that satisfies both $\mathcal{S}$ and $\mathcal{T}$.
Note that this point may not be a sampled point. It can be any point along a trajectory's edge.
\end{definition}


\begin{example}
Continuing Example~\ref{ex:tra}, consider the range query $q = (\mathcal{S} = ([116.1,116.3],
[39.4,39.6]), \mathcal{T} = [8, 9])$. In Fig.~\ref{Fig:definition}, the trajectories that intersect $\mathcal{S}$ (the red rectangle) are $\textit{Tr}_2$ (orange dashed line), $\textit{Tr}_3$ (green dashed line), and $\textit{Tr}_5$ (cyan dashed line). Based on the map-matching details of these trajectories, since the temporal ranges for $\textit{Tr}_2$ and $\textit{Tr}_3$ are $[1, 4]$ and $[3, 6]$ respectively, only $\textit{Tr}_5$ satisfies the temporal range $[8, 9]$. Therefore, the result of $q$ is $\textit{Tr}_5$.
\end{example}

\subsection{Verifiable query methods}

A cryptographic hash function~\cite{preneel1994cryptographic,sobti2012cryptographic} takes input of any size and outputs a fixed-sized byte string.

\begin{definition}
\label{hashf}
A \textbf{cryptographic hash function}, denoted as $\textit{h}(\cdot)$, 
transforms input data $d$ into a hashed value $\textit{h}(d)$. 
The function satisfies three properties: 
(i) \textbf{preimage resistance}---given $\textit{h}(\cdot)$ and $\textit{h}(d)$, it is computationally infeasible to find $d$; 
(ii) \textbf{second preimage resistance}---given $\textit{h}(\cdot)$ and $\textit{h}(d)$, it is computationally infeasible to find $d$ and $d'$ such that $h(d)  = h(d')$; 
and (iii) \textbf{collision resistance}---given $\textit{h}(\cdot)$, it is computationally infeasible to find a pair $(d,d')$ such that $h(d)  = h(d')$.

\end{definition}

An authenticated data structure (ADS) is a data structure that allows efficient and secure verification of the integrity of data. A widely used example of an ADS is the Merkle tree~\cite{merkle1987digital}, which is a binary tree-based ADS that stores cryptographic hash values to efficiently verify data integrity.

\begin{definition}
A \textbf{Merkle tree} is a binary tree-based ADS built over a dataset $\mathcal{D} = \{d_1, \ldots, d_n\}$ that satisfies the following properties:

\begin{itemize}[leftmargin=*]
    \item Leaf nodes store a hash value $h_i = \textit{hash}(d_i)$ for  $1\leq i \leq n$.
    \item Non-leaf nodes store a hash value $h = \textit{hash}(h_{\text{left}} | h_{\text{right}})$, where $h_{\text{left}}$ and $h_{\text{right}}$ are the hash values of the left and right child nodes, respectively, and $|$ denotes concatenation.
    \item The hash value at the root node, called the \textit{Merkle root}, uniquely represents the entire dataset.
\end{itemize}

\end{definition}

\begin{figure}[]
\centerline{\includegraphics[width=0.48\textwidth]{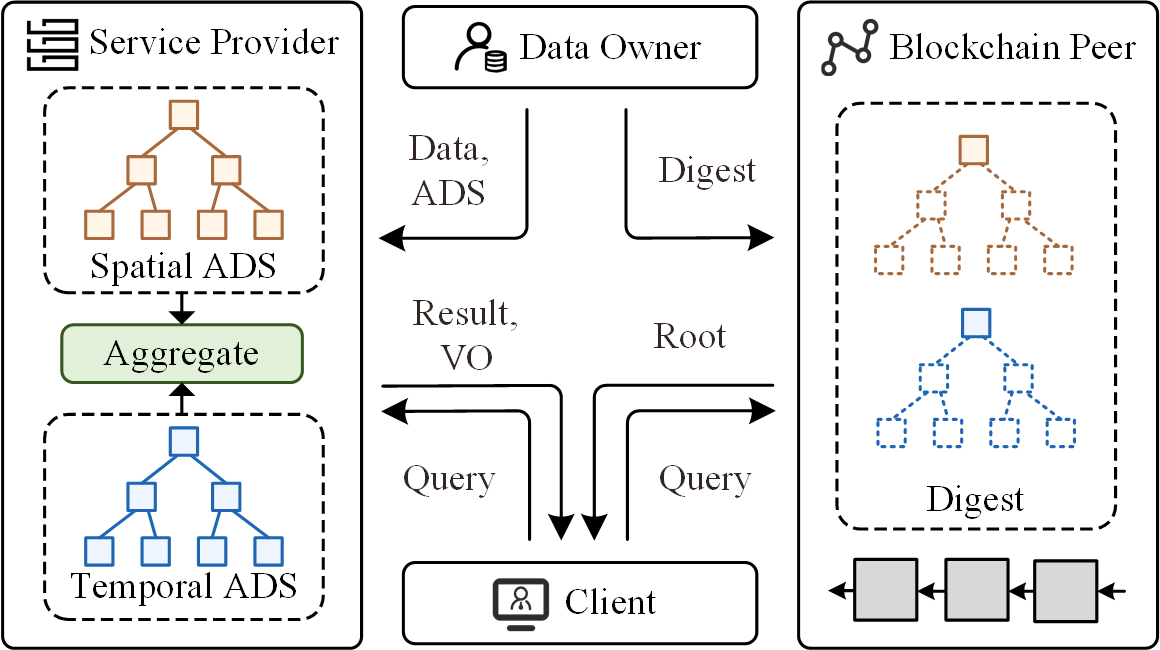}}
\caption{The \texttt{VTRQ} framework.}
\label{Fig:framework}
\end{figure}

\section{\texttt{VTRQ} framework}\label{sec:ADS}

\subsection{System model and threat model}
\textbf{Hybrid-storage blockchain system model.}
{\color{black}As shown in Fig.~\ref{Fig:framework}, the system involves four types of participants.
(i) \textit{Data owners (DO)} generate and maintain raw trajectory data. Given their limited computational and storage capabilities, data owners outsource their data to storage nodes.
(ii) \textit{Service providers (SP)} store the raw data off-chain and provide query processing services. 
 (iii) \textit{Clients} submit queries to service providers to retrieve information from storage nodes.
(iv) \textit{Blockchain peers (BP)} maintain the on-chain authenticated digests of the outsourced data using a blockchain.

DOs first outsource the trajectory data to SP while they construct ADSs and upload them to the SP. The digests of these ADSs (i.e., their root hashes, which serve as fingerprints) are stored on the BPs to enable verification. 
When a client issues a spatio-temporal query, it is decomposed into spatial and temporal components and processed separately  using two specialized ADSs: S-ADS and T-ADS. 
Both candidate trajectory sets returned are intersected and further filtered at a fine granularity based on query constraints. Finally, the SP returns the query results along with VOs. 
The client uses these VOs to reconstruct the root digests and verifies them against the blockchain-stored digests to ensure completeness and soundness.
The intersection and fine-grained filtering can also be locally reproduced by the client to validate the process.

}


{In \texttt{VTRQ}, the blockchain serves as a third-party trust anchor that stores and publishes the digests of two ADSs for query-result verification; it does not store raw data. \texttt{VTRQ} is blockchain-agnostic and does not rely on any specific platform or consensus protocol, as long as the blockchain provides append-only and tamper-evident storage and publication. Therefore, configuration choices such as the consensus mechanism, validator set, and tolerated number of malicious nodes are deployment-level concerns and do not affect the correctness of the proposed verification mechanism.}

\begin{figure*}[]
  \centering
  \begin{tikzpicture}
    \node[anchor=north east,inner sep=0] (img1) at (0,0)
      {\includegraphics[width=0.95\textwidth]{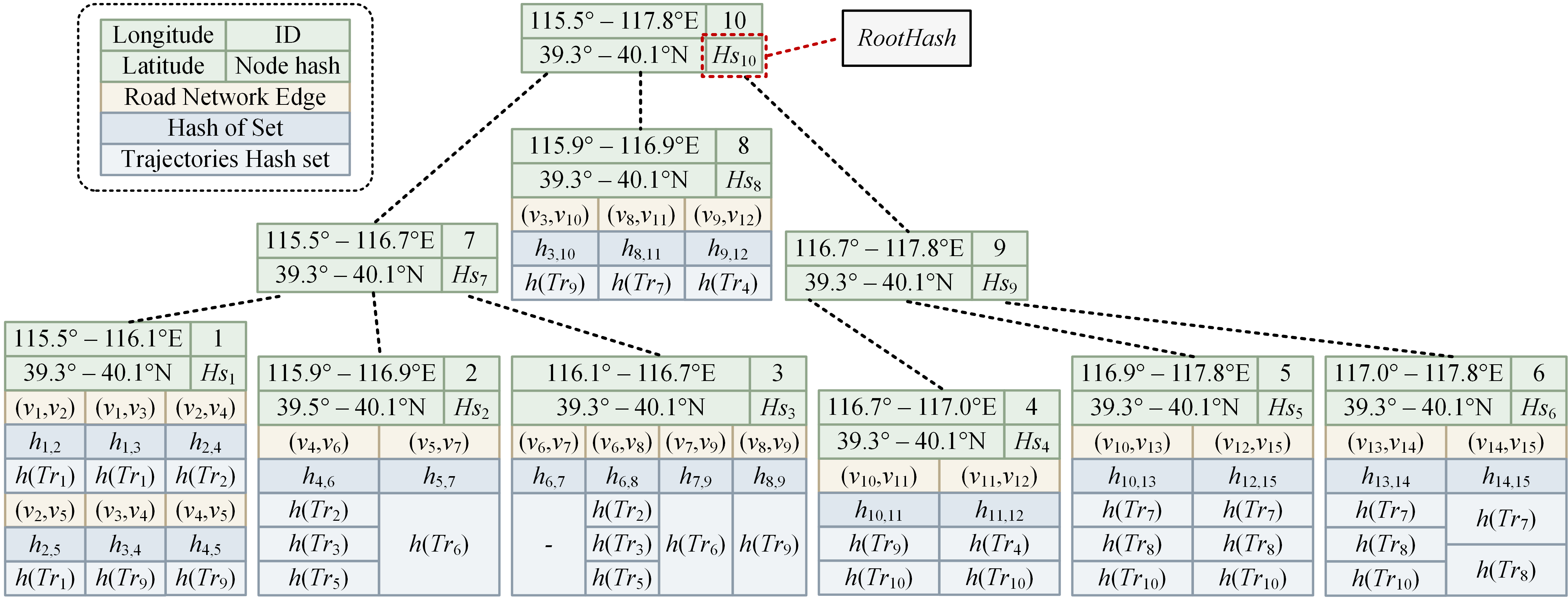}};
       \node[below=0.05cm of img1.south, align=center] {\small(a) S-ADS.};

       \node[anchor=north east,inner sep=0] (img2) at (0,0)
      {\includegraphics[width=0.31\textwidth]{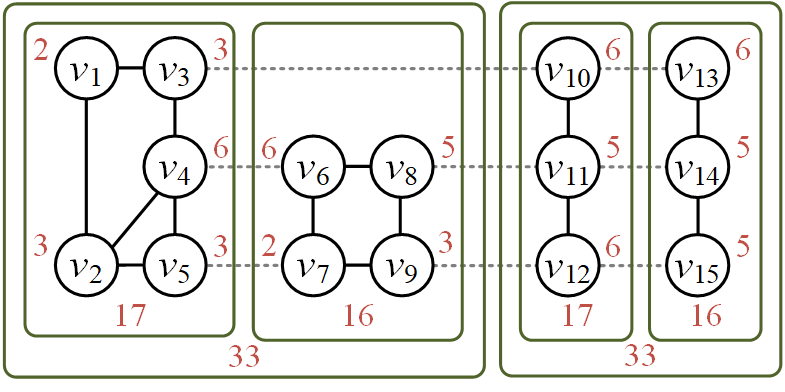}};
   \node[anchor=north, yshift=0cm] at (img2.south)
  {\small (b) Road network partition.};
  \end{tikzpicture}

  \caption{An example of S-ADS.}

  \label{Fig:spatial_ads}
\end{figure*}

\noindent\textbf{Threat model.}
According to existing studies~\cite{wu2023enabling,zhang2019gem,zhang2021authenticated,li2024authenticated1,li2024authenticated2,zhou2023veridkg}, in blockchain-based data outsourcing architectures, the clients, DOs, and BPs are generally considered trustworthy.
However, the SP, as an external service resource, may return tampered or incomplete query results due to program failures, security vulnerabilities, or commercial interests, which can compromise query correctness. Therefore,  following the setting in~\cite{xu2019vchain,zhou2023veridkg}, we treat the SP as an untrusted peer. 
The goal of verifiable query processing is to enable the SP to prove that its query results satisfy completeness and soundness. \textit{Completeness} ensures that no valid results are omitted in a trajectory range query, while \textit{soundness} means that all returned results are authentic and meet the query conditions.


\subsection{Road-aware spatial ADS}

\subsubsection{Construction}

Inspired by~\cite{wang2022road}, we represent a trajectory as a sequence of edges with corresponding timestamps (see Definition~\ref{def:trajectory}) and index trajectories based on these edges. Building on this representation, we develop a road-aware spatial ADS that can efficiently prune irrelevant trajectories and effectively support both range queries and the verification of query results.

\begin{definition}
Given a road network $G = (V, E)$ and a set of trajectories $S_\textit{Tr} = \{\textit{Tr}_1, \textit{Tr}_2, \dots, \textit{Tr}_n\}$, a \textbf{road-aware spatial ADS (S-ADS)} \textit{Ts} is recursively represented as a triple tree:

\[\small
\textit{Ts} =
\begin{cases}
([x,x'], [y,y'], E_h,\textit{Hs}({\textit{Ts}}))   \quad\quad\quad \text{if $\textit{Ts}$ is a leaf node} \\
([x,x'], [y,y'], (\textit{Ts}_0, \textit{Ts}_l, \textit{Ts}_1),\textit{Hs}({\textit{Ts}}))   \text{\quad\quad   otherwise}
\end{cases} 
\]


\begin{itemize} [leftmargin=*]
\item $[x,x']$ is a longitude range.
\item $[y,y']$ is a latitude range.
\item $E_h$ is a set of edges. Each element in $E_h$ is $(v_i,v_j,H(\textit{Tr})_{i,j},h_{i,j})$, where $v_i,v_j \in V$, $H(\textit{Tr})_{i,j}$ is a set of hashes of trajectories passing through the edge $v_i\rightarrow v_j$, and $h_{i,j}=h(h(\textit{Tr}_a)\|h(\textit{Tr}_b)\|h(\textit{Tr}_c) \| \dots)$ is an aggregate hash. $h(\textit{Tr}_a),h(\textit{Tr}_b),h(\textit{Tr}_c) \in H(\textit{Tr})_{i,j}$.

\item $\textit{Ts}_0$, $\textit{Ts}_l$, and $\textit{Ts}_1$ are the left, middle, and right subtrees of $Ts$, respectively. 
$\textit{Ts}_l$ is a leaf node, which stores the border edges between the two subgraphs obtained after partitioning the road network. $\textit{Ts}_0$ and $\textit{Ts}_1$ have the same structure as $Ts$. At the lowest level, the left and right subtrees are composed of leaf nodes that store the edges corresponding to their spatial regions.
\item $\textit{Hs} (\textit{Ts})$ is the hash value of the node:
\[\small
\textit{Hs}({\textit{Ts}}) =
\begin{cases}
h([x,x']\|[y,y']\|E_h)  \text{\;\;\;\;\;\;\;\;\;\;\;\;\;\;\;if \textit{Ts} is a leaf node} \\
h([x,x']\|[y,y']\|\textit{Hs}_{\textit{node}})   \text{\;\;\;\;\;\;\;\;\;\;\;\;\;\;\;\;\;\;\;\;\;\;\;\;\;otherwise}
\end{cases}
\]
where  $\textit{Hs}_{\textit{node}}=\textit{Hs}({\textit{Ts}_0})\|\textit{Hs}(\textit{Ts}_l)\|\textit{Hs}(\textit{Ts}_1)$.
\end{itemize}
\end{definition}

To build S-ADS, we recursively partition the road network using the method proposed in~\cite{wang2022road}. Edges shared between the two subgraphs are referred to as border edges. The weight of each edge is  the number of trajectories traversing it, and the weight of a vertex is the sum of the weights of its incident edges. A binary search strategy is employed to select a splitting point that minimizes the difference in trajectory weights between the resulting subgraphs. The graph is then divided into left and right subgraphs at this point. This process repeats recursively until the number of vertices in each leaf node falls below or equals a threshold $\theta$, which controls the size of the leaf nodes. The partitioning scheme can be extended to support various constraints~\cite{hendrickson1995multi,karypis1995analysis,karypis1998multilevel}, such as minimizing cut edges or balancing subgraph sizes.

\begin{example}
Continuing Fig.~\ref{Fig:definition}, Fig.~\ref{Fig:spatial_ads}(a) provides an example of S-ADS $\textit{Ts}$. 
Given a threshold $\theta=5$, the partitioning of $\textit{Ts}$ is shown in Fig.~\ref{Fig:spatial_ads}(b). 
We first calculate the weight of each node and subgraph, i.e., the red values.
Take $v_4$ as an example:  $v_4$ is connected by four edges, i.e., $v_2 \rightarrow v_4$, $v_3 \rightarrow v_4$, $v_4 \rightarrow v_5$, and $v_4 \rightarrow v_6$.
According to the number of trajectories passing through each edge in Fig.~\ref{Fig:definition}, their respective weights are 1, 1, 1, and 3. Therefore, the total weight of vertex $v_4$ is 6.
After calculating the weights of all vertices, we partition the road network. Since the longitudinal range from 115.5°E to 117.8°E is longer, we partition along the longitude direction.
After~splitting at $v_3 \rightarrow v_{10}$, $v_8 \rightarrow v_{11}$, and $v_9 \rightarrow v_{12}$, the~graph~can~be~divided into two subgraphs. However, the number of vertices in both subgraphs exceeds the threshold~$\theta$.~We~continue to partition them following the same procedure. Finally, the resulting four subgraphs, consisting of the node sets $\{v_1,v_2,v_3,v_4,v_5\}$, $\{v_6,v_7,v_8,v_9\}$, $\{v_{10},v_{11},v_{12}\}$, and $\{v_{13},v_{14},v_{15}\}$, 
all have at most $\theta$ vertices.

Using this partitioning, a $\textit{Ts}$ can be constructed. There are two types of vertices in $\textit{Ts}$: leaf vertices and internal vertices. For ease of description, we assign a unique identifier to each node in this example. 
Take node 10 as an example.
It is both the root node and an internal node, and it has three child nodes. Among them, nodes 7 and 9 are internal nodes, while node 8 is a leaf node. This leaf node stores the border edges: $v_3 \rightarrow v_{10}$, $v_8 \rightarrow v_{11}$,  and $v_9 \rightarrow v_{12}$. Node 7 also has three child nodes, all of which are leaf nodes. However, only node 2 records the border edges connecting its two subgraphs. Node 1 and node 3 store the edges within the subgraphs formed by  $\{v_1,v_2,v_3,v_4,v_5\}$ and $\{v_6,v_7,v_8,v_9\}$, respectively.
\end{example}

\begin{figure*}[]
\centering
\subfigure[Example]{\includegraphics[width=0.61\textwidth]{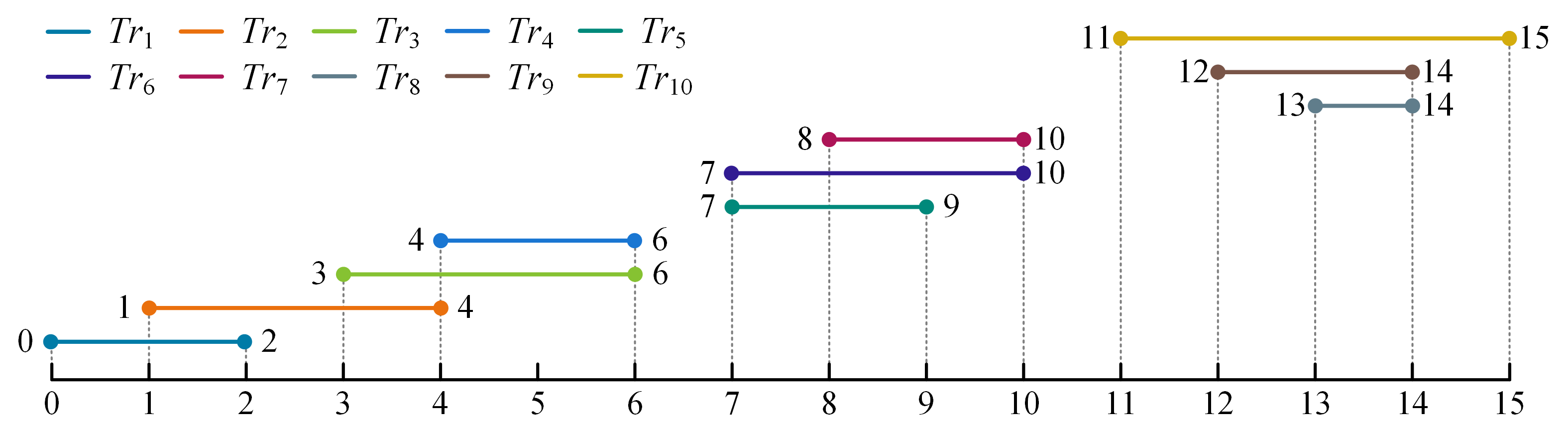}}\hspace{0.3cm}
\subfigure[T-ADS]{\includegraphics[width=0.35\textwidth]{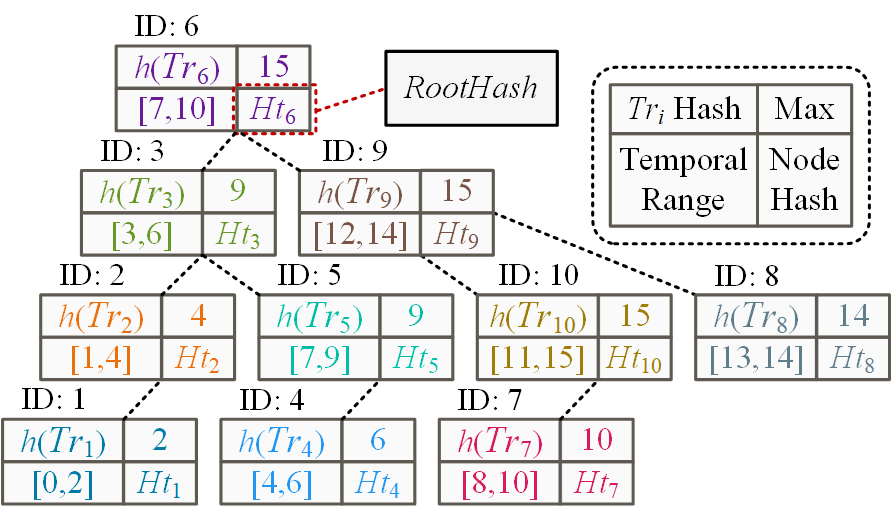}}

\caption{An example of T-ADS.}

\label{fig:temporal}
\end{figure*}

To enhance the verifiability support of the S-ADS, our method introduces three key improvements: (i) it separates spatial and temporal attributes by removing timestamps from the spatial index, allowing temporal and spatial data to be indexed and queried independently, which improves overall efficiency; (ii) it constructs MBRs for each group of border edges, enabling early filtering of non-matching MBRs during query processing and thus reducing unnecessary spatial computations; and (iii) it employs a three-level hierarchical hashing scheme to improve reusability, where each trajectory is first hashed individually, then edge-level hashes are generated by aggregating the hashes of all trajectories passing through each edge, and finally, node-level hashes are computed by combining the hashes of associated edges with relevant metadata. 

\subsubsection{Update}
Similar to~\cite{wang2022road}, when updating the trajectory dataset, a batch of new raw trajectories is first map-matched to the road network. Their hash values are then computed, and each map-matched trajectory edge is indexed on the leaf node. The weights of each subgraph are updated accordingly, which may lead to an imbalance in the number of trajectories across ADS nodes.
After updating the subgraph weights, unbalanced nodes are identified based on the partitioning optimization criteria. We then remove from this set any nodes whose ancestors are already included, retaining only the minimal set of affected nodes. For each remaining node, the partitioning is adjusted: if the left child’s weight exceeds that of the right child, the split boundary is shifted left; otherwise, it is shifted right, until balance is achieved. This process recurses until the leaf node size reaches the stopping threshold.
Finally, for the updated nodes, edge hashes are computed by concatenating and hashing the trajectory hashes. The node hashes are computed by concatenating and hashing the edge hashes. The node hash values are propagated upward from the modified nodes to the root.

\subsubsection{Time complexity analysis}

Given a spatial range query, the system traverses from the root of the S-ADS tree to locate all nodes whose regions intersect with the query range. Let $N$ denote the number of nodes in the ADS tree. Suppose the query overlaps with $s$ leaf nodes, each containing up to $k$ trajectory edges, then the candidate set size is $O(s \cdot k)$. Therefore, the overall query time complexity is $O(\log N + s \cdot k)$, where $\log N$ accounts for tree traversal and $s \cdot k$ for filtering and matching.

\subsection{Interval tree-based temporal ADS}

\subsubsection{Construction}
The RP-Tree~\cite{wang2022road} stores the timestamps of all trajectories at each vertex in the leaf nodes. As the number of trajectories increases, the number of timestamps grows significantly, resulting in very slow query performance. If the RP-tree is adapted as a verification tree for trajectory range queries, the verification process will become inefficient due to some timestamps being unnecessary. Therefore, in designing the S-ADS, we did not store any time information. Instead, we additionally designed an interval tree-based temporal ADS to record only the start and end times of each trajectory~\cite{yao2024tsec}, which reduces the large number of timestamps in the structure. This design improves query efficiency and reduces storage overhead. 
The T-ADS adopts a red-black tree structure for balanced indexing and efficient updates, ensuring logarithmic query and insertion complexity~\cite{finis2015indexing}.
Moreover, it provides more flexible query capabilities, since we do not embed time information into the spatial index, nor do we embed spatial information into the temporal index. As a result, purely spatial queries and purely temporal queries can be supported independently.

\begin{definition}

Given a set of trajectories $S_\textit{Tr} = \{\textit{Tr}_1, \textit{Tr}_2, \dots, \textit{Tr}_n\}$, where each $T_i$ has a  sequence of road network nodes and their associated timestamps  $(v_k, t_k)$, an \textbf{interval tree-based temporal ADS (T-ADS)} \textit{Tt} is recursively represented as a binary tree:
\[
\textit{Tt} =
(h(\textit{Tr}_w), [t,t'],t_{max},(\textit{Tt}_0,\textit{Tt}_1),\textit{Ht}({\textit{Tt}})).
\]

\begin{itemize} [leftmargin=*]
\item $h(\textit{Tr}_w)$ is the hash value of a trajectory $\textit{Tr}_w$, which can uniquely identify the trajectory.
\item $[t, t']$ is a temporal range, where $t$ is the start time of the $\textit{Tr}_w$ and $t'$ is the end time of the $\textit{Tr}_w$.
\item $t_{max}$ is the maximum timestamp among all nodes in the tree $\textit{Tt}$.

\item $\textit{Tt}_0$ and $\textit{Tt}_1$ are the left and right subtrees of $\textit{Tt}$, respectively. $\textit{Tt}_0$ and $\textit{Tt}_1$ have the same structure as $T_t$.

\item $\textit{Ht} (\textit{Tt})$ is the hash value of the node:
\[\small
\textit{Ht}({\textit{Tt}}) =
\begin{cases}
h(\textit{Tr}_w\|[t,t']\|t_{min}\|t_{max}) \text{\quad\quad  if \textit{Tt} is a leaf node} \\
h(\textit{Tr}_w\|[t,t']\|t_{min}\|t_{max}\|\textit{Ht}_{\textit{node}})  \text{ \quad \quad\;\;otherwise}
\end{cases}
\]
where   $\textit{Ht}_{\textit{node}}=\textit{Ht}({\textit{Tt}_0})\|\textit{Ht}(\textit{Tt}_1)$.

\end{itemize}
\end{definition}

\begin{example}
Fig.~\ref{fig:temporal}(a) shows the time intervals of the 10 trajectories from Fig.~\ref{Fig:definition}. The axes represent their timestamps. Fig.~\ref{fig:temporal}(b) is the Ts-tree corresponding to  Fig.~\ref{fig:temporal}(a), constructed based on these time intervals. For convenience, we use $i$ as the node label; it is the index of the trajectory $\textit{Tr}_i$.  
Node 6 is an internal node and also a root node. It has two child nodes, 3 and 9. Interval $[7,10]$ represents the start time of trajectory $\textit{Tr}_6$. Next, 15 is the maximum value of the subtree, and $\textit{Ht}_6$ is the hash value combining the information of $\textit{Tr}_6$ with the information of the subtree. Node~1 is a leaf node, and its $\textit{Ht}_1$ only contains the information of trajectory~1.

Given a query time interval $[5,6]$, we start from the root node, whose temporal range is $[7,10]$. Although this range does not intersect with the query interval, its subtree maximum value is 15, which exceeds the lower bound of the query (5), requiring us to continue searching its left subtree. Node 3 has a temporal range of $[3,6]$, overlapping with the query, and a maximum value of 9, which is also greater than 5; therefore, we proceed to its children. Among them, node 4 has a temporal range of $[4,6]$ that overlaps with the query and stores trajectory $Tr_4$, which is returned as a result. The left child of node 3 (node 1) has a temporal range of $[0,2]$, which does not intersect with the query interval, and its maximum value is 2, smaller than 5, so it can be safely pruned. Returning to the right child of the root node (node 9), its temporal range is $[12,14]$, which does not intersect the query interval. Although its maximum value is 15, larger than 5, there is no overlap with $[5,6]$, so no further traversal is necessary. Ultimately, only trajectory $Tr_4$ matches the query interval $[5,6]$.
\end{example}


\subsubsection{Update}

In a T-ADS, the update process for inserting a new interval transaction consists of three main steps. First, the interval is inserted into the interval tree, which is built upon a red-black tree structure. To maintain balance, the red-black tree may perform rotations (left or right) to restructure the tree and recoloring operations to preserve red-black properties—such as preventing consecutive red nodes and ensuring consistent black-height across all paths. Second, the $\textit{Max}$ field of each affected node along the insertion path is updated upward to reflect the maximum end time among all intervals in its subtree. Finally, the hash values along the insertion path are recomputed from the bottom up, culminating in an updated $\textit{RootHash}$ at the top of the structure to preserve the completeness and soundness of the data. 

\subsubsection{Time complexity analysis}

Inserting a new interval involves three steps: (i) inserting the interval into the red-black tree and performing necessary rebalancing, which takes $O(\log n)$ time, where $n$ denotes the current number of intervals in the tree; (ii) updating the $\textit{Max}$ fields along the insertion path, affecting at most $\log n$ nodes; (iii) recomputing the hash values of these affected nodes to preserve the integrity of the authenticated data structure, which also involves at most $\log n$ nodes. Since each step operates within the tree's height, the overall time complexity of the insertion is $O(\log n)$.
Given a query time interval, we traverse the tree to locate all intervals that intersect with it. Each query path has a depth of $O(\log n)$. Let $k$ be the number of matched intervals, then the time complexity of the query is $O(\log n + k)$.

\section{Query and verification}\label{sec:query_verification}

\subsection{Spatio-temporal edge aggregation}\label{sec:aggregation}


After independently retrieving candidate sets from S-ADS and T-ADS, we can take the intersection of the two sets to obtain the query result. However, this approach overlooks cases where only part of an edge overlaps with a query interval. For example, in Fig.~\ref{Fig:introduction}, consider  again the query $q = (([4, 7], [3, 6]), [4, 5])$. Suppose a leaf node covers the region $([3,7], [3,6])$, and $\textit{Tr}_1$ is among the candidate trajectories. The index compares the time interval $[2,7]$ of $\overline{p_{1,2}p_{1,3}}$—the sub-trajectory of $\textit{Tr}_1$ within the leaf node region—with the query interval $[4,5]$. Since $[2,7]$ covers $[4,5]$, the system considers it a match. In reality, only the edge $\overline{p' p_{1,3}}$ actually lies within the spatial range, and its time interval begins after $5$, which does not overlap with $[4,5]$. Thus, $\textit{Tr}_1$ does not satisfy the~query.


\begin{algorithm}
\caption{Edge aggregation}
\label{alg:stea}
\KwIn{Spatial result set $\mathcal{R}_s$, Temporal result set $\mathcal{R}_t$, spatial region $\mathcal{S}$, temporal interval $T$}
\KwOut{Spatio-temporal result set $\mathcal{R}$}

$\mathcal{R}' \gets \mathcal{R}_s \cap \mathcal{R}_t$\;
$\mathcal{R} \gets \emptyset$\;

\ForEach{$\textit{Tr}_n \in \mathcal{R}'$}{
    \ForEach{edge $e = (v_i, v_j)$ in $\textit{Tr}_n$}{
        $T_e \gets$ time interval of $e$\;
        
        \uIf{$T_e \subseteq T$ \textbf{and} $e \in (E_{ri} \cup E_{rc})$}{
            $\mathcal{R} \gets \mathcal{R} \cup \{e\}$\;
        }
        \uElseIf{$T_e \cap T \neq \emptyset$}{
            \uIf{$e \in E_{ri}$}{
                $\mathcal{R} \gets \mathcal{R} \cup \{e\}$\;
            }
            \uElseIf{$e \in E_{rc}$}{
                Compute intersection points $p_{in}, p_{out}$\;
                Estimate timestamps $t_{in}, t_{out}$\;
                \If{$[t_{in}, t_{out}] \cap T_q \neq \emptyset$}{
                    $\mathcal{R} \gets \mathcal{R} \cup \{e\}$\;
                }
            }
        }
    }
}
\Return $\mathcal{R}$\;
\end{algorithm}

To address this problem, we propose a spatio-temporal edge aggregation mechanism (STEA) to achieve fine-grained alignment between spatial and temporal attributes in verifiable trajectory range queries.
We first perform independent queries on the S-ADS and T-ADS to obtain result sets $ \mathcal{R}_s $ and $ \mathcal{R}_t $, which satisfy the spatial and temporal constraints, respectively. The STEA process then proceeds in two main steps: (i) compute the intersection $ \mathcal{R}' = \mathcal{R}_s \cap \mathcal{R}_t $ to derive an initial set of candidate trajectories; and (ii) for each trajectory in $ \mathcal{R}' $, examine all of its edges to verify whether any edge simultaneously satisfies both the spatial and temporal query conditions.

To support this aggregation, we categorize the relationship between an edge $e = (v_i, v_j)$ and the spatial region $\mathcal{S}$ into three types:

\begin{definition}
For an edge $e = (v_i, v_j)$, if both $v_i$ and $v_j$ are located within the spatial query region $\mathcal{S}$, $e$ is said to be \textbf{inside} $\mathcal{S}$. All such edges are denoted as $E_{ri}$.
\end{definition}

\begin{definition}
For an edge $e = (v_i, v_j)$, if $v_i$ is inside $\mathcal{S}$ and $v_j$ is outside $\mathcal{S}$, or vice versa, or both are outside but $e \cap \mathcal{S} \neq \emptyset$, then $e$ is said to \textbf{intersect} $\mathcal{S}$. All such edges are denoted as $E_{rc}$.
\end{definition}

\begin{definition}
For an edge $e = (v_i, v_j)$, if $e \notin (E_{ri} \cup E_{rc})$ then $e$ is considered \textbf{outside} $\mathcal{S}$. These edges are denoted as $E_{ro}$, and defined by $E_{ro} = E - E_{ri} - E_{rc}$.
\end{definition}

Based on the above classification and the edge’s temporal overlap with the query time range, STEA filters candidate results through the following steps:

\begin{itemize}[leftmargin=*]
    \item \textbf{Full temporal inclusion with spatial overlap:} If an edge's time interval is completely within the query time range and it belongs to $E_{ri} \cup E_{rc}$, the edge is retained.

    \item \textbf{Partial temporal overlap:} If the time interval of the edge overlaps with the query's time interval, then:
    \begin{itemize}[leftmargin=*]
        \item If it belongs to $E_{ri}$, it is retained directly.
        \item If it belongs to $E_{rc}$, the spatial intersection points between the edge and the query region are computed. The corresponding timestamps are estimated using linear interpolation. If the interpolated time interval of spatial overlap intersects with the query time range, the edge is retained.
    \end{itemize}

\end{itemize}

Algorithm~\ref{alg:stea} describes the process of STEA. Given the spatial result set $\mathcal{R}_s$, temporal result set $\mathcal{R}_t$, spatial region $\mathcal{S}$, and temporal interval $T$, the algorithm first computes the initial candidate set $\mathcal{R}' = \mathcal{R}_s \cap \mathcal{R}_t$ (Lines~1-2). It then iterates through each trajectory in $\mathcal{R}'$ and examines all of its edges to verify their spatio-temporal validity (Lines~3--15).
For each edge, if its time interval is fully contained within $T$ and it has spatial overlap (i.e., it belongs to $E_{ri}$ or $E_{rc}$), it is directly retained (Lines~5-6). If the time interval partially overlaps with $T$, then: (i) if the edge is spatially inside $\mathcal{S}$ (i.e., $e \in E_{ri}$), it is also retained (Lines~8--9); (ii) if it intersects $\mathcal{S}$ (i.e., $e \in E_{rc}$), the algorithm computes the timestamps of the spatial intersection points. If the interpolated time interval overlaps with $T$, the edge is retained as a valid spatio-temporal segment (Lines~11--14).

\subsection{Verification of query results}

Once the query execution is completed, the system must verify the soundness and completeness of the result before it can be trusted. The verification process consists of two stages: (i) verifying the correctness of the query results obtained from the two ADSs; and (ii) verifying the correctness of the final result after aggregation.

\noindent\textbf{Verification of ADS query results.}
The system uses two ADSs, corresponding to the spatial and temporal indexes, to support verifiable range queries. Each internal node stores a cryptographic hash summarizing its child nodes, while each leaf node contains a hash digest of its own content.

During query execution, each ADS performs a depth-first traversal and returns not only the matched result set but also a VO, which includes: (i) the query path from the matched leaf nodes to the root; (ii) the authentication hashes of all sibling nodes along the path. When a query returns no results (i.e., the result set is empty), the ADS still provides a verification object to prove that the empty result is trustworthy. Specifically, the VO includes a Merkle verification path that demonstrates the query range falls outside or between the indexed intervals and does not match any leaf node.
Using this information, the client can locally reconstruct the ADS root hash $\textit{RootHash}$ and compare it with the trusted root hash $\textit{RootHash}'$ published on the blockchain. If $\textit{RootHash} = \textit{RootHash}'$, the query result is considered correct and complete with respect to the corresponding ADS. This verification is performed independently for both the S-ADS and T-ADS.

\noindent\textbf{Verification of aggregated results.}
After verifying the query results from the S-ADS and T-ADS, the client can re-execute the spatio-temporal aggregation process locally, as described in Section~\ref{sec:aggregation}, to further ensure that the returned trajectory edges satisfy the combined query constraints. This re-execution incurs minimal overhead and does not impose a substantial burden on the client.

\section{Security analysis}\label{sec:verifiability}

{

We provide the formal security definition of \texttt{VTRQ} in line with common security standards for verifiable queries~\cite{zhou2023veridkg,10.65109/GPHO5000,yao2025vgq}. 

\begin{definition} \label{query_ve} (Query verifiability). A trajectory range query is verifiable if the probability that any polynomial-time adversary $\mathcal{A}$ succeeds in the following experiment is negligible:

Given a trajectory range query $q$, $\mathcal{A}$ is picked as the service provider responsible for executing trajectory queries, and $\mathcal{A}$ executes $q$ and produces a result set $\mathcal{R}$ with proof $VO$. $\mathcal{A}$ succeeds if: (i) An element exists in $\mathcal{R}$ that does not satisfy $q$ (violating soundness); or (ii) there exists a result element that satisfies $q$ but is not in $\mathcal{R}$ (violating completeness). \end{definition}

\begin{theorem} \label{theo} The trajectory range query is verifiable under Definition~\ref{query_ve} if the following assumptions hold: (i) the hash function is a pseudo-random function, and (ii) the blockchain peers are trustworthy.
\end{theorem}

\begin{proof}We prove the theorem by contradiction, addressing two cases:

\underline{Case 1}. An element of the result that satisfies $q$ is missing from $\mathcal{R}$, violating the completeness.

\underline{Case 2}. $\mathcal{R}$ contains an element that does not satisfy $q$, violating the soundness.

During the verifiable query process, the SP queries the S-ADS and T-ADS. From the S-ADS, the SP obtains an intermediate result $ \mathcal{R}_s $ along with the verification object $ \textit{VO}_s$. From the T-ADS, an intermediate result $ \mathcal{R}_t $ and its proof $ \textit{VO}_t$ are obtained. The client combines each result with its corresponding proof to reconstruct the root digests of the ADSs, which are then compared against the digests stored on the blockchain to verify the completeness and soundness of the results.
If any portion of the result is missing, this implies that the ADS allows the construction of an incomplete result that still produces the same root digest as the correct one, thereby violating completeness (Case 1). 
If the result is tampered with or forged, it suggests that two distinct query results can be derived from the ADS while producing the same root digest, thereby violating soundness (Case 2).
If either Case 1 or Case 2 occurs, one of two contradictions arises:

\begin{itemize}[leftmargin=*]
\item The ADS digest root is generated by a cryptographic hash function. If two distinct values produce the same digest, this implies a hash collision, contradicting Assumption i.
\item The ADS digest root stored on the blockchain has been tampered with, indicating that a majority of peers maintaining the blockchain are malicious and have gained control over the system, contradicting Assumption ii.
\end{itemize}

Thus, under the assumptions, $\mathcal{A}$ cannot violate verifiability.
\end{proof}
}

{
\begin{table*}[!t]
\centering

\caption{Trajectory query time (s) for all methods.}

\begin{tabular}{cccccccccccccccc}
   \toprule

\multicolumn{2}{c}{Dataset} & \multicolumn{3}{c}{Xi'an} & \multicolumn{3}{c}{Chengdu} &\multicolumn{3}{c}{Geolife}\\
\cmidrule(lr){1-2} \cmidrule(lr){3-5} \cmidrule(lr){6-8} \cmidrule(lr){9-11}
 \textbf{Tem}&\textbf{Spa} & \textbf{1km} & \textbf{2km} & \textbf{3km}  & \textbf{1km} & \textbf{2km} & \textbf{3km} & \textbf{1km} & \textbf{2km} & \textbf{3km} \\

\midrule
\multirow{4}{*}{\rotatebox{90}{\textbf{1 min}}}

& MBT
& 0.488  & 0.935 & 1.321
& 0.506 & 0.856 & 1.261 
& 0.245 & 0.300 & \underline{0.528}  \\

& MPT
& 0.822 & 2.449 & 4.967
& 0.903 & 2.841 & 5.557
& \underline{0.117} & \underline{0.225} & 0.905  \\

& RPMT 
& \underline{0.098} & \underline{0.211} & \underline{0.319}
& \underline{0.110} & \underline{0.302} & \underline{0.503} 
& 0.235 & 0.304 & 1.215  \\

& VTRQ  
& \textbf{0.051} & \textbf{0.128} & \textbf{0.159}
& \textbf{0.049} & \textbf{0.146} & \textbf{0.241} 
& \textbf{0.087} & \textbf{0.100} & \textbf{0.497} \\

\midrule

\multirow{4}{*}{\rotatebox{90}{\textbf{5 min}}} 
& MBT 
& 0.487 & 0.976 & 1.377
& 0.524 & 0.933 & 1.315
& 0.248 & 0.309 &  \underline{0.541}
 \\

& MPT
& 1.291 & 2.512 & 5.166
& 0.911 & 2.854 & 5.658
&  \underline{0.119} &  \underline{0.240} & 0.935  \\

& RPMT 
& \underline{0.108} & \underline{0.218} & \underline{0.357}
& \underline{0.117} & \underline{0.319} & \underline{0.523}
& 0.243 & 0.331 & 1.333
 \\

& VTRQ  
& \textbf{0.062} & \textbf{0.130} & \textbf{0.160}
& \textbf{0.066} & \textbf{0.152} & \textbf{0.244}
& \textbf{0.093} & \textbf{0.108} & \textbf{0.507} \\

\midrule
\multirow{4}{*}{\rotatebox{90}{\textbf{10 min}}} 
& MBT
& 0.516 & 1.005 & 1.387
& 0.532 & 0.971 & 1.421
& 0.251 & 0.316 & \underline{0.559}  \\

& MPT
& 1.223 & 2.477 & 5.788
& 0.944 & 3.382 & 5.744
& \underline{0.123} & \underline{0.248} & 0.985  \\

& RPMT 
& \underline{0.114} & \underline{0.220} & \underline{0.393}
& \underline{0.123} & \underline{0.320} & \underline{0.535}
& 0.251 & 0.361 & 1.439
  \\

& VTRQ  
& \textbf{0.067} & \textbf{0.135} & \textbf{0.172}
& \textbf{0.070} & \textbf{0.155} & \textbf{0.251}
& \textbf{0.104} & \textbf{0.103} & \textbf{0.512} \\

\bottomrule
\end{tabular}
\label{tab:Query on trajectory index1}

\end{table*}}

{
\begin{table}[t]
\renewcommand\tabcolsep{5pt}
\centering
 \caption{Number of results for all methods.}

\begin{tabular}{cccc}
   \toprule

Dataset &Xi'an &Chengdu  &Geolife\\

\midrule

MBT & 54  & 31 & 3   \\
 MRT & 54  & 40 & 4  \\
 RPMT & 55  & 42  &5  \\
 VTRQ  & 55  & 42 & 5  \\

\bottomrule
\end{tabular}
\label{tab:result1}

\end{table}}

\section{Experiments}\label{sec:experiments}
\subsection{Experimental settings}

\noindent\textbf{Implementation details.} 
We implement a service provider and a client in Python 3.10.11, using JSON to store the spatial and temporal ADSs due to its lightweight nature, readability, and efficient parsing. The service provider runs on an Intel Core i9 CPU @ 2.67GHz with 16 GB of memory, while the client runs on an Intel Core i7 CPU @ 2.60GHz with 16 GB memory.
{
{\texttt{VTRQ} uses the blockchain only as a third-party trust anchor for storing and publishing the digests of two ADSs, and it does not depend on any specific blockchain platform or consensus protocol. We choose to implement this layer using a Tendermint-based blockchain system~\cite{tendermintWebsite}. The system is deployed on macOS 13.4 with LevelDB as the underlying storage engine.}

\noindent\textbf{Datasets.} 
We use three datasets: Chengdu, Xi'an, and Geolife. The first two stem from Didi~\cite{didiGlobalWebsite}
, while the third one is from Microsoft~\cite{zheng2011geolifeUserGuide}.

\begin{itemize}[leftmargin=3.4mm]




\item \textbf{Xi'an}. The Xi'an trajectory dataset was extracted from the online car-hailing in Xi'an between October 1 and November 30, 2018. It comprises 45,851 trajectories with approximately 79 million trajectory points and incorporates part of the Xi'an road network, which includes 8,651 edges and 3,946 vertices.
    
\item \textbf{Chengdu}. The Chengdu trajectory dataset was extracted from the online car-hailing in Chengdu from October 1, 2018 to November 30, 2018. It contains 284,608 trajectories with approximately 220 million trajectory points and incorporates part of the Chengdu road network, which includes 8,882 edges and 5,725 vertices.
    
\item \textbf{Geolife}. The Geolife trajectory dataset was collected by 178 users in the Microsoft Research Asia Geolife project between April 2007 and October 2011. It comprises 17,621 trajectories with a total distance of 1,251,654 kilometers and a total duration of 48,203 hours. It also includes part of the corresponding Beijing road network, which contains approximately 57,202 edges and 48,662 vertices.

    
    

\end{itemize}

All datasets consist of $((x, y), t)$  points.
We apply an existing map-matching method~\cite{mao2025dutytte} to project data from the datasets onto the road networks. 


\noindent\textbf{Baselines.}
We evaluate \texttt{VTRQ} against three baselines.




\noindent\underline{\textit{\textbf{Merkle B-Tree (MBT)}}} extends the existing ADS~\cite{pei2020efficient} to the trajectory scenario. Each trajectory point’s longitude, latitude, and timestamp are treated as separate keys, while the corresponding trajectory is stored as the attribute. Three independent Merkle B-Trees are constructed, one per dimension. During query processing, all three trees must be accessed and their results merged.

\noindent\underline{\textit{\textbf{Merkle R-Tree (MRT)}}} extends the existing ADS~\cite{yung2012authentication} to the trajectory scenario. Each trajectory point’s longitude, latitude, and timestamp serve as keys, with the trajectory stored as the attribute. A Merkle R-Tree is then built, which retains the spatial-temporal pruning capability of the R-Tree while providing verifiability through the Merkle structure.
}

\noindent\underline{\textit{\textbf{Road Network Partition Merkle Tree (RPMT)}}} adapts the state-of-the-art Road Network Partition Tree~\cite{wang2022road} to support verifiable queries. The tree is constructed following the RP-Tree rules, and each node is augmented with Merkle hash values, resulting in the Road Network Partition Merkle Tree.




{
\begin{table*}[!t]
\renewcommand\tabcolsep{1.5pt}
\centering
\caption{Comprehensive evaluation of trajectory query time (s).}

\begin{tabular}{cccccccccccccccccccccc}
   \toprule

\multicolumn{2}{c}{Dataset} & \multicolumn{5}{c}{Xi'an} & \multicolumn{5}{c}{Chengdu} &\multicolumn{5}{c}{Geolife}\\
\cmidrule(lr){1-2} \cmidrule(lr){3-7} \cmidrule(lr){8-12} \cmidrule(lr){13-17}
 \textbf{Tempora}l&\textbf{Spatial} &\textbf{1km} & \textbf{2km} & \textbf{3km} & \textbf{4km} & \textbf{5km} & \textbf{1km} & \textbf{2km} & \textbf{3km} & \textbf{4km} & \textbf{5km} & \textbf{1km} & \textbf{2km} & \textbf{3km} &\textbf{ 4km} & \textbf{5km} \\

\midrule
\multirow{2}{*}{\textbf{1min}} 
& RPMT 
& \underline{0.584} & \underline{2.043} & \underline{3.340} & \underline{4.942} & \underline{7.095}  
& \underline{6.080} & \underline{14.98} & \underline{21.82} & \underline{29.19} & \underline{38.83}  

& \underline{0.406} & \underline{1.195} & \underline{2.423} & \underline{2.894} & \underline{3.304} 
 \\

& VTRQ  
& \textbf{0.190} & \textbf{0.507} & \textbf{0.950} & \textbf{1.339} & \textbf{1.847} 
& \textbf{1.048} & \textbf{3.617} & \textbf{5.335} & \textbf{7.435} & \textbf{10.64} 

& \textbf{0.098} & \textbf{0.247} & \textbf{0.547} & \textbf{0.743} & \textbf{0.979}  \\




\midrule
\multirow{2}{*}{\textbf{30min}} 
& RPMT 
& \underline{0.614} & \underline{1.943} & \underline{3.305} & \underline{5.247} & \underline{6.942}
& \underline{5.860} & \underline{14.81} & \underline{21.24} & \underline{28.99} & \underline{40.31}  
  
& \underline{0.403} & \underline{1.089} & \underline{2.274} & \underline{3.359} & \underline{3.560} 
\\

& VTRQ  
& \textbf{0.197} & \textbf{0.525} & \textbf{0.914} & \textbf{1.327} & \textbf{1.901} 
& \textbf{1.058} & \textbf{3.599} & \textbf{5.202} & \textbf{7.581} & \textbf{10.87} 

& \textbf{0.101} & \textbf{0.223} & \textbf{0.553} & \textbf{0.755} & \textbf{0.979} 
 \\

\midrule
\multirow{2}{*}{\textbf{1h}} 
& RPMT 
& \underline{0.606} & \underline{2.175} & \underline{3.785} & \underline{4.950} & \underline{7.151} 
& \underline{6.150} & \underline{15.31} & \underline{21.02} & \underline{29.32} & \underline{40.10}  
 
& \underline{0.381} & \underline{1.028} & \underline{2.262} & \underline{2.994} & \underline{3.356} 
 \\

& VTRQ  
& \textbf{0.208} & \textbf{0.531} & \textbf{1.031} & \textbf{1.338} & \textbf{1.934} 
& \textbf{1.068} & \textbf{3.513} & \textbf{5.328} & \textbf{7.427} & \textbf{10.33} 

& \textbf{0.101} & \textbf{0.238} & \textbf{0.581} & \textbf{0.755} & \textbf{1.011} \\

 



\midrule
\multirow{2}{*}{\textbf{4h}} 
& RPMT 
& \underline{0.632} & \underline{1.954} & \underline{3.652} & \underline{5.067} & \underline{7.247}  
& \underline{5.977} & \underline{15.01} & \underline{21.08} & \underline{29.76} & \underline{39.41}  

& \underline{0.427} & \underline{1.071} & \underline{2.332} & \underline{3.200} & \underline{3.379} \\

& VTRQ  
& \textbf{0.265} & \textbf{0.597} & \textbf{1.007} & \textbf{1.432} & \textbf{1.959} 
& \textbf{1.112} & \textbf{3.520} & \textbf{5.298} & \textbf{7.608} & \textbf{10.81} 

& \textbf{0.104} & \textbf{0.239} & \textbf{0.548} & \textbf{0.768} & \textbf{0.970} \\

\midrule
\multirow{2}{*}{\textbf{12h}} 
& RPMT 
& \underline{0.667} & \underline{2.201} & \underline{3.702} & \underline{5.230} & \underline{7.227}  
& \underline{5.992} & \underline{15.00} & \underline{21.39} & \underline{30.02} & \underline{40.59}  

& \underline{0.507} & \underline{1.101} & \underline{2.518} & \underline{3.040} & \underline{3.431} \\

& VTRQ 
& \textbf{0.385} & \textbf{0.710} & \textbf{1.140} & \textbf{1.553} & \textbf{2.132} 
& \textbf{1.333} & \textbf{3.889} & \textbf{5.405} & \textbf{7.780} & \textbf{11.13} 

& \textbf{0.107} & \textbf{0.271} & \textbf{0.571} & \textbf{0.851} & \textbf{0.970} \\

\midrule
\multirow{2}{*}{\textbf{24h}} 
& RPMT 
& \underline{0.751} & \underline{2.425} & \underline{3.884} & \underline{5.412} & \underline{7.545} 
& \underline{6.158} & \underline{15.68} & \underline{21.48} & \underline{30.00} & \underline{42.26}  
 
& \underline{0.390} & \underline{1.057} & \underline{2.307} & \underline{3.163} & \underline{3.366} \\

& VTRQ  
& \textbf{0.517} & \textbf{0.884} & \textbf{1.302} & \textbf{1.736} & \textbf{2.267} 
& \textbf{1.858} & \textbf{4.319} & \textbf{6.087} & \textbf{8.237} & \textbf{11.84} 
& \textbf{0.124} & \textbf{0.296} & \textbf{0.588} & \textbf{0.793} & \textbf{0.952} \\

\bottomrule
\end{tabular}
\label{tab:Query on trajectory index}

\end{table*}}

\subsection{Query performance}

We first evaluate all methods on smaller subsets of the datasets because MBT and MRT are impractical on the full-scale datasets, particularly the Xi’an and Chengdu datasets. Specifically, we use three days of data from Xi’an, one day of data from Chengdu, and 30\% of the Geolife dataset. The spatial query range is selected from $\{1\textit{km}, 2\textit{km}, 3\textit{km}\}$, and the temporal range is selected from $\{1\textit{min}, 5\textit{min}, 10\textit{min}\}$.
Table~\ref{tab:Query on trajectory index1} reports the trajectory query time on the three datasets.
\texttt{VTRQ} achieves the lowest query latency, outperforming the three baselines.

\begin{figure*}[]
\centering
\subfigure[Xi'an]{\includegraphics[width=0.31\textwidth]{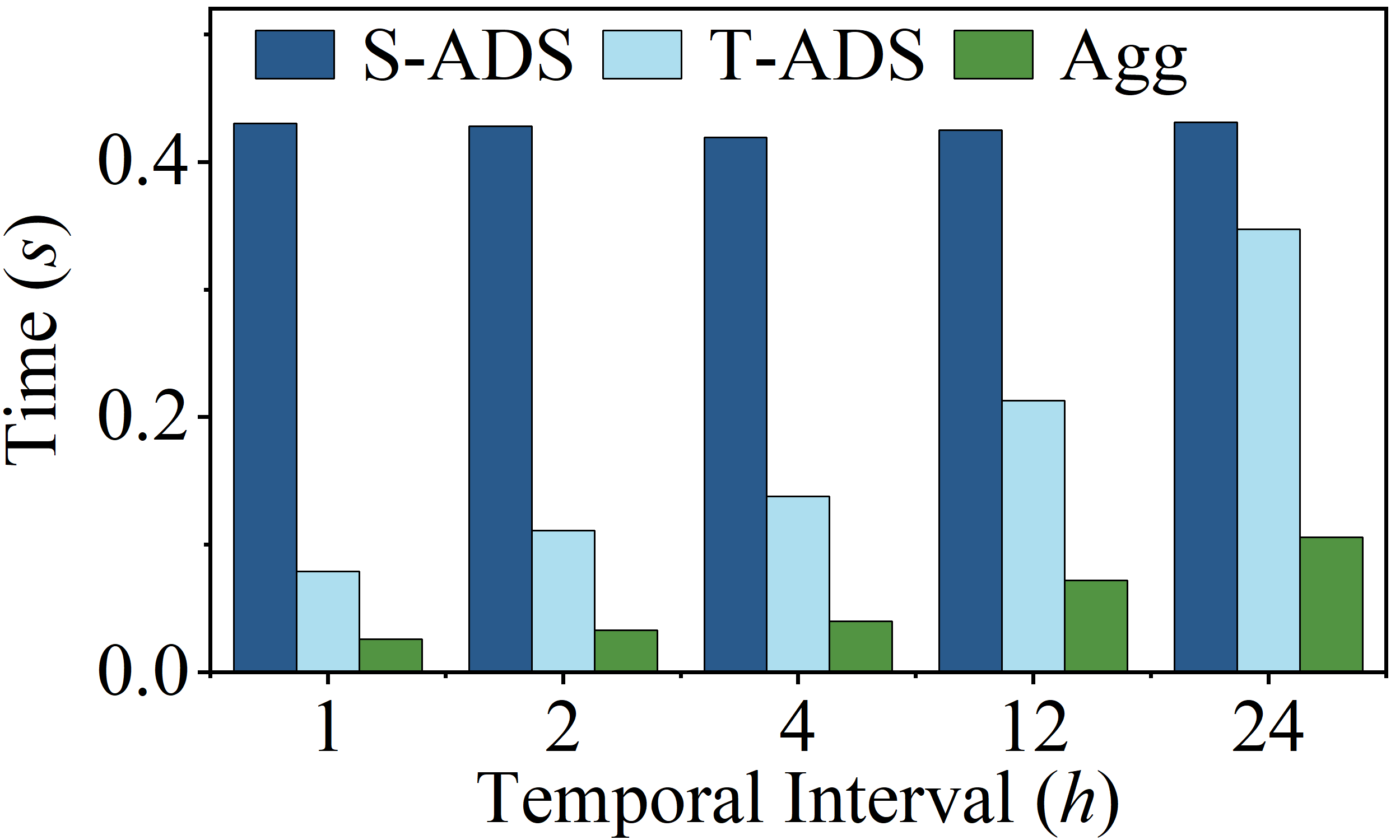}}
\hspace{0.0025\textwidth}
\subfigure[Chengdu]{\includegraphics[width=0.31\textwidth]{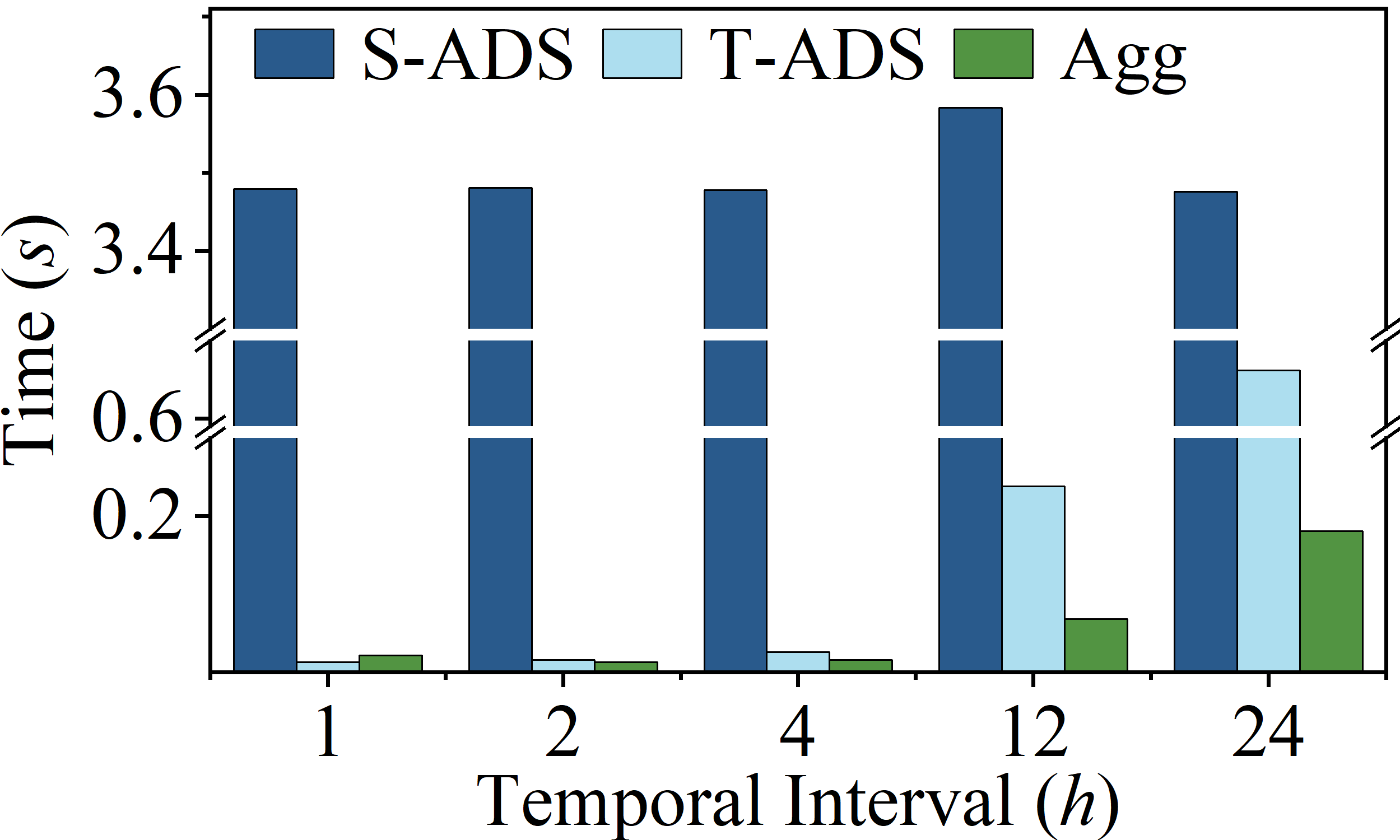}}
\hspace{0.0025\textwidth}
\subfigure[Geolife]{\includegraphics[width=0.31\textwidth]{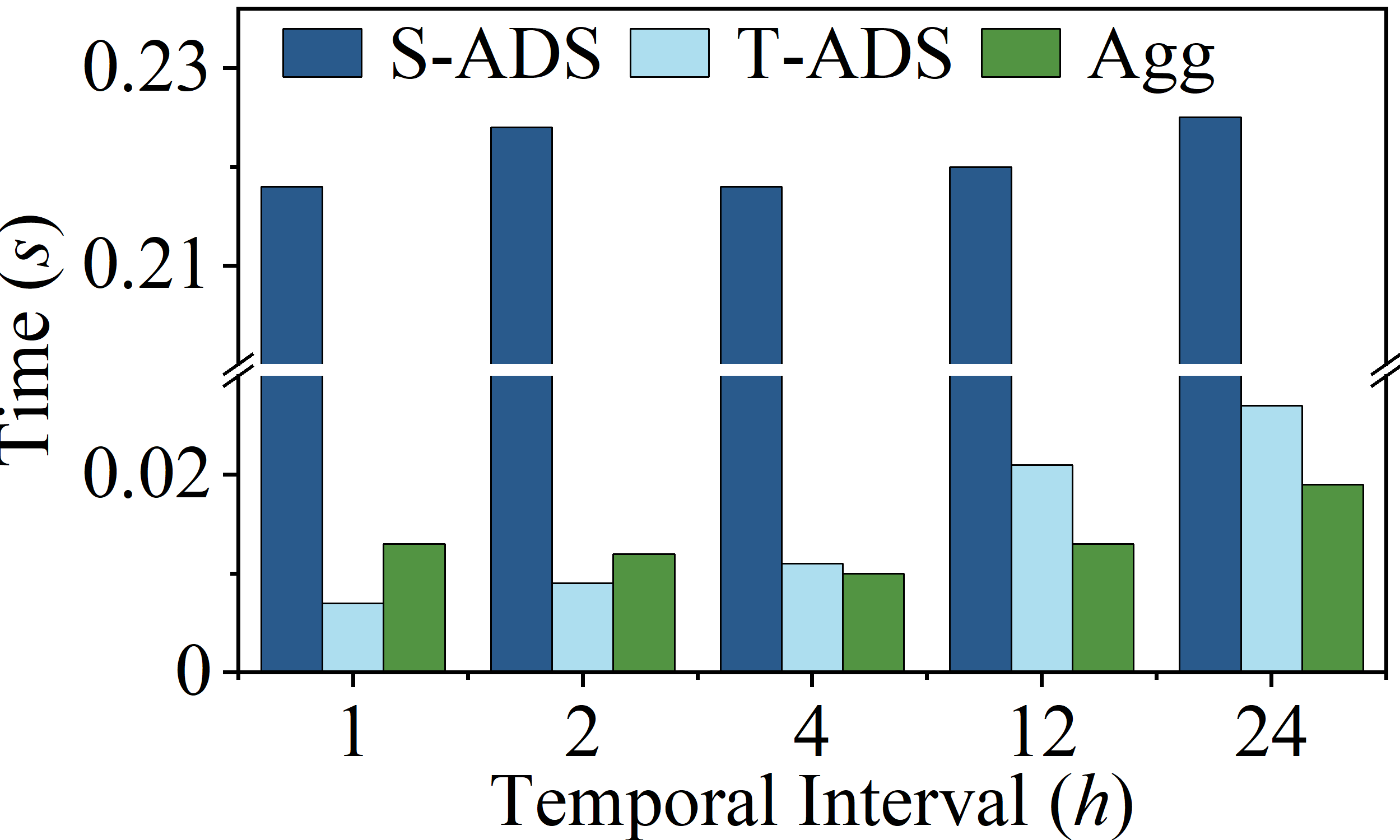}}
\caption{Impact on query time of the temporal range.}
\label{exfig:query with temporal}
\end{figure*}

\begin{figure*}[]
\centering
\subfigure[Xi'an]{\includegraphics[width=0.31\textwidth]{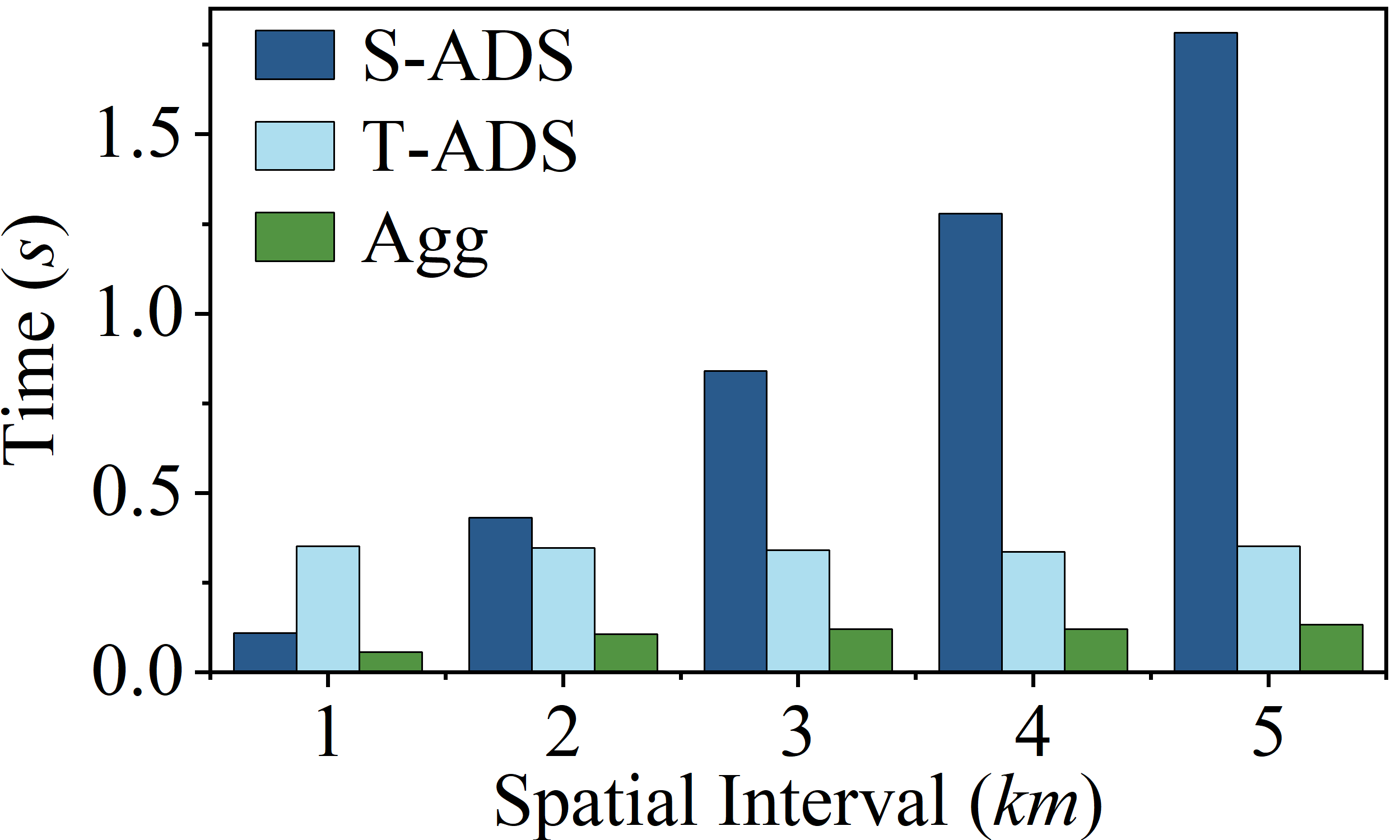}}
\hspace{0.0025\textwidth}
\subfigure[Chengdu]{\includegraphics[width=0.31\textwidth]{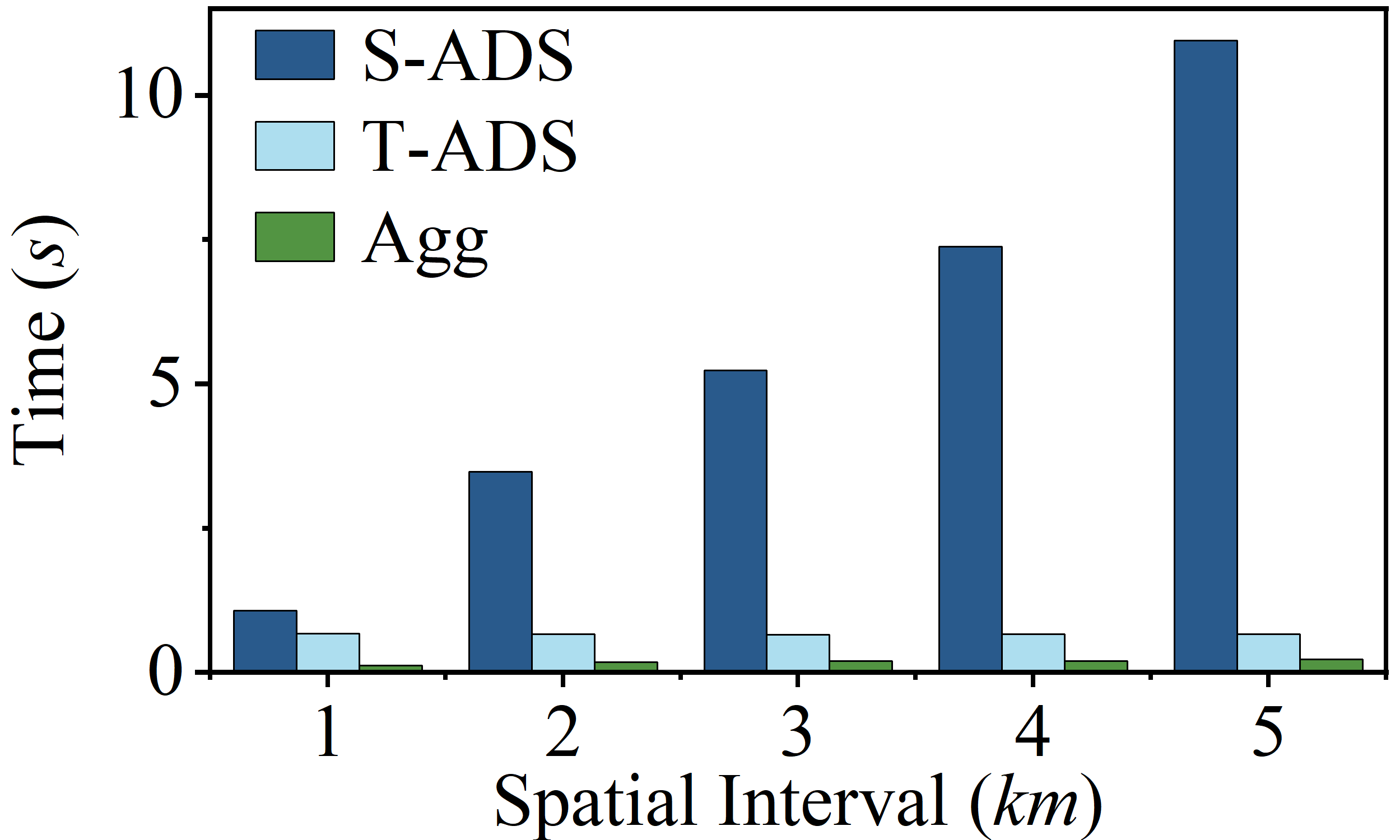}}
\hspace{0.0025\textwidth}
\subfigure[Geolife]
{\includegraphics[width=0.31\textwidth]{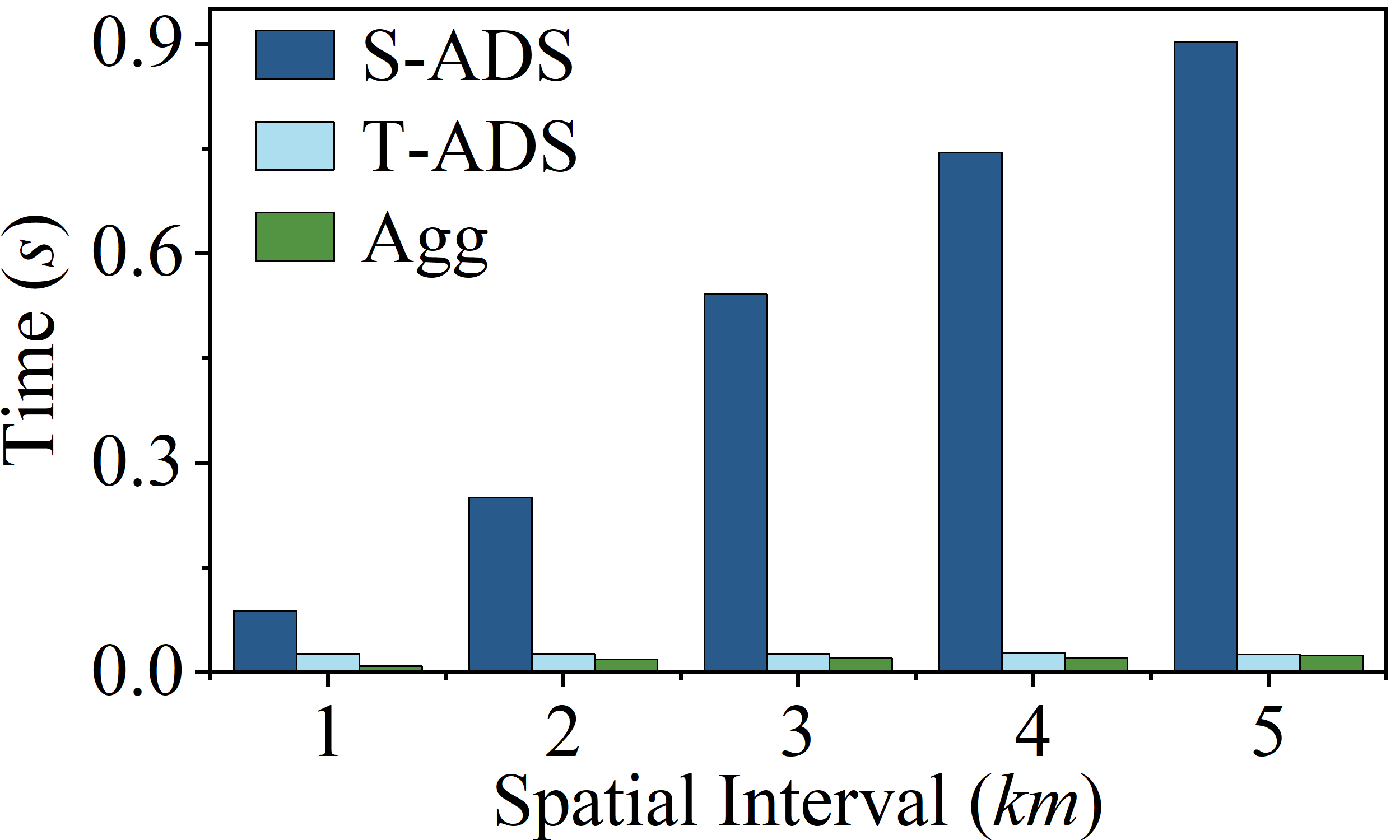}}
\caption{Impact on query time of the spatial range.}
\label{exfig: query with spatial}
\end{figure*}

{
\begin{table}[!t]
\renewcommand\tabcolsep{3pt}
\centering
\caption{STEA evaluation.}

\begin{tabular}{ccccccccccccc}
   \toprule

\multicolumn{2}{c}{Dataset} & \multicolumn{2}{c}{Xi'an} & \multicolumn{2}{c}{Chengdu} &\multicolumn{2}{c}{Geolife}\\
\cmidrule(lr){1-2} \cmidrule(lr){3-4} \cmidrule(lr){5-6} \cmidrule(lr){7-8}
\multicolumn{2}{c}{\textbf{Phase}} & \textbf{Cand.}  & \textbf{Agg.}  & \textbf{Cand.}  & \textbf{Agg.} & \textbf{Cand.}  & \textbf{Agg.} \\

\midrule

& RPMT 
 & 992 & 991
 & 1899 &  1897
 & 54 &54   \\

& VTRQ  
 & 1002 & 991
 & 1906 & 1897 
 & 55 & 54  \\

\bottomrule
\end{tabular}
\label{tab:STEA evaluation}
\vspace{-1mm}
\end{table}}

This is because:
(i) MBT constructs three independent Merkle B-Trees for longitude, latitude, and time, where each trajectory point is indexed separately. During query execution, all trees must be accessed, and their results must be  merged, causing many node accesses and high computational overhead.
(ii) MRT organizes all trajectory points in a single Merkle R-Tree. Although it preserves the spatial-temporal pruning capabilities of the R-Tree, it also indexes individual trajectory points, causing excessive node accesses.
(iii) RPMT stores both spatial and temporal information in each node. For every trajectory located on edges in the spatial query region, RPMT must check whether its timestamp range intersects with the temporal query interval, causing longer query times. In addition, since a trajectory may traverse the same edge multiple times at different times, RPMT stores duplicate trajectory IDs on the same edge to avoid omissions.
(iv) \texttt{VTRQ} separates the spatial and temporal attributes and records the start and end timestamps for each trajectory, thereby avoiding complex timestamp management. In addition, by constructing MBRs for the border edges that connect subgraphs, \texttt{VTRQ} enables efficient pre-filtering and reduces the computational overhead associated with border-edge checking.
In addition to being impractical for large datasets, MBT and MRT have a major limitation—result loss, as shown in Table~\ref{tab:result1}.
Since their indexes are point-level, trajectories that intersect a query range but have no sampled points in it, are not retrieved, leading to incomplete query results.

We report the trajectory query performance of \texttt{VTRQ} and RPMT on the complete datasets in Table~\ref{tab:Query on trajectory index}. The spatial query range is selected from $\{1\textit{km}, 2\textit{km}, 3\textit{km}, 4\textit{km}, 5\textit{km}\}$.
The temporal range is selected from $\{1\textit{min}, 30\textit{min}, 1h, 4h, 12h, 24h\}$. 
For all datasets, whether varying the spatial range with a fixed temporal size or varying the temporal range with a fixed spatial size, \texttt{VTRQ} consistently achieves lower query times than RPMT. The query efficiency of \texttt{VTRQ} can be up to six times higher than that of RPMT. The underlying reasons were explained for the small-scale evaluation.
We make three additional observations:
first, as the query time interval increases, the query time of the RPMT does not grow linearly. This is because, when the spatial query range is fixed, all trajectory information on the edges in this spatial range must still be verified against the temporal constraints, regardless of temporal range changes. 
Second, for \texttt{VTRQ}, the query time grows roughly linearly as the temporal query range expands, although the increase is modest. The reason is that while a wider temporal query range inevitably increases the workload of the temporal query process, the time consumed by the spatial query process accounts for a larger proportion of the total query time, to be discussed in more detail shortly.
Finally, across all spatio-temporal query conditions, \texttt{VTRQ} takes longer in terms of query time on the Chengdu dataset than on the Xi’an and Geolife datasets. This is because the Chengdu dataset has more trajectories, leading to larger S-ADS and T-ADS structures, requiring more computation.

{
\begin{table*}[!t]
\centering
\caption{Verification time (s) for all methods.}
\begin{tabular}{cccccccccccccccc}
   \toprule

\multicolumn{2}{c}{Dataset} & \multicolumn{3}{c}{Xi'an} & \multicolumn{3}{c}{Chengdu} &\multicolumn{3}{c}{Geolife}\\
\cmidrule(lr){1-2} \cmidrule(lr){3-5} \cmidrule(lr){6-8} \cmidrule(lr){9-11}
 \textbf{Tem}&\textbf{Spa} & \textbf{1km} & \textbf{2km} & \textbf{3km}  & \textbf{1km} & \textbf{2km} & \textbf{3km} & \textbf{1km} & \textbf{2km} & \textbf{3km} \\

\midrule
\multirow{4}{*}{\rotatebox{90}{\textbf{1 min}}}

& MBT 
& 1.043 & 2.144 & 2.676
& 1.072 & 2.500 & 3.997
& 0.492 & 0.726 & 1.257  \\

& MRT
& 2.069 & 4.918 & 10.28
& 1.977 & 7.579 & 14.34
& \textbf{0.036} & \textbf{0.043} & \textbf{0.222}  \\

& RPMT 
& \underline{0.269} & \underline{0.512} & \underline{1.078}
& \underline{0.281} & \underline{1.061} & \underline{1.790}
& 0.763 & 0.821 & 4.019 \\

& VTRQ  
& \textbf{0.105} & \textbf{0.179} & \textbf{0.270}
& \textbf{0.062} & \textbf{0.218} & \textbf{0.358}
& \underline{0.144} & \underline{0.164} & \underline{0.698}
 \\

\midrule

\multirow{4}{*}{\rotatebox{90}{\textbf{5 min}}} 
& MBT 
& 1.106 & 2.149 & 2.782
& 1.090 & 2.663 & 4.113
& 0.530 & 0.767 & 1.262  \\

& MRT
& 2.029 & 5.087 & 10.57
& 2.055 & 7.804 & 14.42
&\textbf{0.038} & \textbf{0.050} & \textbf{0.233}  \\

& RPMT 
& \underline{0.278} & \underline{0.567} & \underline{1.170}
& \underline{0.298} & \underline{1.071} & \underline{1.866}
& 0.799 & 0.864 & 4.037
  \\

& VTRQ  
& \textbf{0.114} & \textbf{0.188} & \textbf{0.274}
& \textbf{0.076} & \textbf{0.223} & \textbf{0.361}
& \underline{0.150} & \underline{0.180} & \underline{0.712} \\

\midrule
\multirow{4}{*}{\rotatebox{90}{\textbf{10 min}}} 
& MBT 
& 1.107 & 2.313 & 2.858 
& 1.117 & 2.840 & 4.290 
& 0.562 & 0.775 & 1.294  \\

& MRT
& 2.135 & 5.244 & 10.89 
& 2.181 & 8.255 & 14.66 
& \textbf{0.040} & \textbf{0.054} & \textbf{0.256}  \\

& RPMT 
& \underline{0.285} & \underline{0.582} & \underline{1.256} 
& \underline{0.306} & \underline{1.088} & \underline{1.880} 
& 0.808 & 0.894 & 4.169 
 \\

& VTRQ  
& \textbf{0.114} & \textbf{0.206} & \textbf{0.271} 
& \textbf{0.079} & \textbf{0.228} & \textbf{0.383} 
& \underline{0.160} & \underline{0.187} & \underline{0.738}  \\

\bottomrule
\end{tabular}
\label{tab:tab:Verification Performance1}
\end{table*}}

Fig.~\ref{exfig:query with temporal} shows the composition of query time for \texttt{VTRQ} on the three datasets in a spatial range of 2 km as the temporal query interval varies. The query time is consumed by three activities: spatial query, temporal query, and result aggregation. The results show that spatial queries consume most of the query time on all three datasets. This is because the spatial query involves three steps: (i) locating the matching subgraph; (ii) retrieving the edges in the subgraph that are in the spatial query range; and (iii) incorporating the qualified trajectories on these matched edges into the query result set. In contrast, the temporal query requires a simple pruning operation on the tree structure.
Furthermore, as the temporal query range increases, the proportion of the spatial query time decreases steadily, while the proportions of temporal query time and result aggregation time increase. This is because when the temporal range expands, the number of trajectories satisfying the temporal conditions in the T-ADS grows significantly. Next, as the temporal query range expands, the number of trajectories satisfying both the spatial and temporal queries also increases, which in turn further increases the cost of final result aggregation.
Fig.~\ref{exfig: query with spatial} shows the composition of query time for \texttt{VTRQ} within 24 hours while varying the spatial range. The proportion of the spatial query time increases gradually, whereas the proportion of the temporal query time decreases. The reason is that as the spatial query range expands, the number of covered trajectories increases, resulting in more trajectories being selected as candidate trajectories. 

{
\begin{table*}[!t]
\renewcommand\tabcolsep{1.5pt}
\centering
\caption{Comprehensive evaluation of verification time (s).}

\begin{tabular}{cccccccccccccccccccccc}
   \toprule

\multicolumn{2}{c}{Dataset} & \multicolumn{5}{c}{Xi'an} & \multicolumn{5}{c}{Chengdu} &\multicolumn{5}{c}{Geolife}\\
\cmidrule(lr){1-2} \cmidrule(lr){3-7} \cmidrule(lr){8-12} \cmidrule(lr){13-17}
 \textbf{Tempora}l&\textbf{Spatial} &\textbf{1km} & \textbf{2km} & \textbf{3km} & \textbf{4km} & \textbf{5km} & \textbf{1km} & \textbf{2km} & \textbf{3km} & \textbf{4km} & \textbf{5km} & \textbf{1km} & \textbf{2km} & \textbf{3km} &\textbf{ 4km} & \textbf{5km} \\

\midrule
\multirow{2}{*}{\textbf{1min}} 
& RPMT 
& \underline{1.944} & \underline{7.631} & \underline{14.37} & \underline{20.61} & \underline{28.24}  
& \underline{21.43} & \underline{64.99} & \underline{96.23} & \underline{133.9} & \underline{183.1}   

& \underline{1.196} & \underline{3.786} & \underline{10.21} & \underline{12.55} & \underline{14.65}  \\

& VTRQ 
& \textbf{0.208} & \textbf{0.717} & \textbf{1.410} & \textbf{2.017} & \textbf{2.975}
& \textbf{2.243} & \textbf{7.587} & \textbf{10.88} & \textbf{16.00} & \textbf{22.13}
& \textbf{0.151} & \textbf{0.414} & \textbf{0.835} & \textbf{1.047} & \textbf{1.334}   \\



\midrule
\multirow{2}{*}{\textbf{30min}} 
& RPMT 
& \underline{1.962} & \underline{7.256} & \underline{13.79} & \underline{20.79} & \underline{29.82} 
& \underline{21.54} & \underline{65.37} & \underline{93.28} & \underline{134.5} & \underline{186.9}   
& \underline{1.149} & \underline{3.814} & \underline{9.442} & \underline{12.80} & \underline{14.84} \\

& VTRQ  
& \textbf{0.224} & \textbf{0.768} & \textbf{1.394} & \textbf{2.074} & \textbf{3.099}
& \textbf{2.276} & \textbf{7.686} & \textbf{11.05} & \textbf{15.34} & \textbf{21.84}
& \textbf{0.156} & \textbf{0.365} & \textbf{0.804} & \textbf{1.060} & \textbf{1.380}  \\

\midrule
\multirow{2}{*}{\textbf{1h}} 
& RPMT 
& \underline{1.970} & \underline{7.442} & \underline{14.37} & \underline{20.76} & \underline{29.12} 
& \underline{20.97} & \underline{66.32} & \underline{93.77} & \underline{135.0} & \underline{185.1}   
& \underline{1.158} & \underline{3.944} & \underline{9.354} & \underline{12.23} & \underline{14.96}  \\

& VTRQ  
& \textbf{0.262} & \textbf{0.804} & \textbf{1.483} & \textbf{2.133} & \textbf{3.061}
& \textbf{2.253} & \textbf{7.505} & \textbf{11.19} & \textbf{15.45} & \textbf{21.52}
& \textbf{0.161} & \textbf{0.354} & \textbf{0.859} & \textbf{1.061} & \textbf{1.342} \\



\midrule
\multirow{2}{*}{\textbf{4h}} 
& RPMT 
& \underline{1.916} & \underline{7.461} & \underline{14.27} & \underline{20.30} & \underline{28.73} 
& \underline{21.03} & \underline{66.76} & \underline{94.41} & \underline{131.3} & \underline{185.1}   
& \underline{1.475} & \underline{4.357} & \underline{9.303} & \underline{12.82} & \underline{14.98}  \\

& VTRQ 
& \textbf{0.371} & \textbf{0.891} & \textbf{1.602} & \textbf{2.267} & \textbf{3.396}
& \textbf{2.285} & \textbf{7.693} & \textbf{10.76} & \textbf{15.78} & \textbf{22.06}
& \textbf{0.166} & \textbf{0.358} & \textbf{0.830} & \textbf{1.100} & \textbf{1.388} \\

\midrule
\multirow{2}{*}{\textbf{12h}} 
& RPMT 
& \underline{1.891} & \underline{7.316} & \underline{14.35} & \underline{20.56} & \underline{28.71}  
& \underline{20.92} & \underline{64.09} & \underline{93.55} & \underline{136.0} & \underline{183.8}   
& \underline{1.207} & \underline{4.156} & \underline{9.811} & \underline{12.52} & \underline{15.02}  \\

& VTRQ  
& \textbf{0.558} & \textbf{1.083} & \textbf{1.863} & \textbf{2.484} & \textbf{3.144}
& \textbf{2.811} & \textbf{7.959} & \textbf{11.66} & \textbf{16.03} & \textbf{22.69}
& \textbf{0.190} & \textbf{0.389} & \textbf{0.819} & \textbf{1.059} & \textbf{1.359} \\

\midrule
\multirow{2}{*}{\textbf{24h}} 
& RPMT 
& \underline{2.255} & \underline{7.347} & \underline{14.93} & \underline{20.59} & \underline{29.08}
& \underline{20.79} & \underline{66.51} & \underline{94.24} & \underline{132.6} & \underline{181.8}   
& \underline{1.126} & \underline{4.087} & \underline{9.924} & \underline{12.56} & \underline{14.80}  \\

& VTRQ  
& \textbf{0.849} & \textbf{1.364} & \textbf{2.059} & \textbf{2.653} & \textbf{3.461}
& \textbf{4.027} & \textbf{8.908} & \textbf{12.54} & \textbf{17.02} & \textbf{23.48}
& \textbf{0.204} & \textbf{0.639} & \textbf{0.868} & \textbf{1.134} & \textbf{1.411} \\

\bottomrule
\end{tabular}
\label{tab:Verification Performance}

\end{table*}}

We evaluate the effectiveness of the STEA algorithm using a spatial range of 5 km and a temporal range of 1 day. Table~\ref{tab:STEA evaluation} presents the results after candidate selection. In \texttt{VTRQ}, the candidate set is obtained by intersecting the spatial and temporal conditions, while in RPMT, it comes directly from the query results. After aggregation, we observe that although RPMT produces a smaller candidate set, it still contains trajectories that do not satisfy the query conditions. Therefore, RPMT requires the same aggregation procedure as \texttt{VTRQ}. As discussed in Section~\ref{sec:aggregation}, the reason is that while some trajectory edges satisfy the spatial and temporal constraints, the constraints are not satisfied simultaneously.

\subsection{Verification performance}

We evaluate the verification performance on smaller subsets of the datasets, using the same settings as when evaluating the query performance.
Table~\ref{tab:tab:Verification Performance1} reports the verification time on the three subsets. On both the Chengdu and Xi’an subsets, \texttt{VTRQ} achieves lower verification times than RPMT.
This is because:
(i) MBT and MRT construct indexes at the point level, resulting in excessive verification overhead due to hash computations across numerous trajectory points.
(ii) RPMT embeds temporal information in its spatial structure. During the verification phase, the hash aggregation process must traverse large volumes of temporal data.
(iii) \texttt{VTRQ} verifies the spatial and temporal components of trajectories separately. The S-ADS does not store temporal information, substantially reducing the verification data size, while the T-ADS retains only the start and end timestamps of each trajectory. 
We note that on the Geolife dataset, the verification time of \texttt{VTRQ} exceeds that of RPMT. This is because Geolife contains fewer trajectories but a much larger road network. As shown in Table~\ref{tab:STEA evaluation}, the number of trajectories in the same query range is small, which reduces the benefit of \texttt{VTRQ}’s spatio-temporal separation strategy.

\begin{figure*}[t]
\centering
\subfigure[Xi'an]{\includegraphics[width=0.31\textwidth]{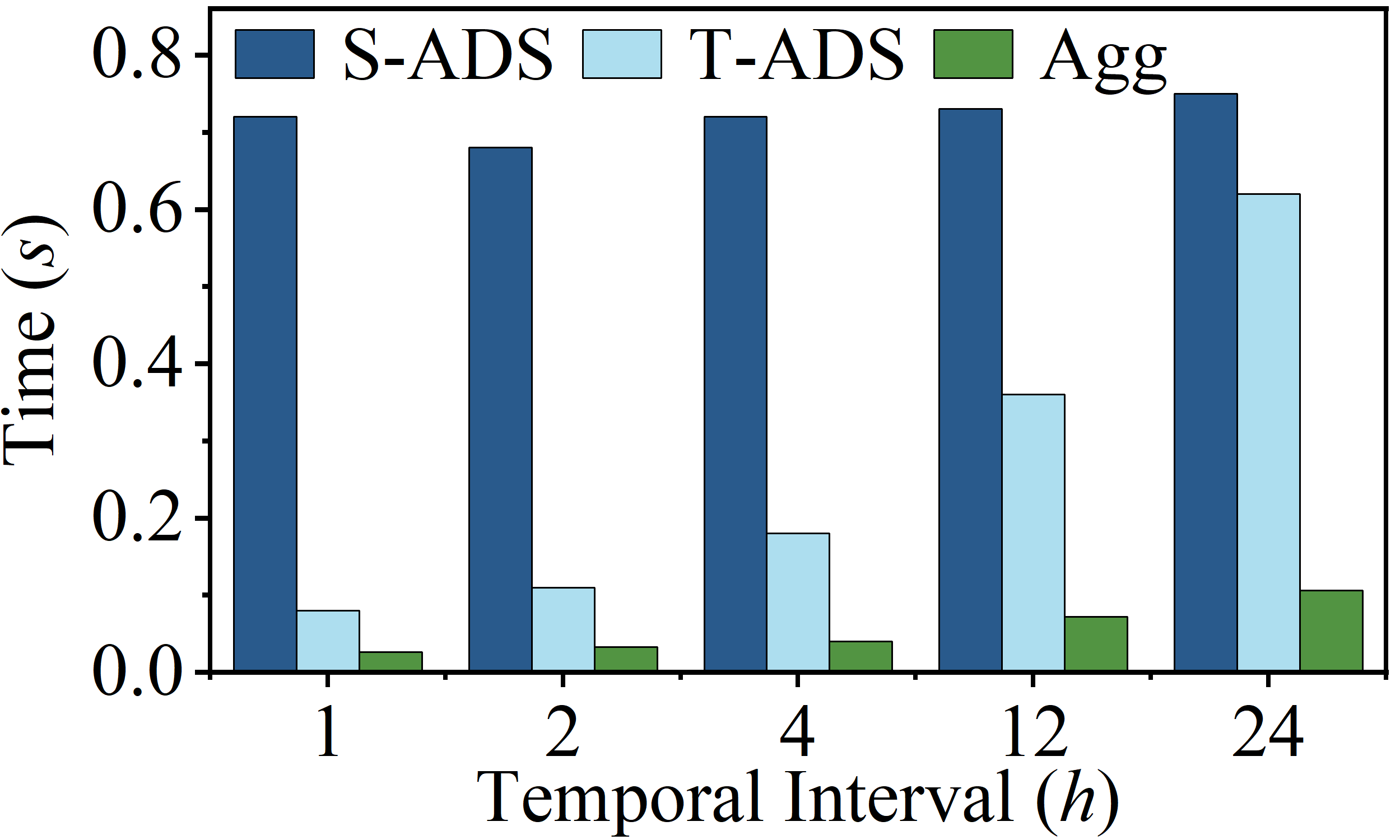}}
\hspace{0.0025\textwidth}
\subfigure[Chengdu]{\includegraphics[width=0.31\textwidth]{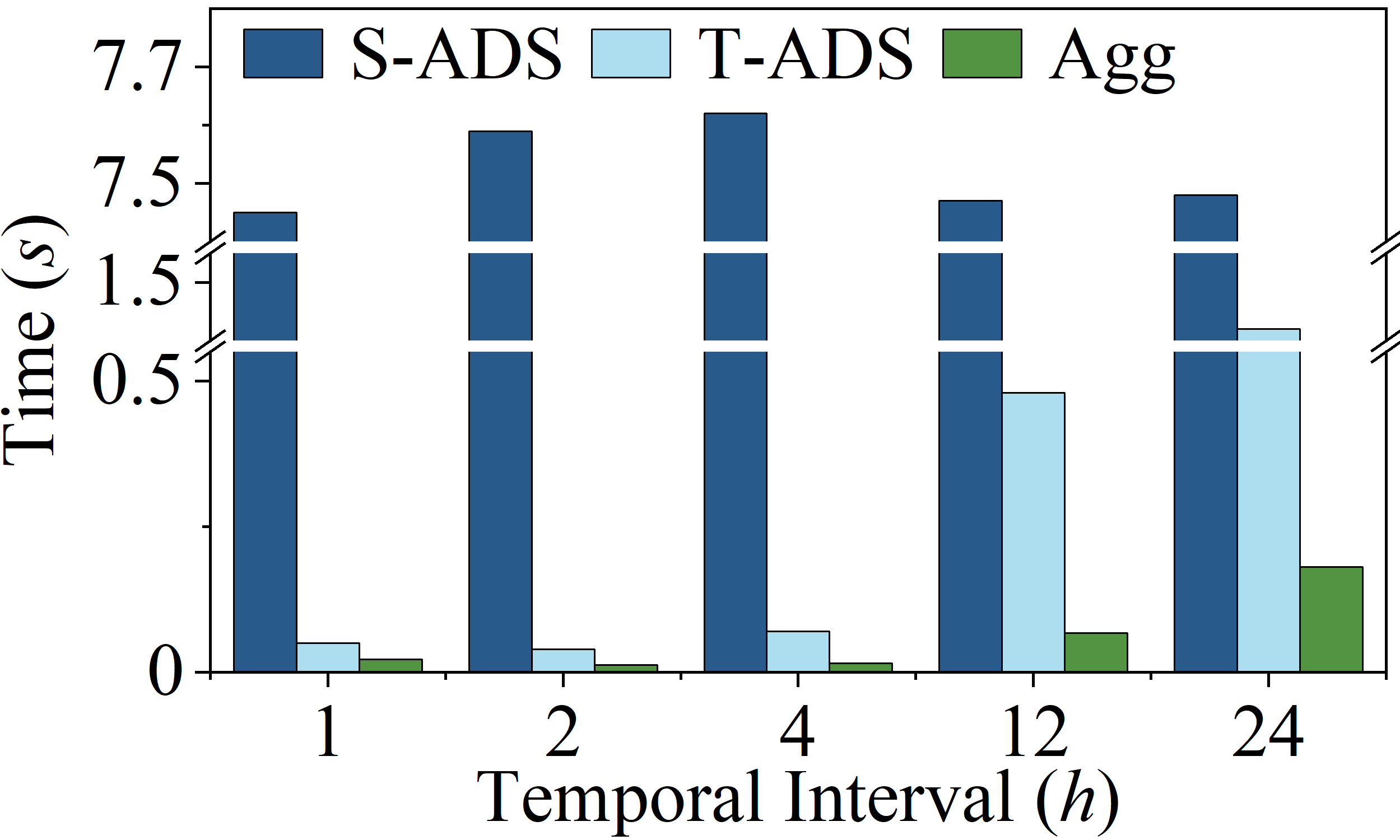}}
\hspace{0.0025\textwidth}
\subfigure[Geolife]{\includegraphics[width=0.31\textwidth]{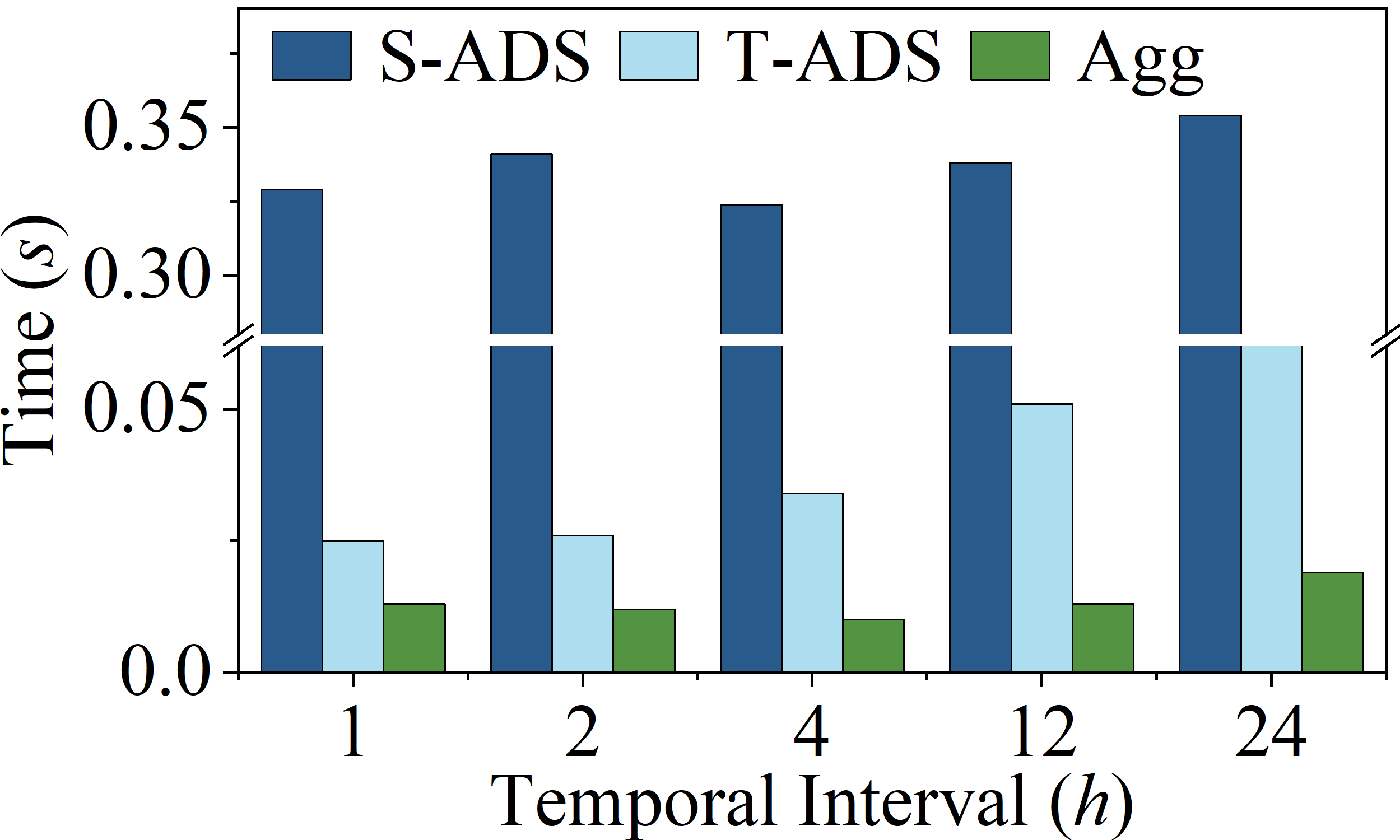}}
\caption{Impact on  verification time of the temporal range.}
\label{verification with temporal}
\end{figure*}

\begin{figure*}[t]
\centering
\subfigure[Xi'an]{\includegraphics[width=0.31\textwidth]{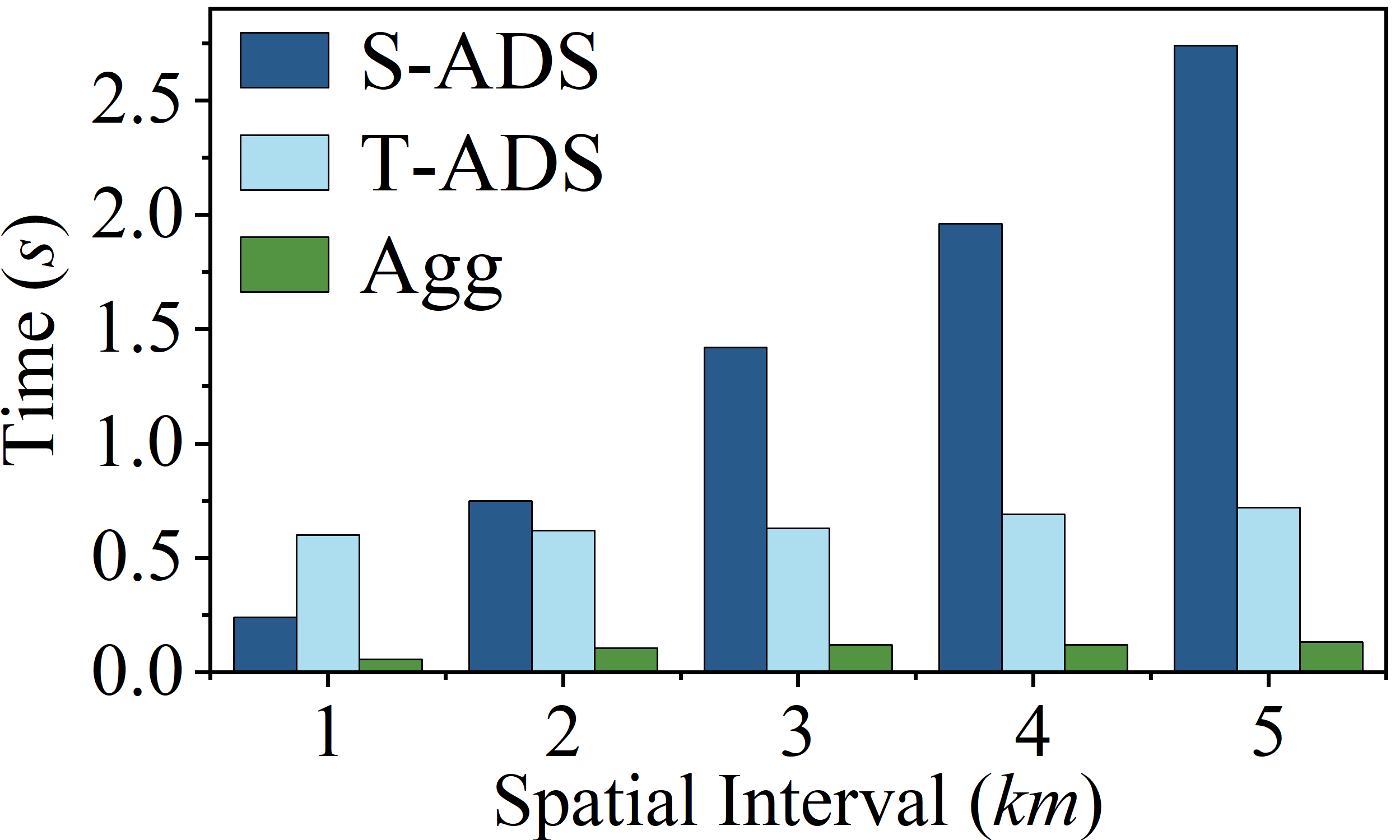}}
\hspace{0.0025\textwidth}
\subfigure[Chengdu]{\includegraphics[width=0.31\textwidth]{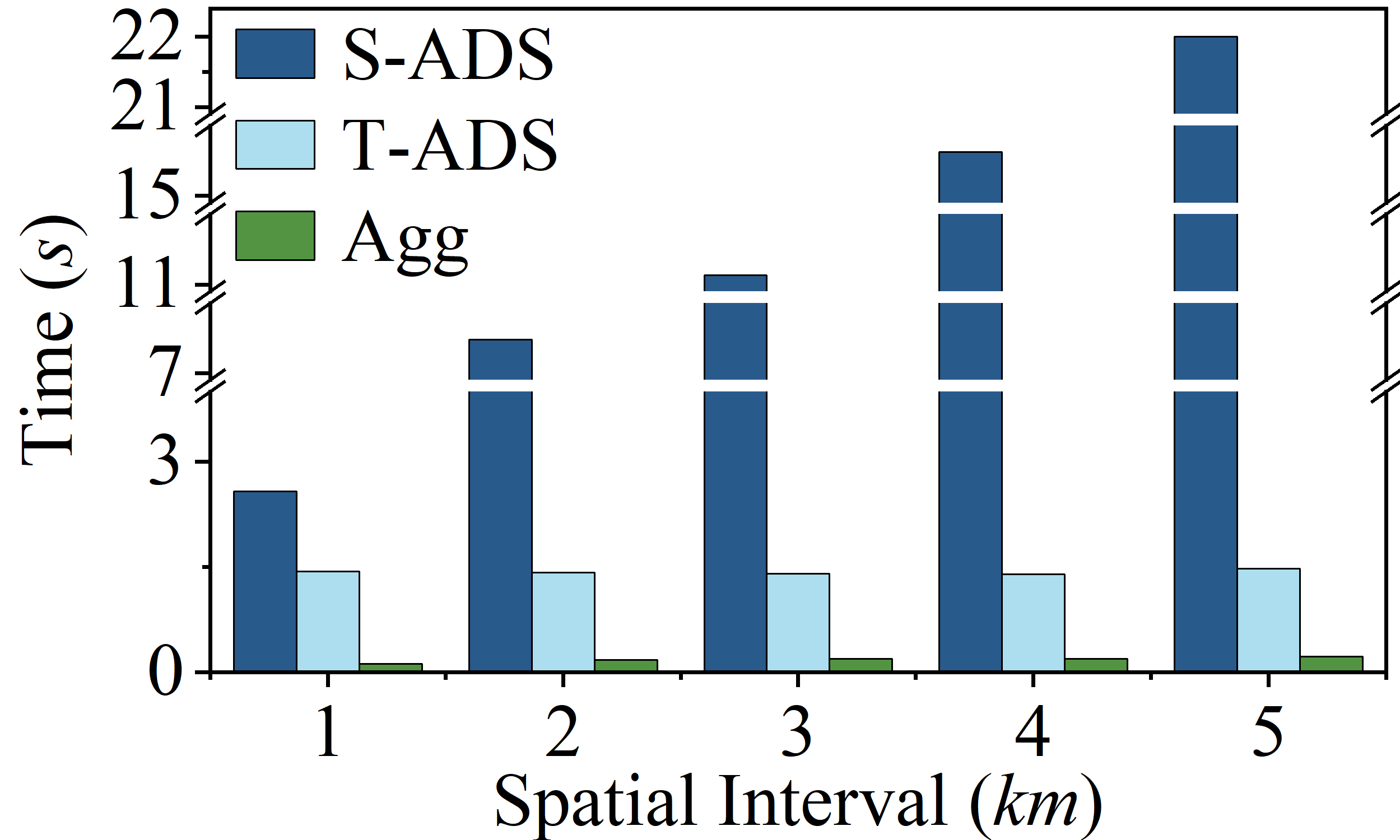}}
\hspace{0.0025\textwidth}
\subfigure[Geolife]{\includegraphics[width=0.31\textwidth]{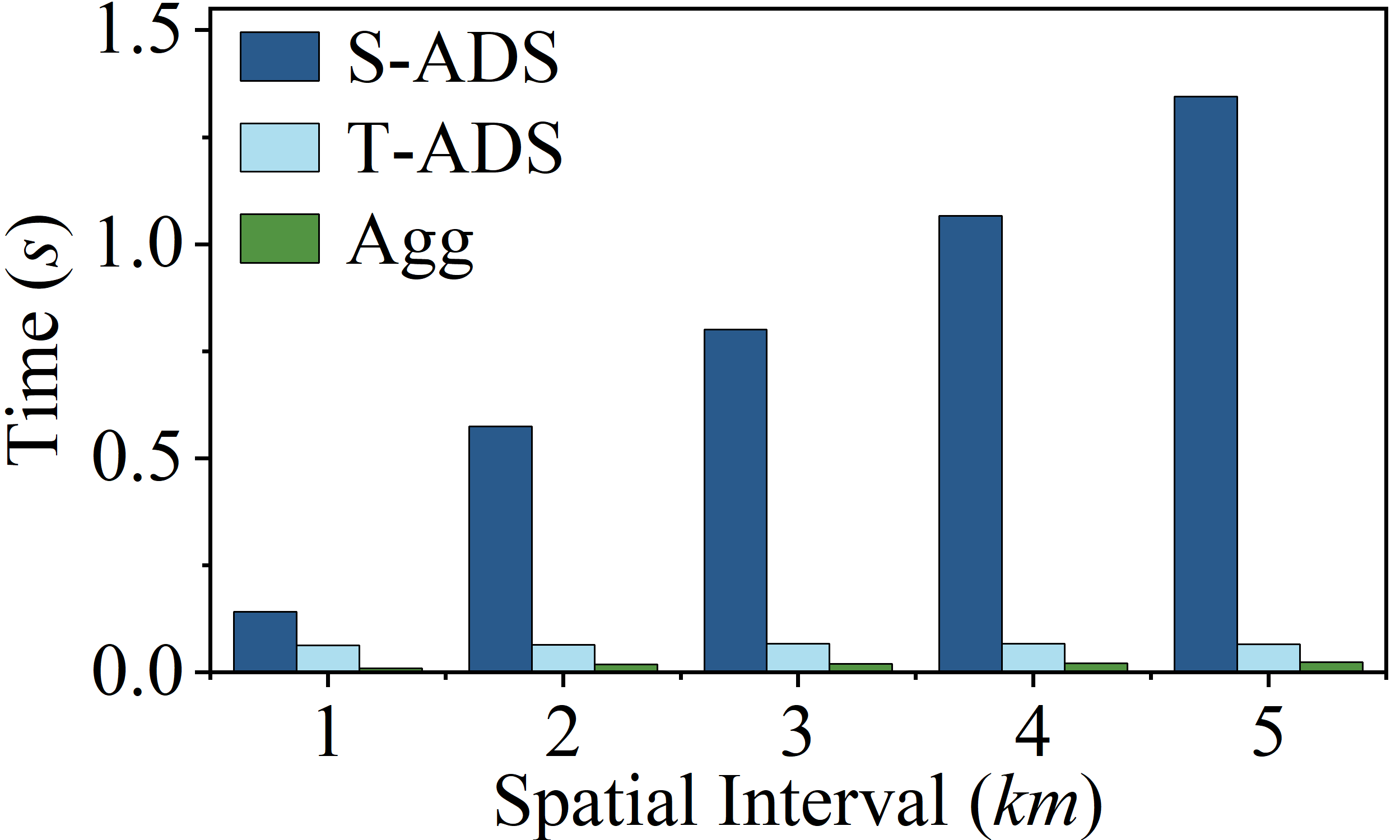}}
\caption{Impact on verification time of the spatial range.}
\label{verification with spatial}
\vspace{-1mm}
\end{figure*}

We evaluate the verification performance of \texttt{VTRQ} and RPMT on the three complete datasets under the same settings as in the query performance evaluation, as shown in Table~\ref{tab:Verification Performance}.
For all datasets, \texttt{VTRQ} shows lower verification time than RPMT. In most cases, its verification efficiency exceeds that of the baseline method by approximately an order of magnitude. 
The reasons for this have been analyzed in the small-scale experiments.
In addition, we make two observations:
first, when the spatial query range is fixed and the temporal query range is expanded, or vice versa, the verification time grows roughly linearly. However, this growth is more pronounced when the spatial range increases. This is because the verification of the T-ADS is relatively straightforward and does not involve nested structures, whereas S-ADS verification is more complex—a phenomenon that will be further analyzed in {\color{black}Fig.~\ref{verification with temporal} and Fig.~\ref{verification with spatial}.} 
{Second, \texttt{VTRQ}'s verification times on Xi’an and Geolife are shorter than those on Chengdu because the Xi’an and Geolife datasets contain fewer trajectories.}

\begin{figure*}[t]
\centering
\subfigure[Xi'an]{\includegraphics[width=0.31\textwidth]{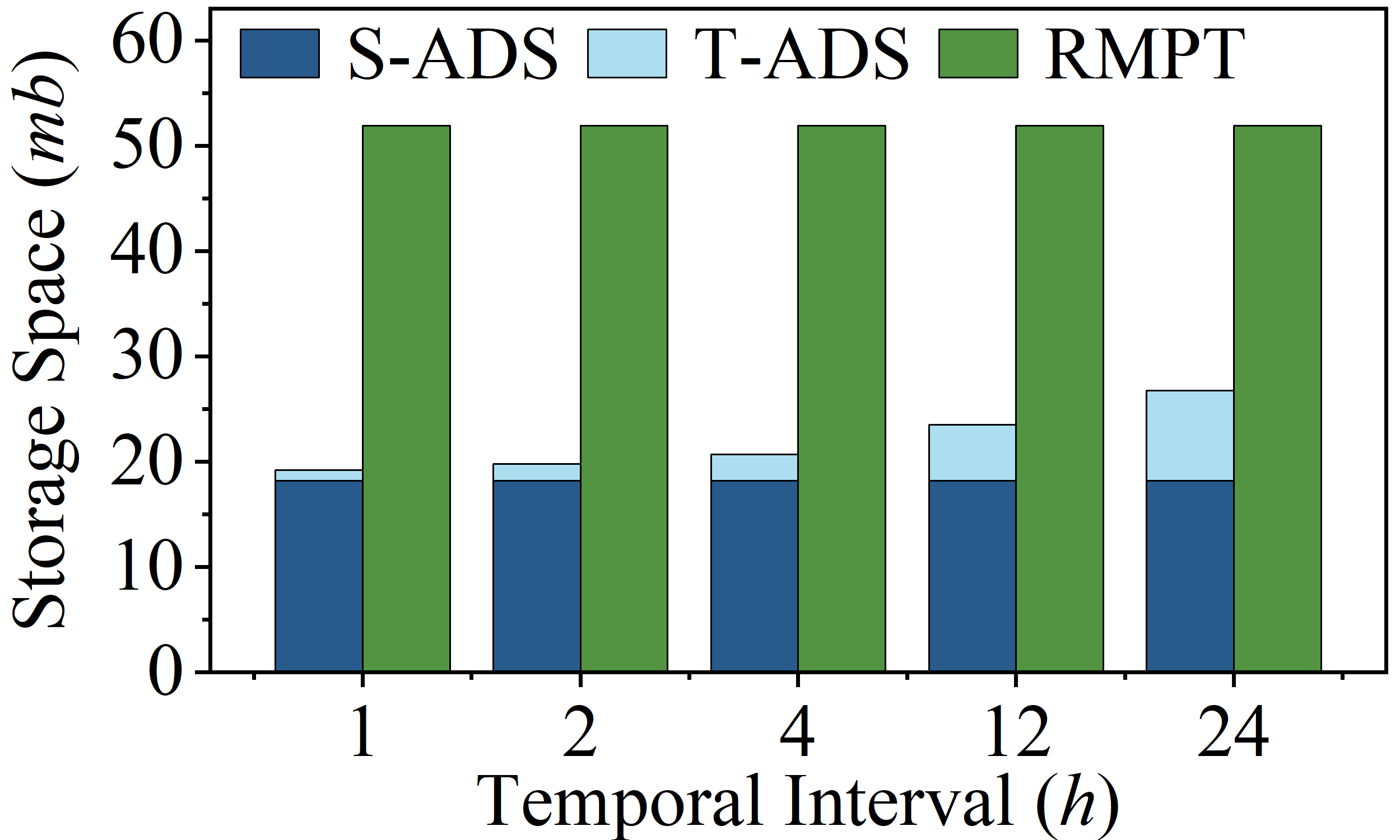}}
\hspace{0.0025\textwidth}
\subfigure[Chengdu]{\includegraphics[width=0.31\textwidth]{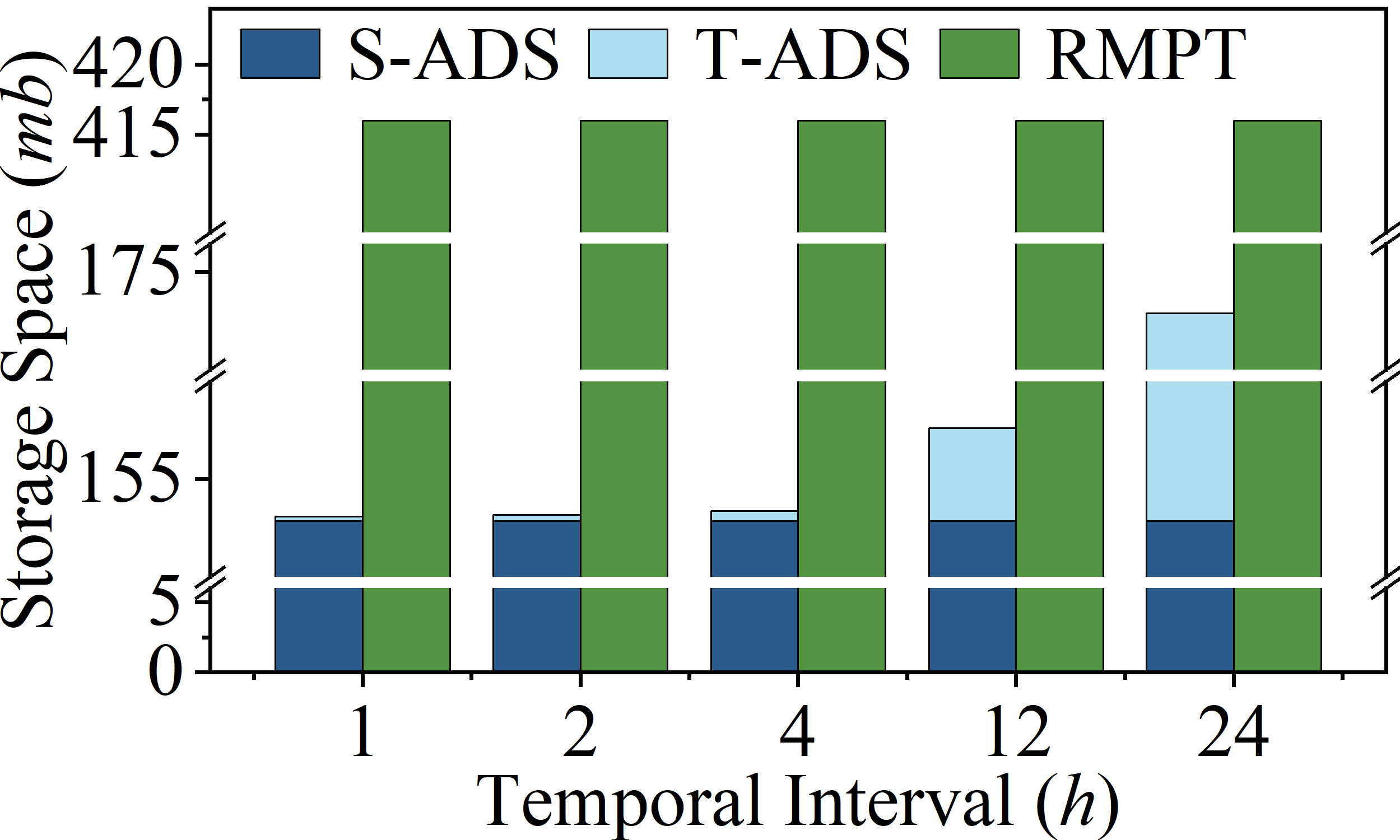}}
\hspace{0.0025\textwidth}
\subfigure[Geolife]
{\includegraphics[width=0.31\textwidth]{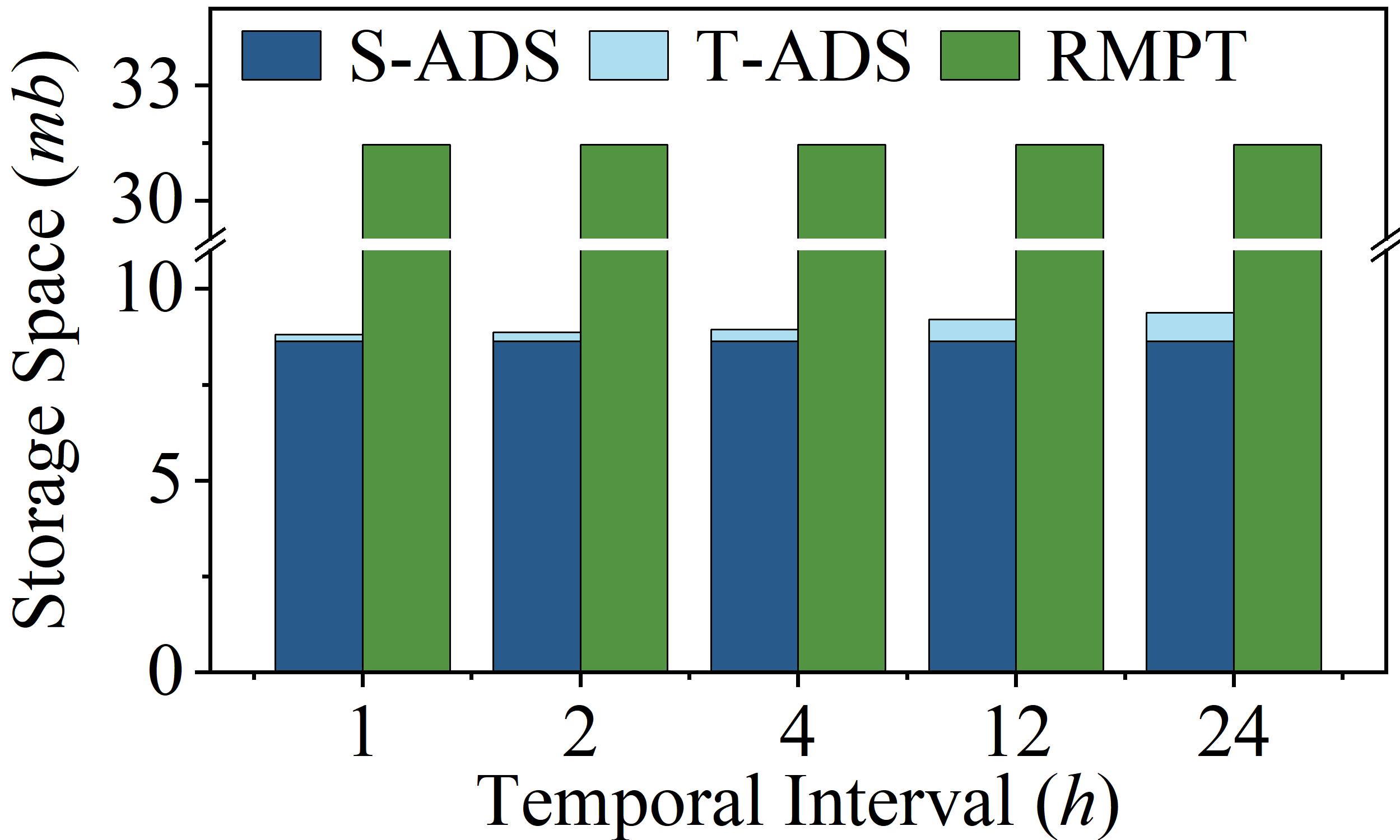}}
\caption{Impact on VO size of the temporal range.}
\label{VO size with tempora}
\end{figure*}

\begin{figure*}[t]
\centering
\subfigure[Xi'an]{\includegraphics[width=0.31\textwidth]{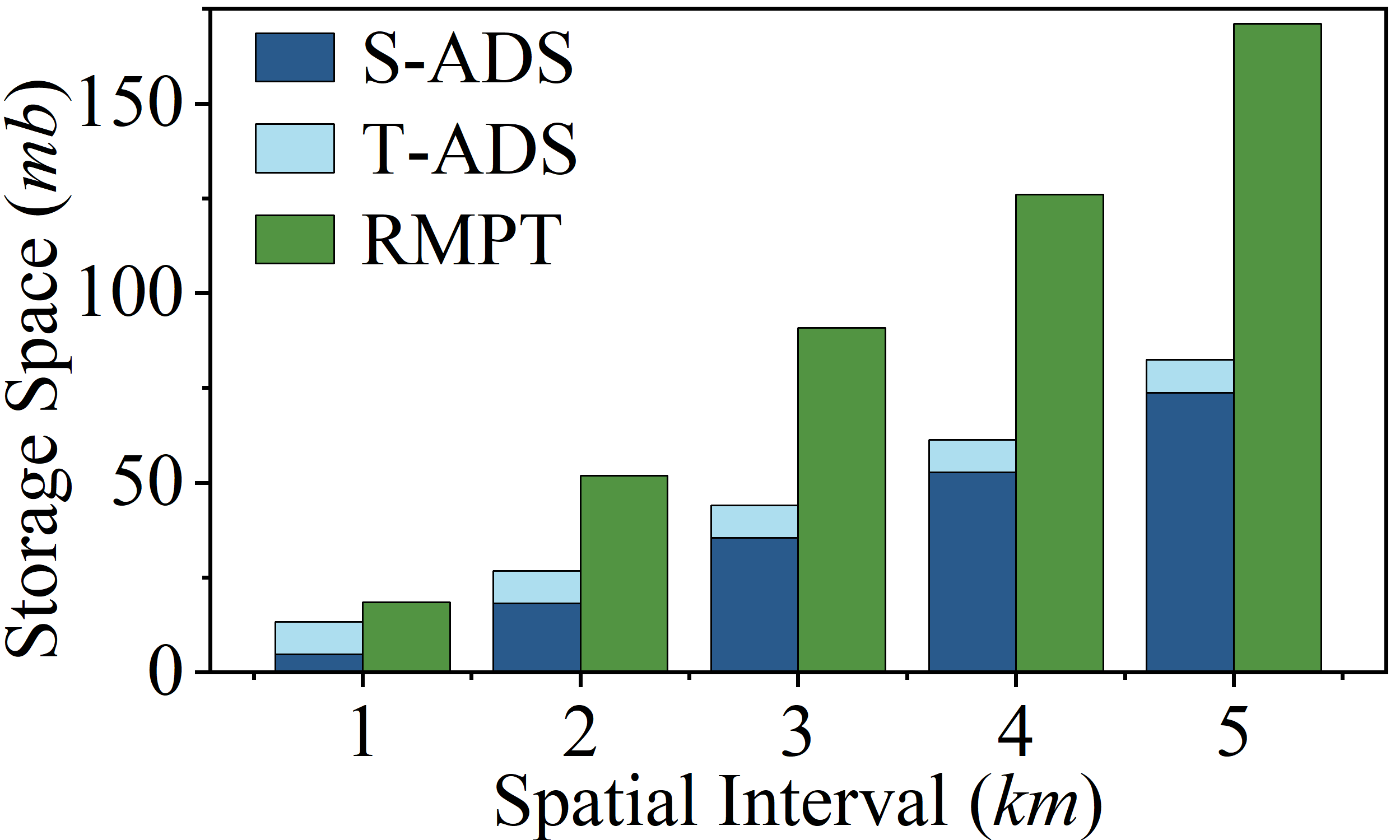}}
\hspace{0.0025\textwidth}
\subfigure[Chengdu]{\includegraphics[width=0.31\textwidth]{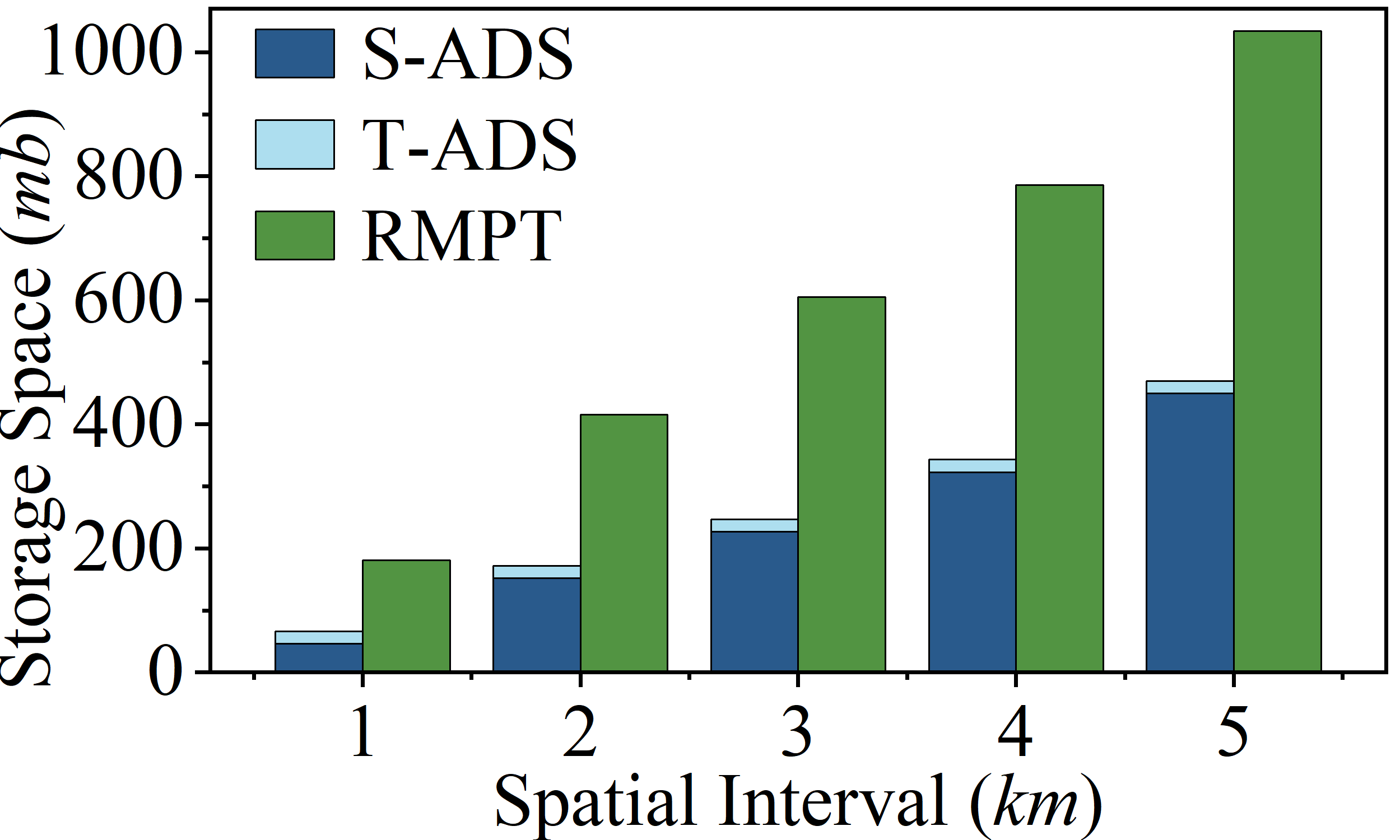}}
\hspace{0.0025\textwidth}
\subfigure[Geolife]{\includegraphics[width=0.31\textwidth]{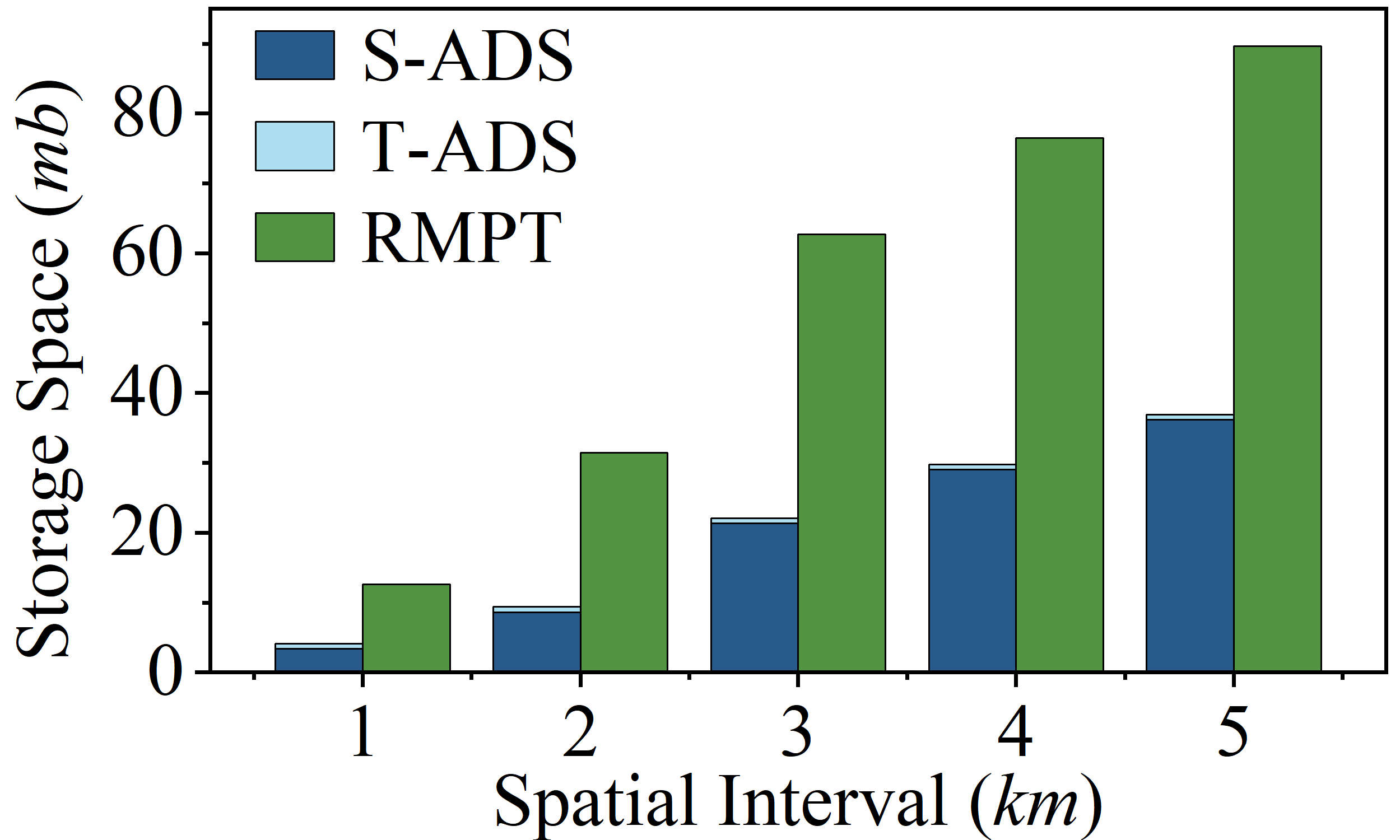}}
\caption{Impact on VO size of the spatial range.}
\label{VO size with spatial}
\end{figure*}

Fig.~\ref{verification with temporal} shows a breakdown of the verification time for \texttt{VTRQ} on the three datasets within 2 km while varying the temporal range. The verification process includes three parts: S-ADS verification, T-ADS verification, and result aggregation. S-ADS verification dominates the verification time. As the temporal verification range increases, the proportion of T-ADS verification increases, while the S-ADS proportion decreases. Although the temporal range expands, its cost is lower than those of the S-ADS verification and result aggregation. This is due to the higher complexity of the spatial verification: it involves not only hashing trajectory data on individual edges but also aggregating all edges in the relevant subgraph. Finally, the root-level aggregation is completed by traversing and combining the query path in the spatial tree. 
Fig.~\ref{verification with spatial} shows the breakdown of the verification time on the three datasets using a temporal range of 24 hours while varying the spatial range. S-ADS verification still dominates the verification time. 
As the spatial range increases, the proportion of the T-ADS verification time gradually decreases, while the proportion of the S-ADS verification time increases accordingly. 
This is because a larger spatial range results in the consideration of more spatial information, thereby increasing the verification overhead of S-ADS.
Notably, on the Xi'an dataset, when the spatial range is set to 1 kilometer, the proportion of the T-ADS verification is higher than that of S-ADS. This is because a smaller spatial range results in fewer trajectories falling in the query region, causing a low spatial verification overhead. In contrast, under the fixed 24-hour temporal range, the query result of T-ADS includes a relatively large set of trajectories, resulting in a higher cost.

Figs.~\ref{VO size with tempora} and ~\ref{VO size with spatial} report the VO sizes for \texttt{VTRQ} and RPMT. \texttt{VTRQ}’s VO consists of two parts: (i) the spatial VO, generated by querying S-ADS, and (ii) the temporal VO, generated by querying T-ADS.
In all cases, \texttt{VTRQ}’s total VO is smaller than for RPMT. This is because \texttt{VTRQ} separates spatial and temporal verification, each of which uses a lightweight ADS structure. RPMT handles both jointly, causing higher overhead.
In Fig.~\ref{VO size with tempora}, \texttt{VTRQ}’s temporal VO size increases linearly as the temporal query range. This is because the spatial range is fixed. The spatial VO size stays constant, while the temporal VO grows linearly as the time range increases.
Similarly, Fig.~\ref{VO size with spatial} shows a linear increase in the spatial VO size as the spatial query range  increases. With the temporal range fixed, temporal VO is constant. The spatial VO grows linearly with the spatial range.
Moreover, the trends in VO size under varying spatial and temporal query ranges align with the verification time results. This is because larger temporal or spatial ranges lead to larger VO sizes and longer verification times.

\subsection{ADS maintenance cost}
We evaluate the maintenance cost of S-ADS and T-ADS in \texttt{VTRQ} from four dimensions: construction time, storage cost, update time, and cost on blockchain.

\begin{table*}[t]
\centering
\caption{Comprehensive evaluation of ADS construction time (s).}
\setlength{\tabcolsep}{3pt}
\renewcommand{\arraystretch}{1.1}

\begin{tabular}{c c c c c c c c c c c c c}
\toprule
Dataset
& \multicolumn{4}{c}{Xi'an}
& \multicolumn{4}{c}{Chengdu}
& \multicolumn{4}{c}{Geolife} \\
\cmidrule(lr){2-5}\cmidrule(lr){6-9}\cmidrule(lr){10-13}
Scale
& S-ADS & T-ADS & VTRQ & RPMT
& S-ADS & T-ADS & VTRQ & RPMT
& S-ADS & T-ADS & VTRQ & RPMT \\
\midrule

\multirow{2}{*}{1}
& 31.701 & 2.9791 & \multirow{2}{*}{\textbf{34.680}} & \multirow{2}{*}{\underline{46.336}}
& 118.56 & 8.1801 & \multirow{2}{*}{\textbf{126.74}} & \multirow{2}{*}{\underline{185.49}}
& 380.55 & 23.659 & \multirow{2}{*}{\textbf{404.21}} & \multirow{2}{*}{\underline{867.55}} \\
& 91.4\% & 8.6\% &  & 
& 93.5\% & 6.5\% &  &
& 94.1\% & 5.9\% &  & \\

\midrule

\multirow{2}{*}{2}
& 70.496 & 6.2989 & \multirow{2}{*}{\textbf{76.795}} & \multirow{2}{*}{\underline{117.51}}
& 240.63 & 18.062 & \multirow{2}{*}{\textbf{258.69}} & \multirow{2}{*}{\underline{539.19}}
& 834.59 & 26.661 & \multirow{2}{*}{\textbf{861.25}} & \multirow{2}{*}{\underline{1496.5}} \\
& 91.8\% & 8.2\% &  & 
& 93.0\% & 7.0\% &  &
& 96.9\% & 3.1\% &  & \\

\midrule

\multirow{2}{*}{3}
& 152.66 & 14.856 & \multirow{2}{*}{\textbf{167.52}} & \multirow{2}{*}{\underline{350.81}}
& 599.56 & 47.833 & \multirow{2}{*}{\textbf{647.39}} & \multirow{2}{*}{\underline{3734.1}}
& 1169.8 & 27.658 & \multirow{2}{*}{\textbf{1197.46}} & \multirow{2}{*}{\underline{1955.9}} \\
& 91.1\% & 8.9\% &  & 
& 92.6\% & 7.4\% &  &
& 97.7\% & 2.3\% &  & \\

\midrule

\multirow{2}{*}{4}
& 203.54 & 20.730 & \multirow{2}{*}{\textbf{224.27}} & \multirow{2}{*}{\underline{550.08}}
& 1361.1 & 127.72 & \multirow{2}{*}{\textbf{1488.8}} & \multirow{2}{*}{\underline{10022}}
& 1705.6 & 27.832 & \multirow{2}{*}{\textbf{1733.45}} & \multirow{2}{*}{\underline{2944.3}} \\
& 90.8\% & 9.2\% &  & 
& 91.4\% & 8.6\% &  &
& 98.4\% & 1.6\% &  & \\

\midrule

\multirow{2}{*}{5}
& 304.76 & 32.880 & \multirow{2}{*}{\textbf{337.64}} & \multirow{2}{*}{\underline{1044.4}}
& 2446.6 & 292.57 & \multirow{2}{*}{\textbf{2739.2}} & \multirow{2}{*}{\underline{36350}}
& 2162.6 & 29.484 & \multirow{2}{*}{\textbf{2192.13}} & \multirow{2}{*}{\underline{3798.8}} \\
& 90.3\% & 9.7\% &  & 
& 89.3\% & 10.7\% &  &
& 98.7\% & 1.3\% &  & \\

\bottomrule
\end{tabular}

\label{tab:ADS construct}
\end{table*}

\begin{figure*}[t]
\centering
\subfigure[Xi'an]{\includegraphics[width=0.31\textwidth]{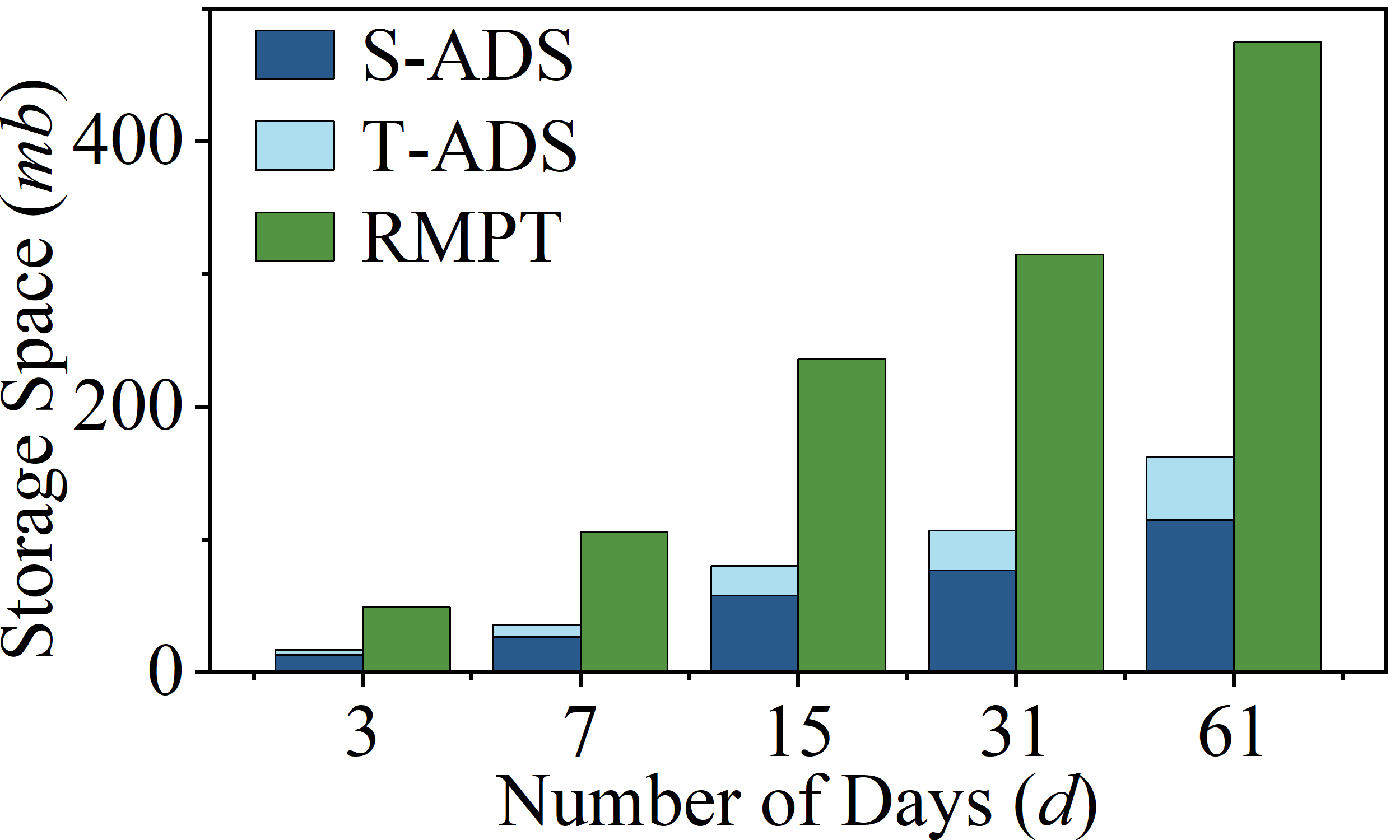}}
\hspace{0.0025\textwidth}
\subfigure[Chengdu]{\includegraphics[width=0.31\textwidth]{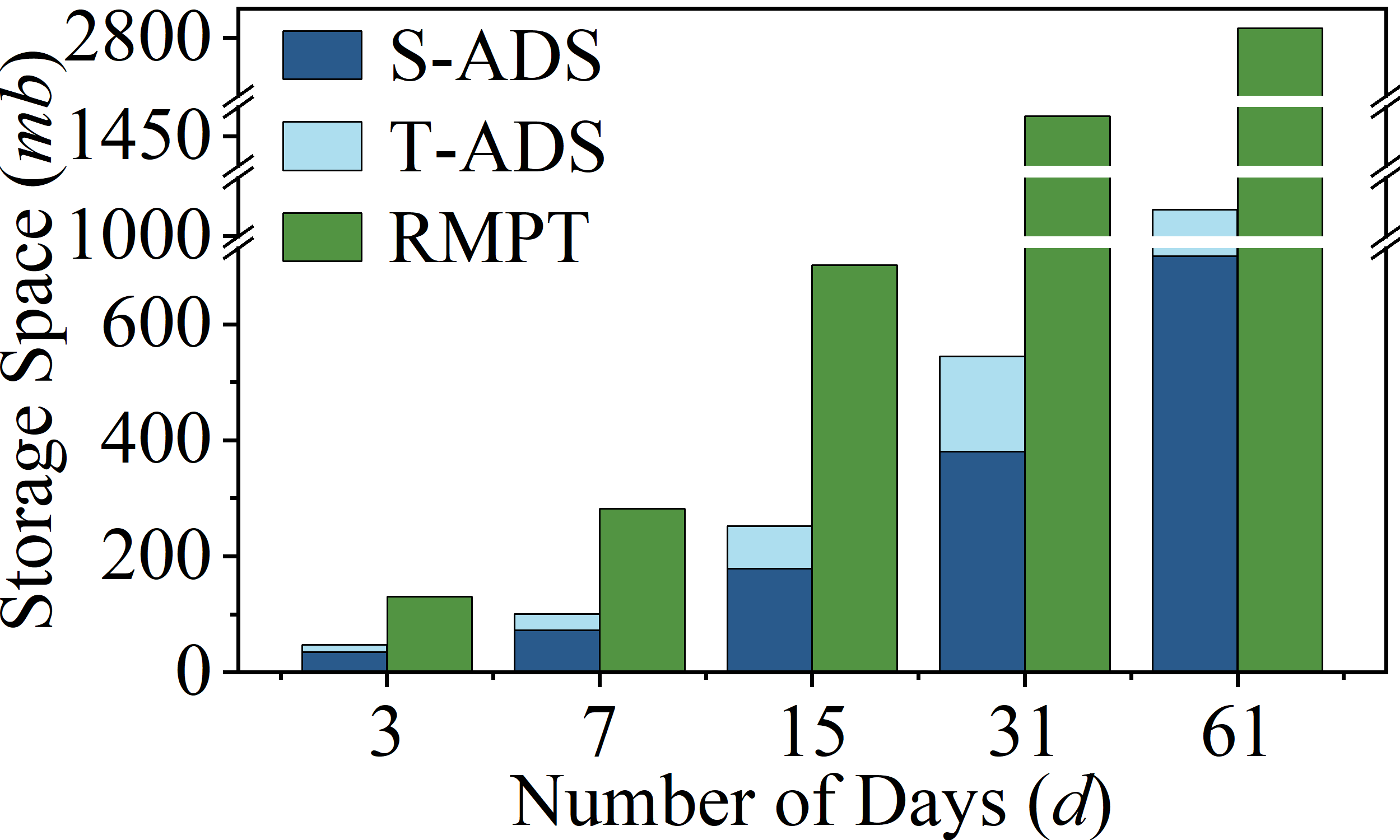}}
\hspace{0.0025\textwidth}
\subfigure[Geolife]
{\includegraphics[width=0.31\textwidth]{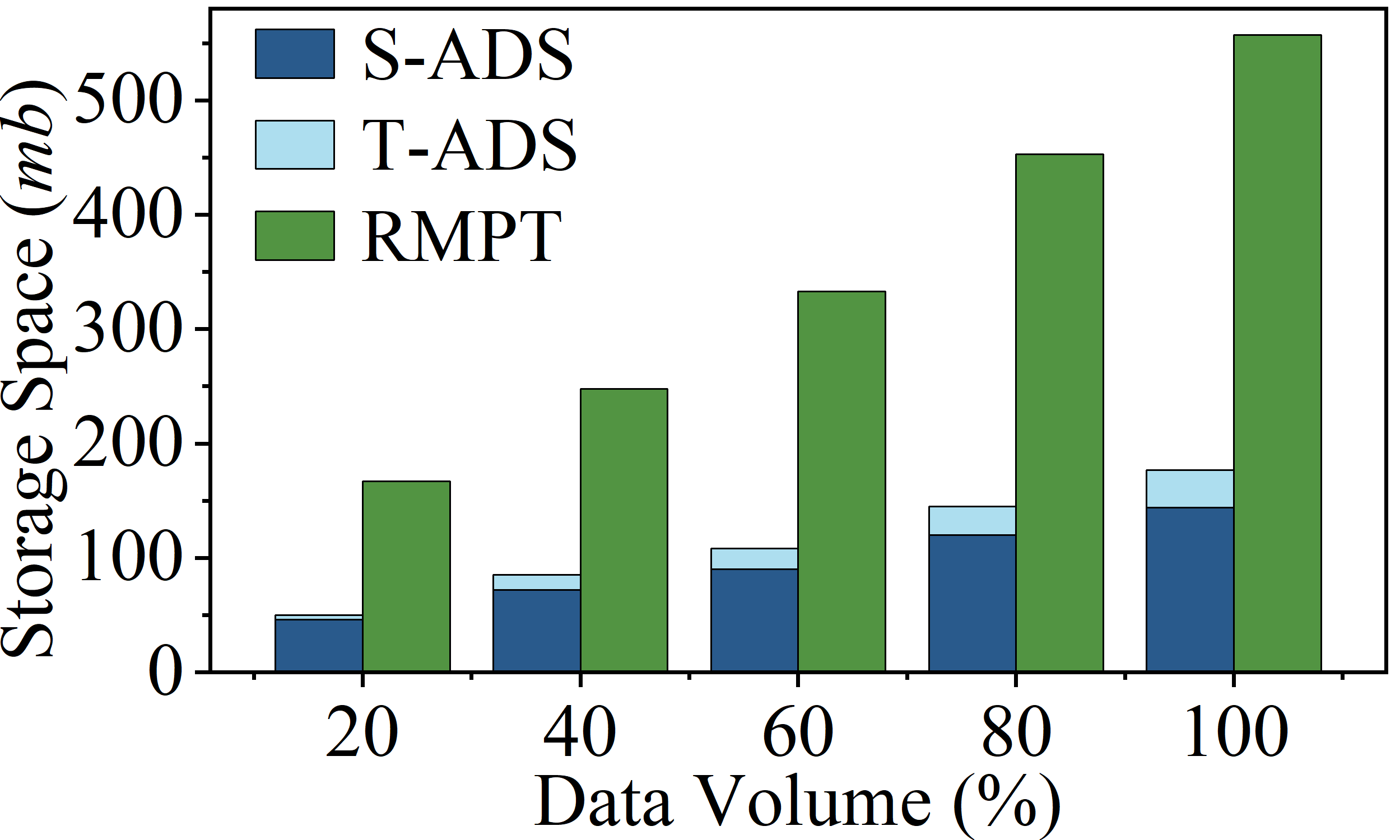}}
\caption{ADS storage cost.}
\label{ADS size}
\end{figure*}
\begin{figure*}[t]
\centering
\subfigure[Xi'an]{\includegraphics[width=0.31\textwidth]{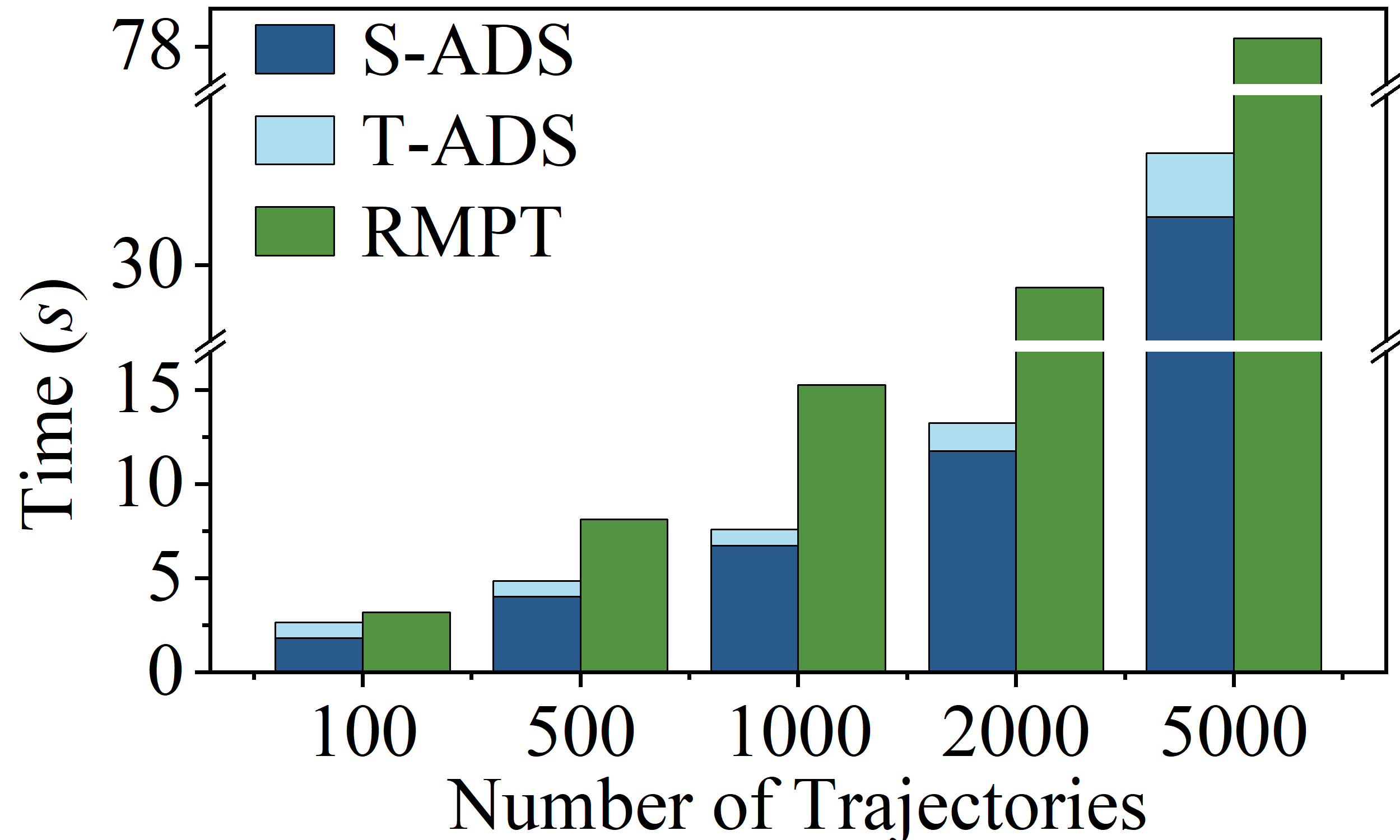}}
\hspace{0.0025\textwidth}
\subfigure[Chengdu]{\includegraphics[width=0.31\textwidth]{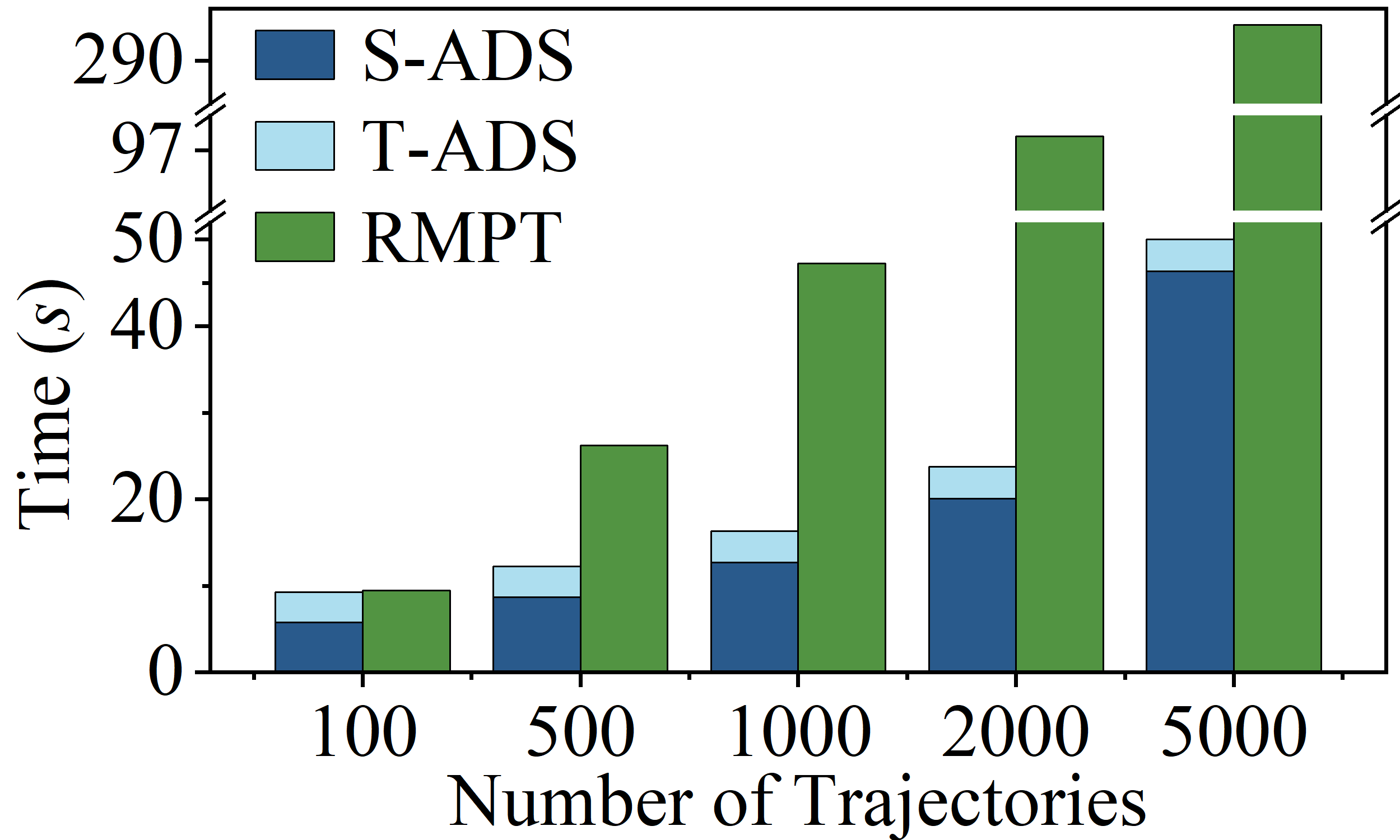}}
\hspace{0.0025\textwidth}
\subfigure[Geolife]{\includegraphics[width=0.31\textwidth]{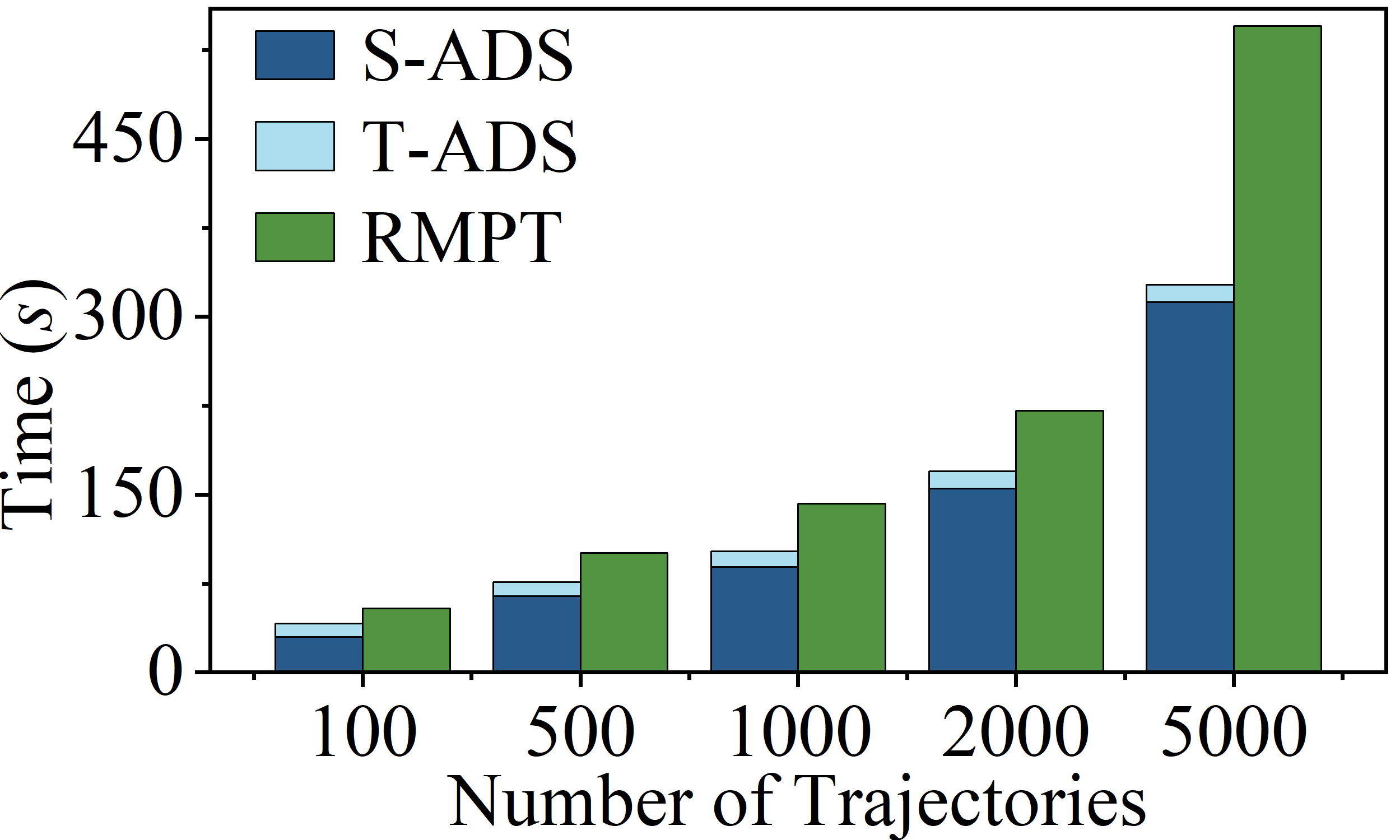}}
\caption{ADS update time.}
\label{ADS update}
\vspace{-1mm}
\end{figure*}

\subsubsection{ADS construction time}
Table~\ref{tab:ADS construct} shows ADS construction time at different dataset scales.
Here, Scale represents the data size, controlled by the number of days of trajectory data (3, 7, 15, 31, and 61 days) for the Xi’an and Chengdu datasets, and the percentage of the dataset (20\%, 40\%, 60\%, 80\%, and 100\%) for the Geolife dataset.
For our method, we also report the time proportion spent on S-ADS and T-ADS construction.
As shown in Table~\ref{tab:ADS construct}, \texttt{VTRQ} achieves approximately a 2× speedup over the baseline on the Chengdu dataset. This improvement comes from the decoupling of the spatial and temporal attributes in the trajectory data, which reduces redundant timestamp insertions.
The improvement is less pronounced on the Xi’an and Geolife datasets due to the smaller data volumes and lower optimization potentials.
In addition, 
both RPMT and \texttt{VTRQ} exhibit longer construction times on the Geolife dataset, due to its larger road network.


\subsubsection{ADS storage cost}

Fig.~\ref{ADS size} presents the storage cost of ADS construction on the three datasets. For the Xi’an and Chengdu datasets, data covering 3, 5, 7, 31, and 61 days are used, while for the Geolife dataset, 20\%, 40\%, 60\%, 80\%, and 100\% of the total data are used.
As shown in Fig.~\ref{ADS size}, \texttt{VTRQ} consistently requires less storage than the RPMT. This improvement stems from the lightweight indexing design, which separates spatial information from complex temporal data and eliminates redundant records of repeated trajectory segments.


\subsubsection{ADS update time}

Fig.~\ref{ADS update} shows the ADS update time when 100, 500, 1,000, 2,000, and 5,000 additional trajectories are inserted after construction. For the Chengdu dataset, ADSs are built with 30 days of data; for the Xi’an dataset, 15 days of data are used; and for the Geolife dataset, 40\% of the total data is used for the initial construction.
As shown in Fig.~\ref{ADS update}, our method outperforms the RPMT across all datasets. This is mainly due to several key optimizations:
In RPMT, trajectory insertion involves spatial matching and temporal information handling. This process is complex and incurs high computational cost.
In contrast, our S-ADS completely eliminates the dependency on temporal data, relying on the spatial information of trajectories for matching, which simplifies the insertion process. For the T-ADS, we retain only the start and end timestamps of each trajectory, avoiding redundant timestamp insertions and reducing the data size. 
Moreover, we replace the original sequential timestamp list with a balanced binary tree-based temporal ADS. This optimization improves insertion and query efficiency, thereby enhancing update performance.

{
\begin{table}[!t]
\renewcommand\tabcolsep{5pt}
\centering

\caption{On-chain cost (s).}

\begin{tabular}{cccccc}
   \toprule

\textbf{Number }&\textbf{2} &\textbf{4} &\textbf{8} &\textbf{16} &\textbf{32}\\

\midrule

512 bits
& 161.21 & 161.26 & 161.33 &161.37 & 161.43  \\

256 bits
& 161.06 & 161.17 & 161.20 & 161.27 & 161.34  \\

\bottomrule
\end{tabular}
\label{tab:on chain cost}

\end{table}}

\subsubsection{Cost on blockchain}


To evaluate the blockchain performance of our framework, we uploaded  two ADS root hashes to Tendermint. Each root hash is 256 bits, for a total of 512 bits. We also evaluate an alternative configuration in which the two hashes are concatenated and re-hashed using SHA256, yielding one 256-bit combined root hash. During verification, the same operation can be applied to compute the combined root hash, ensuring identical verification correctness with negligible time overhead while reducing on-chain storage.
As reported in Table~\ref{tab:on chain cost}, both configurations show nearly identical end-to-end delay. This is because the latency mainly comes from block generation and confirmation. Given the extremely small data sizes involved, the impact of data volume on blockchain upload performance is negligible.

\section{Conclusion}\label{sec:conclusion}

We proposed \texttt{VTRQ}, the first framework to support verifiable trajectory range queries in hybrid-storage blockchain systems. By designing two tailored ADSs—a S-ADS and a T-ADS—\texttt{VTRQ} enables efficient and trustworthy query processing over outsourced trajectory data. Furthermore, a spatio-temporal edge aggregation mechanism was introduced to ensure that only trajectory edges satisfying both spatial and temporal constraints are retained. Experimental results demonstrate that \texttt{VTRQ} significantly outperforms the existing method, achieving up to 6× improvement in query efficiency and up to an order-of-magnitude enhancement in verification performance. 
In future research, it is of interest to extend \texttt{VTRQ} to support higher-dimensional trajectories such as 3D trajectories for drone flights.



\balance
\bibliographystyle{unsrt}
\bibliography{sample-base}

\begin{thebibliography}{10}

\bibitem{wu2024fedapt}
Jia Wu, Tingyi Dai, Peiyuan Guan, Su~Liu, Fangfang Gou, Amir Taherkordi, Yushuai Li, and Tianyi Li.
\newblock {FedAPT}: Joint adaptive parameter freezing and resource allocation for communication-efficient federated vehicular networks.
\newblock {\em IoT-J}, 11(11):19520--19536, 2024.

\bibitem{li2022evolutionary}
Tianyi Li, Lu~Chen, Christian~S. Jensen, Torben~Bach Pedersen, Yunjun Gao, and Jilin Hu.
\newblock Evolutionary clustering of moving objects.
\newblock In {\em ICDE}, pages 2399--2411, 2022.

\bibitem{kong2022trajectory}
Fanhui Kong, Jianqiang Li, Bin Jiang, Huihui Wang, and Houbing Song.
\newblock Trajectory optimization for drone logistics delivery via attention-based pointer network.
\newblock {\em TITS}, 24(4):4519--4531, 2022.

\bibitem{shang2014personalized}
Shuo Shang, Ruogu Ding, Kai Zheng, Christian~S. Jensen, Panos Kalnis, and Xiaofang Zhou.
\newblock Personalized trajectory matching in spatial networks.
\newblock {\em VLDB Journal}, 23:449--468, 2014.

\bibitem{didiGlobalWebsite}
{DiDi Global}.
\newblock {DiDi} global official website.
\newblock https://www.didiglobal.com/, 2026.
\newblock Accessed January 18, 2026.

\bibitem{aliyunStartup1077761}
{Alibaba Cloud}.
\newblock Alibaba cloud startup program case.
\newblock https://startup.aliyun.com/info/1077761.html, 2026.
\newblock Accessed January 18, 2026.

\bibitem{googleCloudGeotabCaseStudy}
{Google Cloud}.
\newblock Geotab customer case study.
\newblock https://cloud.google.com/customers/geotab, 2026.
\newblock Accessed January 18, 2026.

\bibitem{awsBmwGroupCaseStudy}
{Amazon Web Services}.
\newblock {BMW} group case study.
\newblock https://aws.amazon.com/solutions/case-studies/bmw-group-case-study/, 2026.
\newblock Accessed January 18, 2026.

\bibitem{xia2022litmus}
Yu~Xia, Xiangyao Yu, Matthew Butrovich, Andrew Pavlo, and Srinivas Devadas.
\newblock {Litmus}: Towards a practical database management system with verifiable acid properties and transaction correctness.
\newblock In {\em SIGMOD}, pages 1478--1492, 2022.

\bibitem{zhou2021veridb}
Wenchao Zhou, Yifan Cai, Yanqing Peng, Sheng Wang, Ke~Ma, and Feifei Li.
\newblock {VeriDB}: An {SGX}-based verifiable database.
\newblock In {\em SIGMOD}, pages 2182--2194, 2021.

\bibitem{yu2021secure}
Xixun Yu, Yidan Hu, Rui Zhang, Zheng Yan, and Yanchao Zhang.
\newblock Secure outsourced top-k selection queries against untrusted cloud service providers.
\newblock In {\em IWQOS}, pages 1--10, 2021.

\bibitem{wu2023enabling}
Haotian Wu, Zecheng Li, Rui Song, and Bin Xiao.
\newblock Enabling privacy-preserving and efficient authenticated graph queries on blockchain-assisted clouds.
\newblock {\em TKDE}, 35(9):9728--9742, 2023.

\bibitem{zhang2019gem}
Ce~Zhang, Cheng Xu, Jianliang Xu, Yuzhe Tang, and Byron Choi.
\newblock {{GEM$^2$-tree}}: A gas-efficient structure for authenticated range queries in blockchain.
\newblock In {\em ICDE}, pages 842--853, 2019.

\bibitem{zhang2021authenticated}
Ce~Zhang, Cheng Xu, Haixin Wang, Jianliang Xu, and Byron Choi.
\newblock Authenticated keyword search in scalable hybrid-storage blockchains.
\newblock In {\em ICDE}, pages 996--1007, 2021.

\bibitem{li2024authenticated1}
Siyu Li, Zhiwei Zhang, Jiang Xiao, Meihui Zhang, Ye~Yuan, and Guoren Wang.
\newblock Authenticated keyword search on large-scale graphs in hybrid-storage blockchains.
\newblock In {\em ICDE}, pages 1958--1971, 2024.

\bibitem{li2024authenticated2}
Siyu Li, Zhiwei Zhang, Meihui Zhang, Ye~Yuan, and Guoren Wang.
\newblock Authenticated subgraph matching in hybrid-storage blockchains.
\newblock In {\em ICDE}, pages 1986--1998, 2024.

\bibitem{msspalertAwsVoterLeak}
Dan Kobialka.
\newblock {AWS} database leak exposed california voter data.
\newblock https://www.msspalert.com/news/aws-database-leak-california-voter-data-exposed-held-for-ransom, 2017.
\newblock Accessed January 18, 2026.

\bibitem{mykletun2006authentication}
Einar Mykletun, Maithili Narasimha, and Gene Tsudik.
\newblock Authentication and integrity in outsourced databases.
\newblock {\em TOS}, 2(2):107--138, 2006.

\bibitem{pang2005verifying}
HweeHwa Pang, Arpit Jain, Krithi Ramamritham, and Kian-Lee Tan.
\newblock Verifying completeness of relational query results in data publishing.
\newblock In {\em SIGMOD}, pages 407--418, 2005.

\bibitem{li2010authenticated}
Feifei Li, Marios Hadjieleftheriou, George Kollios, and Leonid Reyzin.
\newblock Authenticated index structures for aggregation queries.
\newblock {\em TISSEC}, 13(4):1--35, 2010.

\bibitem{li2020compression}
Tianyi Li, Ruikai Huang, Lu~Chen, Christian~S. Jensen, and Torben~Bach Pedersen.
\newblock Compression of uncertain trajectories in road networks.
\newblock {\em PVLDB}, 13(7):1050--1063, 2020.

\bibitem{li2021trace}
Tianyi Li, Lu~Chen, Christian~S. Jensen, and Torben~Bach Pedersen.
\newblock {TRACE}: Real-time compression of streaming trajectories in road networks.
\newblock {\em PVLDB}, 14(7):1175--1187, 2021.

\bibitem{hamdi2022spatiotemporal}
Ali Hamdi, Khaled Shaban, Abdelkarim Erradi, Amr Mohamed, Shakila~Khan Rumi, and Flora~D Salim.
\newblock Spatiotemporal data mining: a survey on challenges and open problems.
\newblock {\em AIR}, 55(2):1441--1488, 2022.

\bibitem{yu2023continuous}
Xiaofeng Yu, Shunzhi Zhu, and Yongjun Ren.
\newblock Continuous trajectory similarity search with result diversification.
\newblock {\em FGCS}, 143:392--400, 2023.

\bibitem{hu2023spatio}
Danlei Hu, Lu~Chen, Hanxi Fang, Ziquan Fang, Tianyi Li, and Yunjun Gao.
\newblock Spatio-temporal trajectory similarity measures: A comprehensive survey and quantitative study.
\newblock {\em TKDE}, 36(5):2191--2212, 2023.

\bibitem{xu2019vchain}
Cheng Xu, Ce~Zhang, and Jianliang Xu.
\newblock {vChain}: Enabling verifiable {B}oolean range queries over blockchain databases.
\newblock In {\em SIGMOD}, pages 141--158, 2019.

\bibitem{ruan2021lineagechain}
Pingcheng Ruan, Tien Tuan~Anh Dinh, Qian Lin, Meihui Zhang, Gang Chen, and Beng~Chin Ooi.
\newblock {LineageChain}: a fine-grained, secure and efficient data provenance system for blockchains.
\newblock {\em VLDB Journal}, 30:3--24, 2021.

\bibitem{zhu2019sebdb}
Yanchao Zhu, Zhao Zhang, Cheqing Jin, Aoying Zhou, and Ying Yan.
\newblock {SEBDB}: semantics empowered blockchain database.
\newblock In {\em ICDE}, pages 1820--1831, 2019.

\bibitem{pei2020efficient}
Qingqi Pei, Enyuan Zhou, Yang Xiao, Deyu Zhang, and Dongxiao Zhao.
\newblock An efficient query scheme for hybrid storage blockchains based on merkle semantic trie.
\newblock In {\em SRDS}, pages 51--60, 2020.

\bibitem{jia2025mest}
Jinping Jia, Yichen Gao, Yifei Zhen, Zhao Zhang, Kun Qian, and Cheqing Jin.
\newblock {MEST}: An efficient authenticated secondary index in blockchain systems.
\newblock In {\em ICDE}, pages 1636--1649, 2025.

\bibitem{cui2025towards}
Ningning Cui, Dong Wang, Jianxin Li, Huaijie Zhu, Xiaochun Yang, and Jianliang Xu.
\newblock Towards dynamic boolean range query over hybrid-storage blockchains: A secure and reliably verifiable framework.
\newblock In {\em ICDE}, pages 1664--1676, 2025.

\bibitem{sun20253fs}
Pengcheng Sun, Lan Zhang, Jiandong Liu, Chen Tang, and Jialiang Wang.
\newblock {E$^3$FS}: Efficient, secure, and verifiable fuzzy search with data updates in hybrid-storage blockchains.
\newblock In {\em ICDE}, pages 3890--3903, 2025.

\bibitem{pfoser2000novel}
Dieter Pfoser, Christian~S. Jensen, and Yannis Theodoridis.
\newblock Novel approaches to the indexing of moving object trajectories.
\newblock In {\em PVLDB}, pages 395--406, 2000.

\bibitem{tao2001mv3r}
Yufei Tao and Dimitris Papadias.
\newblock The {MV3R-Rree}: A spatio-temporal access method for timestamp and interval queries.
\newblock In {\em PVLDB}, pages 431--440, 2001.

\bibitem{chakka2003indexing}
V.~Prasad Chakka, Adam~C. Everspaugh, and Jignesh~M. Patel.
\newblock Indexing large trajectory data sets with {SETI}.
\newblock In {\em CIDR}, volume~75, page~76, 2003.

\bibitem{yang2017novel}
Xiaochun Yang, Bin Wang, Kai Yang, Chengfei Liu, and Baihua Zheng.
\newblock A novel representation and compression for queries on trajectories in road networks.
\newblock {\em TKDE}, 30(4):613--629, 2017.

\bibitem{bitcoin}
Satoshi Nakamoto.
\newblock Bitcoin: A peer-to-peer electronic cash system.
\newblock https://bitcoin.org/bitcoin.pdf, 2008.

\bibitem{wood2014ethereum}
Gavin Wood.
\newblock Ethereum: A secure decentralised generalised transaction ledger.
\newblock {\em Ethereum project yellow paper}, 151(2014):1--32, 2014.

\bibitem{wang2022vchain+}
Haixin Wang, Cheng Xu, Ce~Zhang, Jianliang Xu, Zhe Peng, and Jian Pei.
\newblock {vChain+}: Optimizing verifiable blockchain {B}oolean range queries.
\newblock In {\em ICDE}, pages 1927--1940, 2022.

\bibitem{singh2023efficient}
Bikash~Chandra Singh, Qingqing Ye, Haibo Hu, and Bin Xiao.
\newblock Efficient and lightweight indexing approach for multi-dimensional historical data in blockchain.
\newblock {\em FGCS}, 139:210--223, 2023.

\bibitem{hao2023efficient}
Kun Hao, Junchang Xin, Zhiqiong Wang, Zhongming Yao, and Guoren Wang.
\newblock Efficient and secure data sharing scheme on interoperable blockchain database.
\newblock {\em TBD}, 9(4):1171--1185, 2023.

\bibitem{yao2023efficient}
Zhongming Yao, Zhiqiong Wang, Liang Wen, Kun Hao, and Junming Xu.
\newblock Efficient blockchain data trusty provenance based on the w3c prov model.
\newblock In {\em ADMA}, pages 61--76, 2023.

\bibitem{zhou2023veridkg}
Enyuan Zhou, Song Guo, Zicong Hong, Christian~S. Jensen, Yang Xiao, Dalin Zhang, Jinwen Liang, and Qingqi Pei.
\newblock {VeriDKG}: a verifiable {SPARQL} query engine for decentralized knowledge graphs.
\newblock {\em PVLDB}, 17(4):912--925, 2023.

\bibitem{hao2022efficient}
Kun Hao, Junchang Xin, Zhiqiong Wang, Zhongming Yao, and Guoren Wang.
\newblock On efficient top-k transaction path query processing in blockchain database.
\newblock {\em DKE}, 141:102079, 2022.

\bibitem{yao2023learned}
Zhongming Yao, Junchang Xin, Kun Hao, Zhiqiong Wang, and Wancheng Zhu.
\newblock Learned-index-based semantic keyword query on blockchain.
\newblock {\em Mathematics}, 11(9):2055, 2023.

\bibitem{yao2024camel}
Yuanyuan Yao, Lu~Chen, Ziquan Fang, Yunjun Gao, Christian~S Jensen, and Tianyi Li.
\newblock Camel: Efficient compression of floating-point time series.
\newblock {\em SIGMOD}, 2(6):1--26, 2024.

\bibitem{hu2024estimator}
Danlei Hu, Ziquan Fang, Hanxi Fang, Tianyi Li, Chunhui Shen, Lu~Chen, and Yunjun Gao.
\newblock Estimator: An effective and scalable framework for transportation mode classification over trajectories.
\newblock {\em TITS}, 25(11):15562--15573, 2024.

\bibitem{lou2009map}
Yin Lou, Chengyang Zhang, Yu~Zheng, Xing Xie, Wei Wang, and Yan Huang.
\newblock Map-matching for low-sampling-rate {GPS} trajectories.
\newblock In {\em SIGSPATIAL}, pages 352--361, 2009.

\bibitem{ma2020forecasting}
Xiaolei Ma, Houyue Zhong, Yi~Li, Junyan Ma, Zhiyong Cui, and Yinhai Wang.
\newblock Forecasting transportation network speed using deep capsule networks with nested {LSTM} models.
\newblock {\em TITS}, 22(8):4813--4824, 2020.

\bibitem{crosby2019embedding}
Henry Crosby, Theodore Damoulas, and Stephen~A Jarvis.
\newblock Embedding road networks and travel time into distance metrics for urban modelling.
\newblock {\em IJGIS}, 33(3):512--536, 2019.

\bibitem{wang2022road}
Yong Wang, Kaiyu Li, Guoliang Li, and Nan Tang.
\newblock Road-aware indexing for trajectory range queries.
\newblock {\em TKDE}, 35(8):8476--8489, 2022.

\bibitem{xu2017range}
Jianqiu Xu, Hua Lu, and Ralf~Hartmut G{\"u}ting.
\newblock Range queries on multi-attribute trajectories.
\newblock {\em TKDE}, 30(6):1206--1211, 2017.

\bibitem{yin2022efficient}
Hongbo Yin, Hong Gao, Binghao Wang, Sirui Li, and Jianzhong Li.
\newblock Efficient trajectory compression and range query processing.
\newblock {\em WWW}, 25(3):1259--1285, 2022.

\bibitem{xu2023query}
Jianqiu Xu, Hua Lu, and Zhifeng Bao.
\newblock A query optimizer for range queries over multi-attribute trajectories.
\newblock {\em TIST}, 14(1):1--28, 2023.

\bibitem{preneel1994cryptographic}
Bart Preneel.
\newblock Cryptographic hash functions.
\newblock {\em ETT}, 5(4):431--448, 1994.

\bibitem{sobti2012cryptographic}
Rajeev Sobti and Ganesan Geetha.
\newblock Cryptographic hash functions: a review.
\newblock {\em IJCSI}, 9(2):461, 2012.

\bibitem{merkle1987digital}
Ralph~C Merkle.
\newblock A digital signature based on a conventional encryption function.
\newblock In {\em CRYPTO}, pages 369--378, 1987.

\bibitem{hendrickson1995multi}
Bruce Hendrickson and Robert~W Leland.
\newblock A multi-level algorithm for partitioning graphs.
\newblock {\em SC}, 95(28):1--14, 1995.

\bibitem{karypis1995analysis}
George Karypis and Vipin Kumar.
\newblock Analysis of multilevel graph partitioning.
\newblock In {\em SC}, pages 29--es, 1995.

\bibitem{karypis1998multilevel}
George Karypis and Vipin Kumar.
\newblock Multilevel algorithms for multi-constraint graph partitioning.
\newblock In {\em SC}, pages 28--28, 1998.

\bibitem{yao2024tsec}
Yuanyuan Yao, Hailiang Jie, Lu~Chen, Tianyi Li, Yunjun Gao, and Shiting Wen.
\newblock Tsec: An efficient and effective framework for time series classification.
\newblock In {\em ICDE}, pages 1394--1406. IEEE, 2024.

\bibitem{finis2015indexing}
Jan Finis, Robert Brunel, Alfons Kemper, Thomas Neumann, Norman May, and Franz Faerber.
\newblock Indexing highly dynamic hierarchical data.
\newblock {\em PVLDB}, 8(10):986--997, 2015.

\bibitem{10.65109/GPHO5000}
Wentao Xu, Zhongming Yao, Weihao Li, Zhenghang Song, Yumeng Song, Tianyi Li, and Yushuai Li.
\newblock Tcrl: Temporal-coupled adversarial training for robust constrained reinforcement learning in worst-case scenarios.
\newblock In {\em Proceedings of the 25th International Conference on Autonomous Agents and Multiagent Systems}, page 3489–3491, 2026.

\bibitem{yao2025vgq}
Zhongming Yao, Tianyi Li, Junchang Xin, Yushuai Li, Chenxu Wang, Zhiqiong Wang, Divesh Srivastava, and Christian~S. Jensen.
\newblock {VGQ}: Enabling verifiable graph queries on blockchain systems.
\newblock In {\em ICDE}, pages 3602--3614, 2025.

\bibitem{tendermintWebsite}
{Tendermint}.
\newblock Tendermint official website.
\newblock https://tendermint.com/, 2026.
\newblock Accessed January 21, 2026.

\bibitem{zheng2011geolifeUserGuide}
Yu~Zheng, Hao Fu, Xing Xie, Wei-Ying Ma, and Quannan Li.
\newblock Geolife {GPS} trajectory dataset.
\newblock https://www.microsoft.com/en-us/research/publication/geolife-gps-trajectory-dataset-user-guide/, July 2011.
\newblock Geolife {GPS} trajectories 1.1 edition. Accessed January 21, 2026.

\bibitem{mao2025dutytte}
Xiaowei Mao, Yan Lin, Shengnan Guo, Yubin Chen, Xingyu Xian, Haomin Wen, Qisen Xu, Youfang Lin, and Huaiyu Wan.
\newblock {DutyTTE}: Deciphering uncertainty in origin-destination travel time estimation.
\newblock In {\em AAAI}, volume~39, pages 12390--12398, 2025.

\bibitem{yung2012authentication}
Duncan Yung, Eric Lo, and Man~Lung Yiu.
\newblock Authentication of moving range queries.
\newblock In {\em CIKM}, pages 1372--1381, 2012.

\end{thebibliography}


\end{document}